\documentclass[11pt]{article}
\usepackage{times}
\usepackage{graphics,graphicx}
\usepackage[dvips]{epsfig}
\usepackage{eepic}
\usepackage{amsmath,amssymb,amsfonts}

\def\inline#1:{\par\vskip 7pt\noindent{\bf #1:}\hskip 10pt}

\long\def\commabs #1\commabsend{}
\long\def\commful #1\commfulend{#1}
\long\def\omitit #1 \omitend{}
\long\def\commcomm #1\commend{}

\newenvironment{smallitemize} {
  \begin{list}{$\bullet$} {\setlength{\parsep}{0pt}
\setlength{\itemsep}{0pt}} } { \end{list} }
\newcommand{\qed}{\hfill $\Box$ \medbreak}
\newenvironment{proof}{\noindent {\bf Proof.}}{\qed}
\def\level{j}

\newcommand{\vacant}[0]{\texttt{Vacant}}
\newcommand{\rchild}[0]{\texttt{Red\_Child}}
\newcommand{\cred}[0]{\texttt{Contains\_Red}}

\newcommand{\TREECOARSEN}[0]{\texttt{TREE\_COARSEN}}

\newcommand{\PARTLEADER}[0]{\texttt{PART\_LEADER}}

\newcommand{\SENDANCESTORINFO}[0]{\texttt{Send\_Anc\_Info}}
\newcommand{\ANCESTORINFO}[0]{\texttt{Anc\_Info}}
\newcommand{\ALG}[0]{\texttt{SYNC\_MST}}

\newcommand{\COUNTSIZE}[0]{\texttt{Count\_Size}}
\newcommand{\FINDOUT}[0]{\texttt{Find\_Min\_Out\_Edge}}
\newcommand{\PARENTID}[0]{\texttt{Parent\_ID}}
\newcommand{\ECHO}[0]{\texttt{ECHO}}
\def\ID{\mbox{\tt ID}}
\def\SP{\mbox{\tt SP}}
\def\NK{\mbox{\tt NumK}}
\def\EDIAM{\mbox{\tt EDIAM}}
\newcommand{\Ch}[0]{\texttt{Child}}

\newcommand{\OR}[0]{\texttt{OR}}
\newcommand{\False}[0]{\texttt{false}}
\newcommand{\True}[0]{\texttt{true}}

\newtheorem{theorem}{Theorem}[section]
\newtheorem{definition}{Definition}[section]
\newtheorem{Comment}{Comment}[section]
\newtheorem{comment}{Comment}[section]

\newtheorem{claim}[theorem]{Claim}
\newtheorem{lemma}[theorem]{Lemma}
\newtheorem{corollary}[theorem]{Corollary}

\newtheorem{observation}[theorem]{Observation}

\def\inline#1:{\par\vskip 7pt\noindent{\bf #1:}\hskip 10pt}
\def\cH{{\cal H}}
\def\cI{{\cal I}}
\def\cE{{\cal E}}

\def\cV{{\cal V}}
\def\cD{{\cal D}}

\def\cM{{\cal M}}

\def\cF{{\cal F}}

\def\cH{{\cal H}}
\def\cP{{\cal P}}

\def\blackslug{\hbox{\hskip 1pt \vrule width 4pt height 8pt
     height 1.5pt \hskip 1pt}}
\def\QED{\quad\blackslug\lower 8.5pt\null\par}

\def\MathN{\hbox{\rm I\kern-2pt I\kern-3.1pt N}}
\def\MathNsub{\hbox{\scriptsize I\kern-2pt I\kern-3.1pt N}}

\long\def\comment #1\commentend{} \long\def\commabs #1\commabsend{}

\newcommand{\ROOTS}[0]{\texttt{Roots}}

\newcommand{\PARENTS}[0]{\texttt{Parents}}
\newcommand{\EndP}[0]{\texttt{EndP}}
\newcommand{\AND}[0]{\texttt{Or\_EndP}}
\newcommand{\DFS}[0]{\texttt{dfs}}
\newcommand{\Info}[0]{\texttt{I}}
\newcommand{\hInfo}[0]{\hat{\texttt{I}}}
\newcommand{\piece}[0]{\texttt{Pc}}
\newcommand{\id}[0]{\texttt{ID}}
\newcommand{\Top}[0]{\texttt{Top}}
\newcommand{\Roo}[0]{\texttt{Root}}
\newcommand{\Bottom}[0]{\texttt{Bottom}}
\newcommand{\Show}[0]{\texttt{Show}}
\newcommand{\Keep}[0]{\texttt{Want}}
\newcommand{\Ask}[0]{\texttt{Ask}}
\newcommand{\WAVE}[0]{\texttt{Wave}}
\newcommand{\MULTI}[0]{\texttt{Multi\_Wave}}
\newcommand{\READY}[0]{\texttt{Ready}}
\newcommand{\FREE}[0]{\texttt{Free}}

\begin{document}

\title{Fast and Compact Self-Stabilizing  Verification, Computation, \\and Fault Detection of an MST}
\author{
Amos Korman
\thanks{CNRS and Univ. Paris Diderot, Paris, 75013, France.  E-mail: {\tt Amos.Korman@liafa.jussieu.fr}.
Supported in part by a  France-Israel cooperation grant
(``Mutli-Computing'' project)
from the France Ministry of Science and Israel Ministry of Science,
by the ANR projects ALADDIN and PROSE, and by the INRIA project GANG.}
\and Shay Kutten
\thanks{Information Systems Group, Faculty of IE\&M, The Technion,
Haifa, 32000 Israel. E-mail: {\tt  my last name at ie dot technion dot ac dot il}.
Supported in part by a  France-Israel cooperation grant
(``Mutli-Computing'' project)
from the France Ministry of Science and Israel Ministry of Science,  by
a grant from the Israel Science Foundation, and by the Technion funds for security research.}
\and
Toshimitsu Masuzawa
\thanks{Graduate School of Information Science and Technology, Osaka University, 1-5 Yamadaoka, Suita, Osaka, 565-0871, Japan.
 E-mail: {\tt masuzawa@ist.osaka-u.ac.jp}.
Supported in part by JSPS Grant-in-Aid for Scientific Research ((B)26280022).}
 }

\date{}

\begin{titlepage}
\def\thepage{}
\maketitle

\date{}

\def\thefootnote{\fnsymbol{footnote}}

\begin{abstract}


This paper demonstrates the usefulness of distributed local verification of proofs, as a tool for the design of self-stabilizing algorithms.
In particular, it introduces a somewhat generalized notion of distributed local proofs, and utilizes it for
improving the time complexity significantly, while maintaining space optimality.
As a result, we show
that optimizing the memory size
carries at most a small cost in terms of time, in the context of Minimum Spanning Tree (MST).
That is, we present algorithms that are both time and space efficient  for both constructing an MST and for verifying it.
This involves several parts that may be considered contributions in themselves.

First, we
generalize the notion of local proofs, trading off  the time complexity
for memory efficiency.
This adds a dimension to the study of distributed local proofs, which has been gaining attention recently.
Specifically, we design a (self-stabilizing) proof labeling scheme which is memory optimal  (i.e., $O(\log n)$ bits per node), and whose time complexity
is $O(\log ^2 n)$ in synchronous networks, or  $O(\Delta \log ^3 n)$ time  in asynchronous ones, where $\Delta$ is the maximum degree of nodes.
This answers an open problem posed by Awerbuch and Varghese (FOCS 1991).
We~also show that $\Omega(\log n)$ time is necessary, even in synchronous networks.
Another property is that
 if $f$ faults occurred,  then, within  the required
detection  time above, they are detected by some node in the $O(f\log n)$
locality of each of the faults.

Second, we show how to enhance a known transformer that makes
 input/output algorithms self-stabilizing. It now takes as input an efficient construction algorithm and an efficient self-stabilizing proof labeling scheme, and produces
 an efficient self-stabilizing algorithm. When used for MST, the transformer produces a memory optimal  self-stabilizing algorithm, whose time complexity, namely, $O(n)$,
 is significantly better even than that of previous algorithms.
(The time complexity of previous  MST algorithms that used $\Omega(\log^2 n)$ memory bits per node  was $O(n^2)$, and the time for optimal space  algorithms was $O(n|E|)$.)
Inherited from our proof labelling scheme, our self-stabilising MST construction algorithm also
has the following two properties: (1) if faults occur after the construction ended, then they are  detected by
some nodes within $O(\log ^2 n)$  time
in synchronous networks, or within $O(\Delta \log ^3 n)$ time  in
asynchronous ones, and (2) if $f$ faults occurred,  then, within  the required
detection  time above, they are detected within the $O(f\log n)$
locality of each of the faults.
We also show how to improve the above two properties, at the expense of some increase in the memory.

\end{abstract}



\bigskip
\noindent {\bf Keywords:} Distributed network algorithms, Locality, Proof labels, Minimum spanning tree,
Distributed property verification, Self-stabilization, Fast fault detection,
Local fault detection.

\end{titlepage}
\pagenumbering{arabic}

\section{Introduction}
\label{sec:intro}

\subsection{Motivation}
In a non-distributed context, solving a problem is believed to be, sometimes, much harder than verifying it (e.g., for NP-Hard problems).
Given a graph $G$ and a subgraph $H$ of $G$, a task introduced by Tarjan \cite{Tarjan79} is to check whether $H$ is a Minimum Spanning Tree (MST) of $G$.
This non-distributed verification seems to be just slightly easier than the non-distributed computation of an MST.
In the distributed context, the given subgraph $H$ is assumed to be represented distributively, such that each node stores pointers to (some of) its incident edges in $H$. The verification task consists of checking whether the collection of pointed edges indeed forms an MST, and if not, then it is required that at least one node {\em raises an alarm}. It was shown recently that such
 an MST verification task requires the same amount of time as the MST computation~\cite{DHKKNPPW, KKP11}.
On the other hand, assuming that each node can store some information, i.e., {\em a  label}, that can be used for the verification, the time complexity of an MST verification can be as small as 1, when using labels of size  $\Theta(\log^2 n)$ bits per node \cite{KormanKutten07,KKP10}, where $n$ denotes the number of nodes.
To make such a proof labeling scheme a useful algorithmic tool, one needs to present a {\em marker} algorithm for computing those labels.
One of the contributions of the current paper is a time and memory  efficient   marker algorithm.

Every decidable graph property (not just an MST) can be verified in a short time given large enough labels~\cite{KKP10}.
A second contribution of this paper is a generalization of such schemes to allow a reduction in the memory requirements, by trading off the locality (or the time).
In the context of MST, yet another (third) contribution is  a reduced space proof labeling scheme for MST. It uses
just $O(\log n)$ bits of memory per node (asymptotically the same as the amount of bits needed for merely representing distributively the MST). This is below the lower bound
of $\Omega(\log ^2 n )$ of \cite{KormanKutten07}. The reason this is possible is that the verification time is increased to $O(\log^2n)$ in synchronous networks and
to $O(\Delta \log^3 n)$ in asynchronous ones, where $\Delta$ is the maximum degree of nodes.
Another important property of the new scheme is that any fault is detected rather close to the node where it occurred.
Interestingly, it turns out that a logarithmic time penalty for verification is unavoidable. That is, we show that $\Omega(\log n)$ time for an MST verification scheme is necessary if the memory size is
restricted to $O(\log n)$ bits, even in synchronous networks.
(This, by the way, means that a verification with $O(\log n)$ bits, cannot be silent, in the sense
of \cite{silent}; this is why they could not be of the kind introduced in
\cite{KKP10}).

Given a long enough time, one can verify $T$ by recomputing the MST.
 An open problem posed by Awerbuch and Varghese \cite{awerbuch-varghese} is to find a synchronous MST verification algorithm whose  time complexity is
 smaller than the MST computation time, yet with a small memory.
 This problem was introduced in  \cite{awerbuch-varghese} in the context of self-stabilization, where
 the verification algorithm is combined with a non-stabilizing {\em construction} protocol to produce a stabilizing protocol. Essentially, for such purposes, the verification algorithm repeatedly checks the output of the non-stabilizing construction protocol, and runs the construction algorithm again if a fault at some node is detected. Hence, the construction algorithm and the corresponding verification algorithm are assumed to be designed together. This, in turn, may significantly simplify the checking process, since the construction algorithm may produce output variables (labels) on which the verification algorithm can later rely. In this context, the above mentioned third contribution solves this open problem by showing an $O(\log^2 n)$ time penalty
(in synchronous networks) when using optimal $O(\log n)$ memory size for the MST verification algorithm.
 In contrast, if we study MST {\em construction} instead of MST verification,
time lower bounds which are polynomial in $n$ for MST construction follow from \cite{lotker-patt-shamir-peleg,mst-lower-bound} (even for constant diameter graphs).

One known application of some
methods of distributed verification is for general transformers that transform non-self-stabilizing algorithms to self-stabilizing ones.
The fourth contribution of this paper is an adaptation of the transformer of \cite{awerbuch-varghese}
such that it can transform
algorithms in our context.
 That is, while the transformer of
 \cite{awerbuch-varghese}
  requires that the size of the network and its diameter are known, the adapted one
 allows networks of unknown size and diameter. Also, here, the verification method is
 a proof labeling scheme whose verifier part is self-stabilizing.
Based on the strength of the original transformer of \cite{awerbuch-varghese} (and that of the companion paper
\cite{APV} it uses), our adaptation yields a  result that is rather useful even without
plugging in the new verification scheme.
This is demonstrated by plugging in the proof labeling schemes of \cite{KormanKutten07,KKP10}, yielding an algorithm which already improves
the time of previous
$O(\log^2 n)$ memory self-stabilizing MST construction algorithm \cite{sebastian-disc-2010}, and also
detects faults using 1 time and at distance  at most $f$ from each fault (if $f$ faults occurred).

Finally, we obtain  an optimal $O(\log n)$ memory size, $O(n)$ time  asynchronous self-stabilizing MST construction algorithm. The state of the art time bound for such optimal memory algorithms was $O(n |E|)$  \cite{sebastian-disc2009, higham}. In fact, our time bound
 improves significantly even the best time bound for algorithms using polylogarithmic memory, which was $O(n^2)$  \cite{sebastian-disc-2010}.

Moreover, our
  self-stabilizing MST algorithm inherits two important properties from our verification scheme, which are: (1) the time it takes to detect faults is small:
 $O(\log^2 n)$ time in a synchronous network, or  $O(\Delta \log^3 n)$ in an asynchronous one; and (2)
if some $f$ faults occur, then each fault is detected within its $O(f\log n)$ neighbourhood.
Intuitively, a short detection distance and a small detection time may be helpful for the
design of local correction, for fault confinement, and for fault containment algorithms
 \cite{containment,AzarKP}.  Those notions were introduced to combat the phenomena of faults ``spreading'' and ``contaminating'' non-faulty nodes. For example,
 the infamous crash of the ARPANET (the predecessor of the Internet) was caused by a fault in a single node.  This caused old updates to be adopted by other nodes, who then generated wrong updates affecting others \cite{arpanet}.
 This is an example of those non-faulty nodes being contaminated.
 The requirement of {\em containment} \cite{containment} is that such a contamination does not occur, or, at least, that it is contained in a small radius around the faults. The requirement
 of {\em confinement} \cite{AzarKP} allows the contamination of a state of a node, as long as this contamination is not reflected in the output (or the externally visible actions) of the non-faulty nodes.
 Intuitively, if the detection distance is short, non-faulty nodes can detect the faults and avoid being contaminated.

\subsection{Related work}
The distributed construction of an MST
 has yielded techniques and insights that were used in the study of many other problems of distributed network protocols. It has also become a standard to check a new paradigm in distributed algorithms theory. The first distributed algorithm was proposed by \cite{first-mst},  its complexity was not analyzed. The seminal paper of Gallager, Humblet, and Spira  presented a message optimal algorithm
 that  used $O(n \log n)$ time, improved by Awerbuch to $O(n)$ time \cite{ghs,awerbuch-mst}, and later improved in  \cite{kdom,kutten-peleg-garay}
  to $O(D + \sqrt{n} \log^* n)$, where $D$ is the diameter of the network. This was coupled  with an almost
   matching lower bound of $\Omega(D+\sqrt{n})$ \cite{mst-lower-bound}.

Proof labeling schemes were introduced in  \cite{KKP10}. The model described therein assumes that the verification is restricted to 1 unit of time. In particular, a 1 time MST verification scheme was described
there using $O(\log^2 n)$ bits per node. This  was shown to be optimal in \cite{KormanKutten07}.  In \cite{GS11}, G\"o\"os and Suomela extend the notion of proof labeling schemes by allowing constant time verification, and
exhibit some efficient proof labeling schemes for recognizing several natural graph families.
In all these schemes, the criterion to decide failure of a proof (that is, the detection of a fault) is the case that at least one node does not manage to verify (that is, detects a fault).
The global state passes a proof successfully if all the nodes verify successfully.
This criterion for detection (or for a failure to prove) was suggested by \cite{AKY1, AKY} in the contexts of self stabilization, and used in self stabilization
(e.g.  \cite{APV, apvd,awerbuch-varghese}) as well as in other other contexts \cite{NS93}.

Self-stabilization \cite{dijkstra} deals with
algorithms that must cope with faults that are rather severe, though
of a type that does occur in reality
\cite{varghese-ss-possible}.
The faults may
cause the states of different nodes to be inconsistent with each
other. For example, the collection of marked edges may not be an MST.

Table \ref{tab:alg-compare} summarizes the known complexity results for self stabilizing MST construction algorithms. The first several entrees show the results of using
(to generate an MST algorithm automatically)
the known transformer of Katz and Perry \cite{katz-perry}, that extends automatically non self stabilizing algorithms to become self stabilizing.
The transformer of Katz and Perry \cite{katz-perry}  assumes a leader whose memory must hold a snapshot of the whole network.
 The time of the resulting self-stabilizing MST algorithm is $O(n)$  and the memory size is $O(|E|n)$.
 We have attributed a higher time to \cite{katz-perry} in the table, since we wanted to remove its assumption of a known leader, to make a fair comparison to the later papers who do not rely on this assumption.

 To remove the assumption, in the first entry we assumed the usage of the only leader election known at the time of \cite{katz-perry}.
That is, in \cite{AKY}, the first self-stabilizing leader election algorithm was proposed in order to remove the assumptions
of~\cite{katz-perry} that a leader and a spanning tree are given. The combination of \cite{AKY} and \cite{katz-perry}
 implied a self-stabilizing MST
in $O(n^2)$ time.
(Independently, a leader election algorithm was also presented by
 \cite{arora-gouda}; however, we cannot use it here since it
needed an extra assumption that a bound on $n$ was known; also, its higher time complexity would have driven the complexity of the transformed MST algorithm higher than the $O(n^2)$ stated above.)

Using unbounded space, the time of self-stabilizing leader election was later improved even to $O(D)$  (the actual diameter) \cite{sonu,datta2010}.
The  bounded memory algorithms of \cite{awerbuch-kutten-patt-et-al} or \cite{afek-bremler,datta2008},
 together with \cite{katz-perry} and \cite{awerbuch-mst}, yield a self-stabilizing MST algorithm using   $O(n|E|\log n)$ bits per node and time $O(D \log n)$ or $O(n)$.


\begin{table}\label{tab:alg-compare}
\caption{\texttt{Comparing self-stabilizing MST construction algorithms}}
\centering
\begin{tabular}{|p{3.5cm}|c|c|c|p{5cm}|}
 \hline
 \hline
 \texttt{Algorithm}  &  space  & time & Asynch & comment   \\
 \hline
      &  &  &  &  \\
     $\rm \cite{katz-perry} \rm $+\cite{AKY}+\cite{awerbuch-mst}& $O(|E|n)$ & $O(n^2)$& yes  &  \\
      \hline
      &  &  &  &  \\
     $\rm \cite{katz-perry} \rm $+ \cite{sonu}+\cite{awerbuch-mst} & unbounded & $O(D)$ & yes & The 2nd component can be replaced by \cite{datta2010}, assuming the $\cal{LOCAL}$ model. \\
    \hline       &  &  &  &  \\
  $\rm \cite{katz-perry}$+   $\rm \cite{awerbuch-mst} \rm $+
         $\rm \cite{awerbuch-kutten-patt-et-al}\rm $  & $O(|E|n)\log n$ &
          $O(\min\{D\log n , n\})$& yes& The third component here can be replaced by
          $\rm \cite{afek-bremler}\rm $ or by $\rm \cite{datta2008} \rm $.  \\
                 \hline    &  &  &  &  \\
     $\rm \cite{gupta-srimani-see-sebastian}\rm $ & $O(n\log n)$ & $O(n)$ & no & Implies an $O(n^2)$ time bound in asynchronous networks, assuming a good bound on the network size is known. The time is based on assuming the $\cal{LOCAL}$ model. \\
          \hline                &  &  &  &  \\
     $\rm \cite{higham}\rm $ & $O(\log n)$ & $\Omega(|E|n)$ &  yes & The time complexity is based on the assumption that a good bound on the network diameter is known.  \\
                \hline     &  &  &  &  \\
     $\rm \cite{sebastian-disc2009} \rm $ & $O(\log n)$ & $\Omega(|E|n)$ & yes & Aims to exchanging less bits with neighbours than $\rm \cite{higham}\rm $. Assumes a leader is known.    \\
      \hline               &  &  &  &  \\
    $\rm \cite{sebastian-disc-2010} \rm $  & $O(\log^2(n))$ & $O(n^2)$ &  yes &   \\
            \hline      \hline             &  &  &  &  \\
     Current paper & $O(\log n)$ & $O(n)$ &  yes &   \\
                  \hline     \hline
\end{tabular}

\end{table}

Antonoiu and Srimani  \cite{srimani-shared} presented a self stabilizing algorithm whose complexities were not analyzed. As mentioned by \cite{higham}, the model in that paper can be transformed to the models of the other papers surveyed here, at a high translation costs. Hence, the complexities of the algorithm of \cite{srimani-shared} may grow even further when stated in these models.
Gupta and Srimani \cite{gupta-srimani-see-sebastian} presented
an  $O(n\log n)$ bits algorithm.
Higham and Liang \cite{higham} improved the core memory requirement to $O(\log n)$, however, the time complexity went up again to $\Omega(n |E|)$.
An algorithm with a similar time complexity and a similar memory per node was also presented by
Blin, Potop-Butucaru, Rovedakis, and Tixeuil
\cite{sebastian-disc2009}.
This latter algorithm exchanges less bits with neighbours than does the algorithm of \cite{higham}.
 The algorithm of \cite{sebastian-disc2009} addressed also another goal- even during stabilization it is {\em loop free}. That is, it also maintains a tree at all times
(after reaching an initial tree).
This algorithm assumes the existence of a unique leader in the network (while the algorithm in the current paper does not). However,
this does not seem to affect the order of magnitude of the time complexity.

Note that the memory size in the last two algorithms above is the same as in the current
paper. However, their time complexity is $O(|E|n)$ versus $O(n)$ in the current paper.
 The time complexity of  the algorithm of
 Blin, Dolev, Potop-Butucaru, and Rovedakis  \cite{sebastian-disc-2010}
  improved the time complexity of \cite{sebastian-disc2009,higham} to
    $O(n^2)$ but at the cost of growing the memory usage to  $O(\log^2 n)$.
  This may be the first paper using labeling schemes
for the design of a self-stabilizing MST protocol, as well as
the first paper implementing the
algorithm by Gallager, Humblet, and Spira in a self-stabilizing manner without using a general transformer.

Additional studies about MST verification in various models appeared in
\cite{DHKKNPPW, DixonRauchTarjan,DixonTarjan, KKP11,KormanKutten07,KKP10}. In particular, Kor et al. ~\cite{KKP11} shows that the verification from scratch (without labels) of
an MST requires $\Omega(\sqrt{n}+D)$ time and $\Omega(|E|)$ messages, and that these bounds are tight up to poly-logarithmic factors.
We note that the memory complexity  was not considered in   \cite{KKP11}, and indeed the memory used therein is much higher than the one used in the current paper.
The time lower bound proof in   \cite{KKP11}
was later extended in \cite{DHKKNPPW} to apply for a variety of  verification and computation tasks.

This paper has results concerning distributed verification. Various additional papers dealing with verification have appeared recently, the models of some of them are rather different than the model here.
Verification in the $\cal{LOCAL}$ model (where congestion is abstracted away) was studied in
\cite{FKP11} from a computational complexity perspective. That paper presents various complexity classes, shows separation between them, and provides complete problems for these classes.
In particular, the class NLD defined therein exhibits
similarities to the notion of proof labeling schemes. Perhaps the main result in \cite{FKP11} is a sharp threshold for the impact of randomization on
local decision of hereditary languages.
Following that paper, \cite{BPLD} showed that the threshold in \cite{FKP11} holds also for any non-hereditary language, assuming it is defined on path topologies. In addition, \cite{BPLD}
showed further limitations of randomness, by presenting a hierarchy of languages, ranging from deterministic, on the one side of the spectrum, to almost complete randomness, on the other side. Still, in the $\cal{LOCAL}$ model, \cite{OPODIS} studied the impact of assuming unique identifiers on local decision.
We stress that the memory complexity  was not considered in neither  \cite{FKP11} nor in its follow-up papers \cite{BPLD,OPODIS}.

\subsection{Our results}
This paper contains the following two main results.

\paragraph{(1) Solving an open problem posed by  Awerbuch and Varghese \cite{awerbuch-varghese}:}
In the context of self-stabilization, an open problem posed in \cite{awerbuch-varghese} is to find a (synchronous) MST verification algorithm whose  time complexity is smaller than the MST computation time, yet with a small memory.
Our first main result solves this question positively by constructing  a time efficient self-stabilizing verification algorithm for an MST while using  optimal memory size, that is $O(\log n)$ bits of
memory per node. More specifically, the verification scheme takes as input  a distributed structure claimed to be an MST. If the distributed structure is indeed an MST, and if
a marker algorithm properly marked the nodes to allow the verification, and if no faults occur, then
our algorithm outputs {\em accept} continuously in every node. However, if faults occur (including the case that the structure is not, in fact, an MST, or that the marker did not perform correctly),
then our algorithm outputs {\em reject} in at least one node. This {\em reject} is outputted in time
$O(\log^2 n)$ after the faults cease, in a synchronous network.
(Recall,  lower bounds which are polynomial in $n$ for MST {\em construction} are known even for synchronous networks \cite{lotker-patt-shamir-peleg,mst-lower-bound}.) In  asynchronous networks, the
 time complexity of our verification scheme grows to $O(\Delta \log^3 n)$.
 We also show that $\Omega(\log n)$ time is necessary if the memory size is restricted to $O(\log n)$, even in
synchronous networks.
 Another property of our verification scheme is that
 if $f$ faults occurred,  then, within  the required
detection  time above, they are detected by some node in the $O(f\log n)$
locality of each of the faults.
Moreover, we present a distributed implementation of  the marker algorithm whose time complexity for assigning the labels is $O(n)$, under the same memory size constraint of
$O(\log n)$ memory bits per node.

\paragraph{(2) Constructing an asynchronous self-stabilizing MST construction algorithm which uses optimal memory ($O(\log n)$ bits) and  runs in $O(n)$ time:} In our second main result, we show how to enhance a known transformer that makes
 input/output algorithms self-stabilizing. It now takes as input an efficient construction algorithm and an efficient self-stabilizing proof labeling scheme, and produces
 an efficient self-stabilizing algorithm. When used with our verification scheme, the transformer produces a memory optimal  self-stabilizing  MST construction algorithm, whose time complexity, namely, $O(n)$,
 is significantly better even than that of previous algorithms.
(Recall, the time complexity of previous  MST algorithms that used $\Omega(\log^2 n)$ memory bits per node  was $O(n^2)$, and the time for optimal space  algorithms was $O(n|E|)$.)
Inherited from our verification scheme, our self-stabilising MST construction algorithm also
has the following two properties. First, if faults occur after the construction ended, then they are  detected by
some nodes within $O(\log ^2 n)$  time
in synchronous networks, or within $O(\Delta \log ^3 n)$ time  in
asynchronous ones, and second, if $f$ faults occurred,  then, within  the required
detection  time above, they are detected within the $O(f\log n)$
locality of each of the faults.
We also show how to improve these two properties, at the expense of some increase in the memory.

\subsection{Outline}
Preliminaries and some examples of simple, yet useful,  proof labeling schemes are given in Section \ref{sec:preliminaries}. An intuition is given in Section \ref{sec:intuition}.
A building block is then given in Section
\ref{sec:mst-construction}. Namely, that section describes a
synchronous MST construction algorithm in $O(\log n)$ bits memory size and $O(n)$ time.
Section \ref{sec:section5} describes the construction of parts of the labeling scheme.
Those are the parts that use labeling schemes of the kind described in \cite{KKP10}- namely, schemes that can be verified in one time unit.
These parts use
the MST construction (of Section \ref{sec:mst-construction}) to assign the labels.
 Sections \ref{sec:proof},
\ref{sec:viewing}, and \ref{sub:utilizing} describe the remaining part of the labeling scheme. This part is a labeling scheme by itself, but of a new kind. It saves on memory by distributing information.  Specifically, Section~\ref{sec:proof} describes how the labels should look if they are constructed correctly (and if an MST is indeed represented in the graph). The verifications, in the special case that no further faults occur, is described in Section \ref{sec:viewing}.
This module verifies (alas, not in constant time) by moving the distributed information around, for a ``distributed viewing''. Because the verification is not done in one time unit, it needs to be made self stabilizing. This is done in Section~\ref{sub:utilizing}.
Section \ref{sec:lower} presents a lower bound for the time of a proof labeling scheme for MST that uses only logarithmic memory. (Essentially, the proof behind this lower bound is based on a simple reduction, using the rather complex lower bound  given in \cite{{KormanKutten07}}.)

The efficient self-stabilizing  MST algorithm is given in Section~\ref{sec:full-mst-alg}.
Using known transformers, we combine efficient MST verification schemes and (non-self-stabilizing) MST construction schemes to yield efficient self-stabilizing schemes.
The MST construction algorithm described in Section~\ref{sec:mst-construction} is a variant of some known time efficient MST construction algorithms. We show there how those can also be made memory efficient (at the time, this complexity measure was not considered),
and hence can be used as modules for our optimal memory self-stabilizing  MST algorithm.

\section{Preliminaries}
\label{sec:preliminaries}
\label{sec:def}

\subsection{Some general definitions}
\label{sub:Some general definitions}
 We use rather standard definitions; a reader unfamiliar with these notions may refer to the model descriptions in the rich literature on these subjects.~In~particular, we use  rather standard definitions of self-stabilization (see, e.g. \cite{dolev-book}).
 Note that the assumptions we make below on time and time complexity imply (in self stabilization jargon) a distributed daemon with a very strong fairness. When we speak of asynchronous networks, this implies a rather fine granularity of atomicity.
  Note that the common self stabilization definitions include the definitions of faults.
  We also use standard definitions of graph theory (including
 an edge weighted graph $G=(V,E)$, with weights that are polynomial in $n=|V|$) to represent a network (see, e.g.  \cite{even-book}).
Each node $v$ has a unique identity $\id(v)$ encoded using $O(\log n)$ bits.
For convenience, we assume that each adjacent edge of each node $v$ has some label that is unique at $v$ (edges at different nodes may have the same labels). This label, called a {\em port-number}, is known only to $v$ and is independent of the port-number of the same edge at the other endpoint of the edge. (Clearly, each port-number can be assumed to be encoded using $O(\log n)$ bits).
Moreover, the network can store an object such as an MST (Minimum Spanning Tree) by having each node store its {\em component} of the representation. A component $c(v)$ at a node $v$ includes
a collection of  pointers (or port-numbers) to  neighbours of $v$, and the collection of the components of all nodes induces a subgraph $H(G)$ (an edge is included in $H(G)$ if and only if at least one of its end-nodes points
at the other end-node). In the verification scheme considered in this current paper, $H(G)$ is supposed to be an~MST and for simplicity, we assume that the component of each node contains a single pointer
(to the parent, if that node is not defined as the root). It is not difficult to extend our verification scheme to hold
also for the case where each component can contain  several pointers. Note that the definitions in this paragraph imply a lower bound of $\Omega(\log n)$ bits on the memory required at each node to even represent an MST (in graphs with nodes of high degree).

Some additional standard (\cite{ghs}) parts of the model include the
assumption that
the edge weights are distinct. As noted often, having distinct edge weights simplifies our arguments
since it guarantees the uniqueness of the MST. Yet,  this assumption is not essential. Alternatively, in case the graph is not guaranteed
to have distinct edge weights, we may  modify the weights locally as was done  in \cite{KKP11}.
The resulted modified weight function $\omega'(e)$ not only assigns distinct edge weights, but also satisfies the property that
 the given subgraph $H(G)$
 is an MST of $G$ under $\omega(\cdot)$ if and only if  $H(G)$ is an MST
of $G$ under $\omega'(\cdot)$.\footnote{We note,  the standard technique (e.g., \cite{ghs}) for obtaining unique weights is not sufficient for our purposes.
Indeed, that technique orders edge weights lexicographically:
first, by their original weight $\omega(e)$, and then, by the identifiers of
the edge endpoints.
This yields a modified graph with unique edge weights, and an MST of the modified graph is necessarily an MST of the original graph.
For construction purposes it is therefore sufficient to consider only
the modified graph. Yet, this is not the case for verification purposes,
as the given subgraph can be an MST of the original graph but not necessarily
an MST of the modified graph.
While the authors in \cite{KKP11} could not guarantee that any MST of the original graph is an MST of
the modified graph (having unique edge weights), they instead make sure that
the particular given subgraph $T$ is an MST of the original graph
if and only if it is an MST of modified one. This condition is sufficient for verification purposes, and allows one to consider only the modified graph.
For completeness, we describe the weight-modification in \cite{KKP11}. To obtain the modified graph, the authors in \cite{KKP11} employ the technique, where
edge weights are lexicographically ordered as follows.
For an edge $e=(v,u)$ connecting $v$ to its  neighbour~$u$,
consider first its original weight $\omega(e)$,
next, the value $1-Y_u^v$ where $Y_u^v$ is the indicator variable of the edge
$e$ (indicating whether $e$ belongs to the candidate MST to be verified),
and finally, the identifiers of the edge endpoints, $\ID(v)$ and $\ID(u)$
(say, first comparing the smaller of the two identifiers of the endpoints,
and then the larger one).
Formally, let
$\omega'(e) ~=~
\left\langle \omega(e), 1-Y_u^v, \ID_{min}(e), \ID_{max}(e) \right\rangle~,$
where $\ID_{min}(e) = \min\{\ID(v),\ID(u)\}$
and $\ID_{max}(e) = \max\{\ID(v),\ID(u)\}$.
Under this weight function $\omega'(e)$, edges with indicator variable set to 1
will have lighter weight than edges with the same weight under $\omega(e)$
but with indicator variable set to 0
(i.e., for edges $e_1\in T$ and $e_2\notin T$ such that
$\omega(e_1)=\omega(e_2)$, we have $\omega'(e_1)< \omega'(e_2)$).
It follows that  the given subgraph $T$
 is an MST of $G$ under $\omega(\cdot) $ if and only if  $T$ is an MST
of $G$ under $\omega'(\cdot)$. Moreover, since $\omega'(\cdot)$ takes into account
the unique vertex identifiers, it assigns distinct edge weights.}

We use the (rather common) ideal time complexity which assumes that a node reads all of its neighbours in at most one time unit,
see e.g.~\cite{sebastian-disc2009,sebastian-disc-2010}. Our results translate easily to
an alternative, stricter, {\em contention} time complexity, where a node can access only one neighbour in one time unit. The time cost of such a translation is at most a multiplicative factor of
$\Delta$, the   maximum degree of a node  (it is  not assumed that $\Delta$ is known to nodes).

As is commonly assumed in the case of self-stabilization,
each node has only some bounded number
 of  memory bits available to be used.
Here, this amount of memory is $O(\log n)$.

\subsection{Using protocols designed for message passing}

We use a self stabilizing transformer of Awerbuch and Varghese as a building block \cite{awerbuch-varghese}.
That protocol was designed for the message passing model.
Rather than modifying that transformer to work on the model used here (which would be very easy, but would take space), we use emulation. That is,
we claim that any self stabilizing protocol designed for the
model of \cite{awerbuch-varghese} (including the above transformer) can be performed in the model used here, adapted from \cite{sebastian-disc2009,sebastian-disc-2010}.
This is easy to show: simply use the current model to implement the links of the model of \cite{awerbuch-varghese}. To send a message from node $v$ to its neighbour $u$,
have $v$ write its shared variable that (only $v$ and) $u$ can read.
This value can be read by $u$ after one time unit in a synchronous network
as required from a message arrival in the model of \cite{awerbuch-varghese}. Hence, this is enough for synchronous networks.

In an asynchronous network, we need to work harder to simulate the sending and the receiving of a message, but only slightly harder, given known art.
Specifically, in an asynchronous network, an event occurs at $u$ when this message arrives. Without some additional precaution on our side, $u$ could have read this value many times (per one writing) resulting in duplications:
multiple message ``arriving'' while we want to emulate
just one message. This is solved by a self stabilizing data link protocol, such as the one used by \cite{AKY}, since this is also the case in a data link protocol in message passing systems where a link may loose a package. There, a message is sent repeatedly, possibly many times, until an acknowledgement from the receiver tells the sender that the message arrived. The data link
protocol overcomes the danger of duplications by simply numbering the messages modulo some small number. That is, the first message is sent repeatedly
with an attached ``sequence number'' zero, until the first acknowledgement arrives. All the repetitions of the second message have as attachments the sequence number~1, etc.
The receiver then takes just one of the copies of the first message, one of the copies of the second, etc. A self stabilized implementation of this idea
in a shared memory model appears in
\cite{AKY} using (basically, to play the role of the sequence number) an additional shared variable called the ``toggle'', which can take one of three values.\footnote{That protocol, called ``the strict discipline'' in \cite{AKY}, actually provides a stronger property (emulating a coarser grained atomicity), not used here.}
When $u$ reads that the toggle of $v$ changes, $u$ can emulate the arrival of a message. In terms of time complexity, this protocol takes a constant time, and hence sending (an emulated) message
still takes a constant time (in terms of complexity only) as required to emulate the notion of ideal time complexity of \cite{sebastian-disc2009,sebastian-disc-2010}. Note that the memory is not increased.

\subsection{Wave\&Echo}
We use the well known Wave\&Echo (PIF) tool. For details, the readers are referred to \cite{pif1, pif2}. For completeness, we remind the reader of the overview of Wave\&Echo when performed over a rooted tree. It is started by the tree root, and every node who receives the wave message forwards it to its children. The wave can carry a command for the nodes. A leaf receiving the wave, computes the command, and sends the output to its parent. This is called an {\em echo}. A parent, all of whose children echoed, computes the command itself (possibly using the outputs sent by the children) and then sends the echo (with its own output) to its parent. The Wave\&Echo terminates at the root when all the children of the root echoed, and when the root executed the command too.

 In this paper, the Wave\&Echo activations carry various commands. Let us describe first two of these commands, so that they will also help clarify the notion of Wave\&Echo and its application.
    The first example is the command to sum up values residing at the nodes. The echo of a leaf includes its value. The echo of a parent includes the sum of its own value and the sums sent by its children. Another example is the counting of the nodes. This is the same as the sum operation above, except that the initial value at a node is $1$. Similarly to summing up, the operation performed by the wave can also be a logical $\OR$.

\subsection{Proof labeling schemes}\label{sec:proof-labeling}
In \cite{KormanKutten07,KKP10,GS11}, the authors consider  a framework
 for maintaining a  distributed proof that the network satisfies some given predicate $\Psi$, e.g., that $H(G)$  is an MST.
We are given a predicate $\Psi$ and a graph family $\cF$
(in this paper, if $\Psi$ and $\cF$ are omitted, then $\Psi$ is MST and $\cF$ (or $\cF(n)$) is  all connected undirected weighted graphs with $n$ nodes).
A {\em proof labeling scheme} (also referred to as a {\em verification} algorithm)
includes the following two components.

\begin{itemize}
 \item
 A   {\em marker} algorithm $\cM$ that  generates a label $\cM(v)$ for every node $v$ in every graph $G\in \cF$.
 \item
A  {\em verifier}, that is a distributed algorithm $\cV$, initiated at
each node of
 a {\em labeled} graph $G\in \cF$, i.e., a graph whose nodes $v$ have labels $L(v)$ (not necessarily correct labels assigned by a marker). The verifier at each node is initiated separately, and at an arbitrary time, and runs forever.
The verifier may {\em raise an alarm} at some node $v$ by outputting ``no'' at $v$.
\end{itemize}
Intuitively,
if the verifier at $v$ raises an alarm, then it detected a fault.
That is, for any graph $G\in \cF$,
\begin{itemize}
\item
If $G$ satisfies the predicate $\Psi$  and
if the label at each node $v$ is $\cM(v)$ (i.e., the  label assigned to $v$ by the marker algorithm $\cM)$
then no node raises an alarm.
In this case, we say that the verifier {\em accepts} the labels.
\item
If $G$ does not satisfy the predicate $\Psi$, then for {\em any} assignment of labels to the nodes of $G$,
after some finite time $t$, there exists a node $v$ that raises an alarm. In this case, we say that the verifier {\em rejects} the labels.
\end{itemize}

Note that the first property above concerns only the labels {\em produced
by the marker algorithm} $\cM$, while the second must hold even if the labels are assigned by some adversary.
We evaluate  a proof labeling scheme
$(\cM,\cV)$
by the following complexity measures.

\begin{itemize}
\item
The {\em memory size}:
 the maximum number of bits stored in the memory of a single node $v$,
 taken over all the nodes $v$ in all graphs  $G\in \cF(n)$ that satisfy the predicate $\Psi$ (and over all the executions);
 this includes:  (1) the bits used for encoding the identity $\id(v)$, (2) the marker memory: number of bits used for constructing and encoding the labels,  and (3) the verifier memory: the number of bits used during the operation of the verifier\footnote{Note that we do not include the number of bits needed for storing the component $c(v)$ at each node $v$. Recall that for simplicity, we assume here that each component contains a single pointer, and therefore, the size of
each component is $O(\log n)$ bits. Hence, for our purposes, including the size of a component in the memory complexity would not increase the asymptotical size of the memory, anyways.
However, in the general case, if multiple pointers can be included in a component, then the number of bits needed for encoding a component would potentially be as large as $O(\Delta \log n)$.
Since, in this case, the verification scheme has no control over the size of the component, we decided to exclude this part from the definition of the memory complexity.}.
\item
The (ideal) {\em detection time}: the maximum,  taken over all the graphs $G\in \cF(n)$ that do
{\em not} satisfy the predicate $\Psi$ and over all the labels given to nodes of $G$ by adversaries (and over all the executions), of the time $t$
 required for some node  to raise an alarm. (The time is counted from the {\em starting time}, when the verifier has been initiated at all the nodes.)
\item
The {\em detection distance}:
 for
 a faulty node $v$, this is the (hop) distance to the closest node $u$ raising an alarm within the detection time after the fault occurs.
 The detection distance of the scheme is the maximum, taken over all the graphs having at most $f$ faults, and over all the  faulty nodes $v$ (and over all the executions), of the detection distance of $v$.
\item
The (ideal) {\em construction-time}: the maximum, taken over all  the graphs  $G\in \cF(n)$ that satisfy the predicate $\Psi$ (and over all the executions), of the
time required for the marker $\cM$ to assign  labels to all nodes of $G$. Unless mentioned otherwise, we measure construction-time in synchronous networks only.

\end{itemize}
In our terms,
the definitions of  \cite{KormanKutten07,KKP10}
allowed only detection time complexity 1. Because of that, the verifier of
\cite{KormanKutten07,KKP10} at a node $v$,  could only
consult the neighbours of~$v$. Whenever we use such a scheme, we
refer to it as a {\em $1$-proof labeling scheme},  to emphasis its running
time. Note that a 1-proof labelling scheme is trivially self-stabilzying.
(In some sense, this is because they ``silently stabilize'' \cite{silent}, and
``snap stabilize'' \cite{BDPV}.)
Also, in \cite{KormanKutten07,KKP10}, if $f$ faults occurred, then the detection distance was $f$.


\subsection{Generalizing the complexities to a computation}
\label{sub:detection-time}
In Section~\ref{sec:proof-labeling}, we defined the memory size, detection time  and the detection distance complexities of a {\em verification} algorithm.
When considering a (self-stabilizing) computation algorithm, we extend the notion of the memory size to include also  the bits needed for encoding the component $c(v)$ at each node. Recall, the definition of a component $c(v)$ in general, and the special case of $c(v)$ for MST, are given in Section \ref{sec:preliminaries}.
(Recall, from Section~\ref{sec:proof-labeling}, that the size of the component was excluded from the definition of memory size for verification because, there, the designer of the verification scheme has no control over
the nodes' components.)

The notions of detection time  and the detection distance can be carried to the very common class of self-stabilizing {\em computation} algorithms that use fault detection. (Examples for such algorithms
are algorithms that have silent stabilization \cite{silent}.)
Informally, algorithms in this class first compute an output.
After that, all the nodes are required to stay in some {\em output state} where they (1) output the computation result forever  (unless a fault occurs); and (2)
  check repeatedly until they discover a fault.
  In such a case, they recompute and enter an output state again.
  Let us term such algorithms {\em detection based} self-stabilizing algorithms.
 We define the {\em detection time} for such algorithms to be
  the time from a fault until the detection.
 However, we only define the detection time (and the detection distance) for faults that occur after all the nodes are in their output states.
(Intuitively, in the other cases, stabilization has not been reached yet anyhow.) The {\em detection distance} is the distance from a node where a fault occurred to the closest node that detected a fault.

\subsection{Some examples of 1-proof labeling schemes}
\label{sub:examples}

As a warm-up exercise, let us begin by describing several simple 1-proof labeling schemes, which will be useful later in this paper.
The first two examples are taken from \cite{KKP10} and are explained there in more details. The reader familiar with proof labeling schemes may skip this subsection.\\

\noindent{\bf $\circ$ Example 1: A spanning tree. (Example $\SP$)} Let $f_{span}$ denote the predicate such that for any input graph $G$
satisfies, $f_{span}(G)=1$ if $H(G)$ is
a spanning tree of $G$, and $f_{span}(G)=0$ otherwise. We now describe an efficient $1$-proof
labeling scheme $(\cM,\cV)$ for the predicate $f_{span}$ and
the family of all  graphs. Let us first describe
the marker $\cM$ operating on a ``correct instance'', i.e.,
a graph $G$ where $T=H(G)$ is indeed a spanning tree of $G$. If there exists
a node $u$ whose component
does not encode a pointer to any of its adjacent edges (observe that there can be at most
one such node), we root $T$
at $u$. Otherwise,
there must be two nodes $v$ and $w$ whose components point at each other. In
this case,
we root $T$ arbitrarily at either $v$ or $w$. Note that after rooting $T$, the
component at each non-root node
$v$ points at $v$'s parent. The label given by $\cM$ to a node $v$ is
composed of two parts. The first part
encodes the identity $\id(r)$ of the root $r$ of $T$, and the second part of
the label of $v$ encodes the distance
(assuming all weights of edges are 1) between $v$ and $r$ in $T$.

Given a labeled graph, the verifier $\cV$ operates at a node $v$ as follows:
first, it checks that the first
part of the label of $v$ agrees with the first part of the labels of $v$'s
neighbours, i.e., that $v$ agrees with its neighbours on
the identity of the root. Second, let $d(v)$ denote the number encoded in
 the second part of $v$'s label. If $d(v)=0$ then
$\cV$ verifies that $\id(v)=\id(r)$ (recall that $\id(r)$ is encoded in the
first part of $v$'s label). Otherwise, if  $d(v)\neq 0$ then
the verifier checks that $d(v)=d(u)+1$, where
$u$ is the node pointed at by the component at $v$. If at least one of these conditions
fails, the verifier $\cV$ raises an alarm at $v$.
It is easy to get convinced that  $(\cM,\cV)$ is indeed a 1-proof labeling
scheme for the predicate $f_{span}$
with memory size $O(\log n)$ and construction time $O(n)$.

\paragraph{Remark.} Observe that in case $T=H(G)$ is indeed a (rooted)
spanning tree of $G$, we can easily let each node~$v$ know which
of its neighbours in $G$ are its children in $T$ and which is its parent.
Moreover, this can be done using one unit of time
 and label size $O(\log n)$ bits. To see this, for each node $v$, we simply
add to its label its identity $\id(v)$  and the identity $\id(u)$ of its
parent $u$.
The verifier at $v$ first verifies that $\id(v)$ is indeed encoded in the
right place of its label. It then looks at the  label of its parent $u$,
and  checks that $v$'s candidate for being the identity of $u$ is indeed
$\id(u)$. Assume now that these two conditions are
satisfied at each node. Then, to identify a child $w$ in $T$, node~$v$
should only look at the labels of its neighbours in $G$  and see which of
them encoded  $\id(v)$ in the
designated place of its label.\\

\noindent{\bf $\circ$ Example 2: Knowing the number of nodes (Example $\NK$)}
Denote by $f_{size}$  the boolean predicate such that  $f_{size}(G)=1$
if and only if one designated part of the label $L(v)$ at each node $v$  encodes the number of nodes in
$G$
 (informally, when $f_{size}$ is satisfied,
we say that each node `knows' the number of nodes in $G$). Let us denote the part of the label of $v$ that is supposed to hold this number by $L'(v)$.

In \cite{KKP10}, the authors give a 1-proof labeling scheme $(\cM,\cV)$  for
$f_{size}$ with   memory size $O(\log n)$.
The idea behind their scheme is simple. First, the verifier checks
 that $L'(u)=L'(v)$ for every two adjacent nodes $u$ and $v$ (if this holds at
each node then all nodes must hold the same candidate
for being the number of nodes). Second,
the marker constructs a  spanning tree $T$ rooted at some node $r$ (and verifies that this is indeed a spanning tree using the Example $\SP$ above).
Third,
the number of nodes in $T$ is
aggregated upwards along $T$ towards its root $r$, by keeping at the
label $\cM(v)$ of each node $v$, the number of nodes $n(v)$
in the subtree of $T$  hanging down from $v$. This again is easily verified
by checking at each node $v$ that $n(v)=1+\sum_{u\in \Ch(v)}n(u)$, where
$\Ch(v)$ is the set of children of $v$. Finally, the root verifies that
$n(r)=L'(r)$. It is straightforward that  $(\cM,\cV)$ is indeed a 1-proof labeling
scheme for the predicate $f_{size}$
with memory size $O(\log n)$ and construction time $O(n)$.\\

\noindent{\bf $\circ$ Example 3: An  upper bound on the diameter of a tree (Example $\EDIAM$)}
Consider a tree $T$ rooted at $r$, and let~$h$ denote
the height of  $T$.
Denote by $f_{height}$  the boolean predicate such that  $f_{height}(T)=1$
if and only if one designated part of the label $L(v)$ at each node encodes the same value $x$,
where $h\leq x$
(informally, when $f_{height}$ is satisfied,
we say that each node `knows' an upper bound of $2x$ on the diameter $D$ of
$T)$. Let us denote the part of the label of $v$ that is supposed to hold this number by $L'(v)$.
We sketch a simple 1-proof labeling scheme $(\cM,\cV)$  for $f_{height}$.
 First, the verifier checks
 that $L'(u)=L'(v)$ for every two adjacent nodes $u$ and $v$ (if this holds at
each node then all nodes must hold the same value $x)$.
Second, similarly to the proof labeling scheme for $f_{span}$ given in Example $\SP$ above, the label in each node~$v$ contains the distance $d(v)$ in the tree
from $v$ to the root. Each non-root node verifies that the distance written
at its parent
is one less than the distance written at itself, and the root verifies that
the distance written at itself is 0. Finally, each node $v$ verifies also
that $x\geq d(v)$.
If no node raises an alarm then~$x$ is an upper bound on the height.
On the other hand, if the value $x$ is the same at all nodes and $x$ is an
upper bound on the height then no node raises an alarm. Hence the scheme is
  a 1-proof labeling scheme for the predicate~$f_{height}$
with memory size $O(\log n)$ and construction time $O(n)$.

\section{Overview of the MST verification scheme and the intuition behind it}
\label{sec:intuition}
The final MST construction algorithm makes use of several modules. The main technical contribution of this paper is the module that verifies that the collection of nodes' components is indeed an MST. This module in itself is composed of multiple modules. Some of those, we think may be useful by themselves. To help the reader avoid getting lost in the descriptions of all the various modules, we present, in this section, an overview of the MST verification part.

Given a graph $G$ and a subgraph that is represented distributively at the nodes of $G$, we wish to produce a self-stabilizing proof labeling scheme that verifies whether the subgraph is an MST of $G$.
By first employing the (self-stabilizing) 1-proof labeling scheme mentioned in Example $\SP$,
we may assume that the subgraph is a rooted spanning tree of $G$ (otherwise, at least one node would raise an alarm). Hence, from now on, we focus on a spanning tree $T=(V(G),E(T))$ of a weighted
graph $G=(V(G),E(G))$,  rooted at some node $r(T)$,  and aim at verifying the minimality of $T$.

\subsection{Background and difficulties}
From a high-level perspective, the proof labeling scheme proves that  $T$ could have been computed by an algorithm that is similar to that
of GHS, the algorithm of Gallager, Humblet, and Spira's described in \cite{ghs}. This leads to a simple idea:
when $T$ is a tree computed by such an algorithm, $T$ is an MST.
Let us first recall a few terms from \cite{ghs}. A {\em fragment} $F$ denotes a connected subgraph of $T$ (we simply refer it to a {\em subtree}).
Given a fragment $F$, an edge $(v,u)\in E(G)$ whose one endpoint $v$ is in
$F$, while the other endpoint $u$ is not, is called {\em outgoing} from $F$. Such an edge
 of minimum weight  is called a {\em minimum outgoing} edge from~$F$.
A fragment  containing a single node is called a {\em singleton}. Recall that
GHS starts when each node is a fragment by itself. Essentially, fragments merge over their minimum outgoing edges to form larger fragments. That is, each node belongs to one fragment $F_1$, then to a larger fragment $F_2$ that contains $F_1$, etc.
This is repeated until one fragment spans the network. A tree constructed in that way is an MST.
Note that in GHS, the collection of fragments is a laminar family, that is, for any two fragments $F$ and $F'$ in the collection, if $F\cap F'\neq \emptyset$ then either $F\subseteq F'$ or $F'\subseteq F$ (see, e.g. \cite{laminar}). Moreover, each fragment has a {\em level}; in the case of $v$ above, $F_2$'s level is higher than that of $F_1$.
This organizes the fragments in a {\em hierarchy} $\cH$, which is a tree whose nodes are fragments, where fragment $F_1$ is a descendant in $\cH$ of $F_2$ if $F_2$ contains $F_1$.
GHS manages to ensure
 that each node belongs to at most one fragment at each level, and that the total number of levels is $O(\log n)$.
Hence, the hierarchy $\cH$ has $O(\log n)$ height.

The marker algorithm in our proof labeling scheme performs, in some sense, a reverse operation. If $T$ is an MST, the marker
``slices'' it back into fragments. Then, the proof labeling scheme computes for each node $v$:
\begin{smallitemize}
\item
The (unique) name of each of the fragments $F_j$ that $v$ belongs to,
\item the level of
$F_j$,  and
\item  the weight of $F_j$'s minimum outgoing edge.
\end{smallitemize}
Note that each node participates in $O(\log n)$ fragments, and the above ``piece of information'' per fragment requires $O(\log n)$ bits.
Hence, this  is really too much information to store in one node. As we shall see later, the verification scheme distributes this information and then brings it
 to the node without violating the memory size bound $O(\log n)$. For now, it suffices to know that given these pieces of information and the corresponding pieces of information of their neighbours, the nodes can verify that
  $T$ could have been constructed by an algorithm similar to GHS. That way, they verify that $T$ is an MST.
  Indeed, the 1-proof labeling schemes for MST verification given in \cite{KormanKutten07,KKP10} follow
  this idea employing memory size of $O(\log^2 n)$ bits.
  (There, each node keeps $O(\log n)$ pieces, each of $O(\log n)$ bits.)

 The current paper allows
 each node
  to hold only $O(\log n)$ memory bits. Hence, a node
 has room for only a constant number of such pieces of information at a time.
 One immediate idea is to store some of $v$'s pieces in some other nodes.
 Whenever $v$ needs a piece, some algorithm should move it towards $v$. Moving pieces would cost  time, hence, realizing some time versus memory size trade-off.

 Unfortunately, the total (over all the nodes) number of pieces in the schemes of \cite{KormanKutten07,KKP10} is $\Omega(n \log n)$. Any way one would assign these pieces to nodes would result in the memory of some single node
 needing to store $\Omega(\log n)$ pieces, and hence, $\Omega(\log^2 n)$ bits. Thus, one technical step we used here is to reduce the total number of pieces
to $O(n)$, so that we could store at each node just a constant number of such pieces. However,
 each node still needs to use $\Omega(\log n)$ pieces.
  That is, some pieces may be required by many nodes.
Thus, we needed to solve also a combinatorial problem: locate each piece ``close'' to each of the nodes needing it, while storing only a constant number of pieces per node.

 The solution of this combinatorial  problem would have sufficed to construct
 the desired scheme in the $\cal{LOCAL}$ model~\cite{peleg-book}. There, node $v$ can ``see'' the storage of nearby nodes.
However, in the congestion aware model, one actually needs to move pieces from node to node, while not violating the $O(\log n)$ memory per node constraint.
This is difficult, since, at the same time $v$ needs to see its own pieces, other nodes need to see their own ones.

\subsection{A very high level overview}
Going back to GHS, one may notice that its correctness follows from the combination of two properties:

\begin{smallitemize}
\item {\bf P1. (Well-Forming)} The existence of a hierarchy tree $\cH$ of fragments, satisfying the following:
\begin{smallitemize}
\item Each fragment $F\in \cH$ has a unique {\em selected} outgoing edge (except when $F$ is the whole tree $T$).
\item A (non-singleton) fragment is obtained by merging its children fragments in $\cH$ through their selected outgoing edges.
\end{smallitemize}
\item {\bf P2. (Minimality)} The selected outgoing edge of each fragment is its minimal outgoing edge.
\end{smallitemize}
In our proof labeling scheme, we verify the aforementioned two properties separately.
In Section \ref{sec:section5}, we show how to verify the first property, namely, property Well-Forming. This turns out to be a much easier task than verifying the second property.
Indeed, the Well-Forming property can be verified using a 1-proof labeling scheme, while still obeying the $O(\log n)$ bits per node memory constraint. Moreover, the techniques we use for verifying the Well-Forming are rather
 similar to the ones described in \cite{KKP10}.
The more difficult verification task, namely, verifying the Minimality property P2, is described in Section~\ref{sec:proof}. This verification scheme requires us to come up with several new techniques which may be considered as contributions by themselves. We now describe the intuition behind these verifications.

\subsection{Verifying the Well-Forming property (described in detail in Sections \ref{sec:mst-construction} and  \ref{sec:section5})}
We want to show that $T$ could have been produced by an algorithm similar to GHS.
Crucially, since we care about the memory size,  we had to come up with a new MST construction algorithm that is similar to GHS but uses only $O(\log n)$ memory bits per node and runs in time $O(n)$. This MST construction algorithm, called $\ALG$, can be considered as a synchronous variant of GHS and is described in Section \ref{sec:mst-construction}.

Intuitively, for a correct instance (the case $T$ is an MST), the marker algorithm $\cM$ produces a hierarchy  of fragments $\cH$ by following the new MST construction algorithm described in Section \ref{sec:mst-construction}. Let $\ell=O(\log n)$ be the height of $\cH$. For a fixed level $j\in [0,\ell]$, it is easy to represent the partition of the tree into fragments of level $j$ using just one bit per node. That is, the root $r'$ of each fragment $F'$ of level $j$ has 1 in this bit, while
the nodes in $F'\setminus \{r'\}$ have 0 in this bit.
Note, these nodes in $F'\setminus \{r'\}$ are the nodes below $r'$ (further away from the root of $T$), until (and not including)
reaching additional nodes whose corresponding bit is~1.
Hence, to represent the whole hierarchy, it is enough to
attach a string of length $\ell+1$-bits at each node $v$.
The string at a node $v$ indicates, for each level $j\in [0,\ell]$, whether $v$ is the root of the fragment of level $j$ containing $v$ (if one exists).

Next, still for correct instances, we would like to represent the selected outgoing edges distributively. That is, each node $v$ should be able to detect, for each fragment $F$ containing $v$, whether $v$ is an endpoint
of the selected edge of $F$. If $v$ is, it should also know which of $v$'s edges is the selected edge.  This representation is used later for verifying the two items of the Well-Forming property specified above. For this purpose, first, we add another string of $\ell+1$ entries at each node $v$, one entry per level $j$. This entry specifies, first, whether there exists a level $j$ fragment $F_j(v)$ containing $v$. If $F_j(v)$ does exist, the entry specifies whether  $v$ is incident to the corresponding selected edge.
Note, storing the information at $v$ specifying the pointers  to all the selected edges of the fragments containing it,  may require $O(\log^2 n)$ bits of memory at $v$.
This is because there may be $O(\log n)$ fragments containing $v$; each of those may select an edge at $v$ leading to an arbitrary neighbour of $v$ in the tree $T$; if $v$ has many neighbours, each edge  may cost $O(\log n)$ bits to encode. The trick is to distribute the information regarding $v$'s selected edges among $v$'s children in $T$.
(Recall that $v$ can look at the data structures of $v$'s children.)

The strings mentioned in the previous paragraphs are supposed to define a hierarchy and selected outgoing edges from the fragments of the hierarchy. However, on incorrect instances, if 
corrupted, the strings may not represent the required. For example, the strings may represent more than one selected edge
 for some fragment. Hence, we need also to attach proof labels for verifying  the hierarchy and the corresponding selected edges represented by those strings. Fortunately, for proving the Well-Forming property only, it is not required to verify that the represented hierarchy (and the corresponding selected edges) actually follow the MST construction algorithm (which is the case for correct instances).
  In Section  \ref{sec:section5}, we present  1-proof labeling schemes to show that the above strings represent {\em some} hierarchy with corresponding selected edges, and that the Well-Forming property does hold for that hierarchy.

\subsection{Verifying The Minimality property (described in detail in Sections  \ref{sec:proof},  \ref{sub:3.2} and \ref{sub:utilizing})}\label{intuition:P2}

A crucial point in
the scheme is letting each node $v$ know, for each of its incident edges
$(v,u) \in E$
 and for each level $j$, whether $u$ and $v$ share the same level $j$ fragment.
 Intuitively, this is needed in order to identify outgoing edges.
For that purpose,
we assign each fragment a unique identifier, and $v$ compares the identifier of its own level $j$ fragment  to the identifier of $u$'s level $j$ fragment.

Consider the number of bits required to represent the identifiers of all the fragments that a node $v$ participates in.
There exists a method to assign unique identifiers such that this total number is only $O(\log n)$  \cite{KormanPeleg08,FG01}.
Unfortunately, we did not manage to use that method here. Indeed, there exists a marker algorithm that  assigns identifiers according to that method. However, we could not find a low space and short time
 method for the verifier to verify that the given identifiers of the fragments
were indeed assigned in that way. (In particular, we could not verify  efficiently that the given identifiers are indeed unique).

Hence, we assign identifiers according to
 another method that appears more memory wasteful, where the identity $\id(F)$ of a fragment $F$ is composed of the (unique) identity of its root together with its level.
We also need each node $v$ to know the weight $\omega(F)$ of the minimum outgoing edge of each  fragment $F$ containing~$v$. To summarize, the {\em piece of information} $\Info(F)$ required in each node $v$ per fragment $F$ containing $v$
 is $\id(F)\circ\omega(F)$. Thus, $\Info(F)$ can be encoded using  $O(\log n)$ bits. Again, note that since a node may participate in $\ell=\Theta(\log n)$ fragments,
the total number of bits used for storing all the $\Info(F)$ for all fragments $F$ containing $v$ would thus be
 $\Theta(\log^2 n)$. Had no additional steps been taken, this would have violated the $O(\log n)$ memory constraint.

  Luckily, the total number of bits needed globally for representing the pieces $\Info(F)$ of all the fragments $F$ is only $O(n \log n)$, since there are at most $2n$ fragments, and $\Info(F)$ of a fragment $F$ is of size $O(\log n)$ bits.
  The difficulty results from the fact that multiple nodes need replicas of the same information. (E.g., all the nodes in a fragment need the $\id$ of the fragment.) If a node does not store this information itself, it is not clear how to bring all the many pieces of information to each node who needs them, in a short time (in spite of congestion) and while having only a constant number of such  pieces at a node at each given point in time.

To allow some node $v$ to check whether its neighbour $u$ belongs to $v$'s level $j$ fragment $F_j(v)$ for some level~$j$,
 the verifier at $v$ needs first to reconstruct the piece of information $\Info(F_j(v))$.
 Intuitively, we had to distribute the information, so that $\Info(F)$ is placed ``not too far'' from every node in $F$.
    To compare $\Info(F_j(v))$  with a neighbour $u$, node $v$
   must also obtain $\Info(F_j(u))$ from $u$. This requires some mechanism to synchronize the reconstructions in neighbouring nodes.
   Furthermore, the verifier must be able to overcome difficulties resulting from faults, which can corrupt the information stored, as well as the reconstruction and the synchronization mechanisms.

   The above distribution of the $\Info$'s is described in Section  \ref{sec:proof}. The
    distributed algorithm for the ``fragment by fragment'' reconstruction (and synchronization)   is described in Section
     \ref{sub:3.2}.
     The required verifications for validating the $\Info$'s and comparing the information of neighbouring nodes are described in Section \ref{sub:utilizing}.

\subsubsection{Distribution of the pieces of information ~(described in detail in Section  \ref{sec:proof})}
\label{subsub:overview-proof}
At a very high level description, each node $v$ stores permanently  $\Info(F)$ for a constant number of  fragments~$F$. Using that,
  $\Info(F)$ is ``rotated''  so that each node in $F$ ``sees'' $\Info(F)$ in
  $O(\log n)$  time.
We term  the mechanism that performs this rotation a {\em train}.
A first idea would
have been to have a separate train for each fragment $F$ that would ``carry'' the piece  $\Info(F)$ and would allow all nodes in $F$ to see it.  However, we did not manage to do that efficiently in terms of time and of space.
That is, one train passing a node could delay the other trains that ``wish'' to pass it. Since neighbouring nodes may share only a subset of their fragments, it is not clear how to pipeline
the trains. Hence,
 those delays could accumulate. Moreover,
 as detailed later, each train utilizes some (often more than constant) memory per node.
 Hence, a train per fragment would
 have
 prevented us from obtaining an $O(\log n)$ memory~solution.

A more refined idea  would
have been to
partition the tree into connected parts, such that each part $P$ intersects  $O(|P|)$ fragments.
Using such a partition,
we could have allocated the $O(|P|)$ pieces (of these $O(|P|)$ fragments), so that
 each node of $P$
would have been assigned  only a
 constant number of  such pieces, costing  $O(\log n)$ bits per node.
 Moreover, just one train per part $P$ could have sufficed to rotate those  pieces among the nodes of $P$. Each node in $P$ would have seen all the pieces
  $\Info(F)$  for fragments $F$ containing it in $O(|P|)$ time. Hence, this would have been time efficient too, had $P$ been small.

 Unfortunately, we did not manage to construct  the above partition. However, we managed to obtain a similar construction:
 we construct {\em two} partitions of $T$, called  $\Top$ and $\Bottom$. We also partitioned
  the fragments into two kinds: top and bottom fragments.
  Now, each part $P$ of partition $\Top$ intersects with $O(|P|)$ top fragments (plus any number of bottom fragments).
  Each part $P$ of partition $\Bottom$ intersects with  $O(|P|)$ bottom fragments (plus top fragments that we do not count here).   For each part in $\Top$ (respectively $\Bottom$),
we shall distribute the information regarding the $O(|P|)$ top (respectively, bottom) fragments it intersects with, so that each node would hold at most a constant number of such pieces of information.  Essentially, the pieces of information regarding the corresponding fragments are put in the nodes of the part (permanently) according to a DFS (Depth First Search) order starting at the root of the part.
For any node $v$, the two parts
containing it encode together the information regarding all fragments containing $v$. Thus, to deliver all relevant information, it suffices to utilize one train per part (and hence, each node participates in two trains only). Furthermore, the partitions are made so that the diameter of each part is $O(\log n)$, which allows each train to quickly pass in all nodes, and hence to deliver the relevant information in short time.

The distributed implementation of this distribution of pieces of information, and, in particular, the distributed construction of the two partitions, required us to come up with a new {\em multi-wave} primitive, enabling an efficient (in $O(n))$ time) parallel (i.e., pipelined) executions of Wave\&Echo operations on all fragments of Hierarchy~$\cH_{\cM}$.

\subsubsection{Viewing the pieces of information ~(described in detail in Section \ref{sub:3.2})}
Consider a node $v$ and a fragment $F_j(v)$ of level $j$ containing it. Recall that the information $\Info(F_j(v))$ should reside in some node of a part $P$ to which $v$ belongs. To allow $v$ to compare $\Info(F_j(v))$ to $\Info(F_j(u))$ for a neighbour $u$, both these pieces
must somehow  be ``brought'' to $v$.  The process handling this task contains several components.
The first component is called the {\em train} which is responsible for moving the pieces of information through $P$'s nodes, such that each node does not hold more than $O(\log n)$ bits at a time, and such that in short time, each node in $P$ ``sees'' all pieces, and in some prescribed order.
Essentially, a train is composed of two ingredients. The first ingredient called {\em convergecast} pipelines the pieces of information in a DFS order towards the root of the part (recall, the  pieces of information of the corresponding fragments are initially located according to a DFS order). The second ingredient {\em broadcasts} the pieces from the root of the part to all nodes in the part. Since the number of pieces is $O(\log n)$ and the diameter of the part is $O(\log n)$, the synchronous environment guarantees that each piece of information  is delivered to all nodes of a part in $O(\log n)$ time. On the other hand, in the asynchronous environment some delays may occur, and the delivery time becomes $O(\log^2 n)$. These time bounds are also required to self-stabilize the trains, by known art, see, e.g.~\cite{CollinDolevDFS,BDPV}.

Unfortunately, delivering the necessary pieces of information at each node is not enough, since $\Info(F_j(v))$ may arrive at $v$ at a different time than $\Info(F_j(u))$ arrives at $u$. Recall that $u$ and its neighbour $v$ need to have these pieces simultaneously in order to compare them (to know whether the edge $e=(u,v)$ is outgoing from $F_j(v)$).

Further complications arise
 from the fact that
the neighbours of a node~$v$ may belong to different parts, so different trains pass there.  Note that $v$ may have many neighbours, and we would not want to synchronize so many trains. Moreover,
had we delayed the train at $v$ for synchronization, the delays would have accumulated, or even would have caused deadlocks.
 Hence, we do not delay these trains. Instead, $v$ repeatedly samples a piece from its train, and synchronizes the comparison of this piece with pieces sampled by its neighbours, while both trains advance without waiting.  Perhaps not surprisingly, this synchronization turns out to be easier in synchronous networks, than in asynchronous ones.
Our synchronization mechanism guarantees that each node can compare all  pieces $\Info(F_j(v))$ with $\Info(F_j(u))$ for all neighbours $u$ and levels $j$ in a short time.
Specifically,  $O(\log^2 n)$ time in synchronous environments and $O(\Delta\log^3 n)$ time in asynchronous ones.

\subsubsection{Local verifications~~(described in detail in Section~\ref{sub:utilizing})}
So far, with respect to verifying the Minimality property, we have not discussed issues of faults that may complicate the verification.
Recall, the verification process must detect if the tree is not an MST. Informally, this must hold despite the fact that
the train processes, the partitions,  and also, the pieces of information carried by the trains may be corrupted by an adversary.
For example, the adversary may  change or erase some (or even all) of such pieces corresponding to existing fragments. Moreover,
even correct pieces that correspond to existing fragments may not arrive at a node in the case that the
adversary corrupted the partitions or the train mechanism.

In Section~\ref{sub:utilizing}, we explain how the verifier does overcome such undesirable phenomena, if they occur.
Intuitively, what is detected is not necessarily the fact that a train is corrupted (for example). Instead, what is detected is the state that {\em some} part is incorrect (either the tree is not an MST, or the train is corrupted, or ... etc.).
Specifically, we show that if an MST is not represented in the network, this is detected in time $O(\log^2 n)$ for  synchronous environments and time $O(\Delta\log^3 n)$ for asynchronous ones.
Note that for a verifier, the ability to detect while assuming any
initial configuration
means that the verifier is self-stabilizing, since the sole purpose of the verifier is to detect.

Verifying that {\em some} two partitions exist is easy.
However, verifying that the given partitions are
as described in Section \ref{sec:partitions}, rather  than being two
arbitrary partitions generated by an adversary seems difficult.
Fortunately, this verification turns out to be unnecessary.

 First, as mentioned, it is a known art to self-stabilize the train process.
After trains stabilize,  we  verify that
the set of pieces stored in a part (and delivered by the train) includes
all the (possibly corrupted) pieces of the form $\Info(F_j(v))$, for every $v$ in the part and for every $j$ such that $v$ belongs to a level $j$ fragment. Essentially,
this is done by verifying at the root $r(P)$ of a part $P$, that (1) the information regarding fragments arrives at it in a cyclic order (the order in which pieces of information are supposed to be stored in correct instances), (2) the levels of pieces arriving at $r(P)$  comply with the levels of fragments to which $r(P)$  belongs to, as indicated by $r(P)$'s data-structure. Next, we verify that the time in which each node obtains all the pieces it needs is short. This is guaranteed by the correct train operation, as long as the diameter of parts is $O(\log n)$, and  the number of pieces stored permanently at the nodes of the part is $O(\log n)$. Verifying these two properties is accomplished using a 1-proof labelling scheme of size $O(\log n)$, similarly to the schemes described in Examples 2 and 3 ($\SP$ and $\EDIAM$, mentioned in Section \ref{sub:examples}).

Finally, if up to this point, no node raised an alarm, then for each node $v$, the (possibly corrupted) pieces of information corresponding to $v$'s fragments reach $v$ in the prescribed time bounds. Now, by the train synchronization process,  each node can compare its pieces of information with the ones of its neighbours. Hence, using similar
arguments as was used in the $O(\log^2 n)$-memory  bits verification scheme of \cite{KKP10}, nodes can now detect the case
 that either one of the pieces of information is corrupted or that $T$ is not an MST.

\section{A synchronous MST construction  in $O(\log n)$ bits memory size and $O(n)$ time}\label{sec:mst-construction}
In this section, we describe an MST construction algorithm, called $\ALG$, that is both linear in its running time and memory optimal, that is, it  runs in $O(n)$ time and has $O(\log n)$ memory size. We note that this algorithm is not self-stabilizing and its correct operation assumes a synchronous environment.
The algorithm will be useful later for two purposes. The first is for distributively assigning the labels of the MST proof labelling scheme, as described in the next section. The second purpose, is to be used as a module in the self-stabilizing MST construction algorithm.

As mentioned, the algorithm of Gallager, Humblet, and Spira (GHS) \cite{ghs}  constructs an MST in $O(n\log n)$ time.
This has been improved by Awerbuch to linear time, using a somewhat involved algorithm. Both algorithms are also efficient in terms of the number of messages they send.
The MST construction algorithm described in this section is, basically, a simplification of the GHS algorithm. There are two reasons for why we can simplify that algorithm, and even get a better time complexity.
The first reason is that our algorithm is synchronous, whereas GHS (as well as the algorithm by Awerbuch) is designed for asynchronous environments.
Our second aid is the fact that we do not care about saving messages (anyhow, we use a shared memory model),
while the above mentioned algorithms strive to have an optimal message complexity.
Before describing our MST construction algorithm, we recall the main features of the GHS  algorithm.

 \subsection{Recalling the MST algorithm of Gallager, Humblet, and Spira (GHS)}
 \label{subsec:ghs}

For full details of GHS, please refer to \cite{ghs}. GHS uses connected subgraphs of the final MST, called {\em fragments}.
Each node in a fragment, except for the fragment's root, has a pointer to its parent in the fragment. When the algorithm starts, every node is the root of the singleton fragment including only itself. Each fragment is associated with its
{\em level} (zero for a singleton fragment) and the identity of its root  (this is a slight difference from the description in \cite{ghs}, where a fragment is said to be rooted at an edge).
Each fragment $F$ searches for its {\em minimum outgoing edge} $e_{\min}(F)=(v,u)$.
Using the selected edges, fragments are merged to produce larger fragments
of larger levels. That is, two or more fragments of some level $j$, possibly together with some fragments of levels
lower than $j$, are merged to create a fragment of level $j+1$. Eventually, there remains only one fragment spanning
the graph which is an MST.

In more details, each fragment sends an offer (over $e_{\min}(F))$ to merge with the other fragment $F'$, to which the other endpoint $u$ belongs. If $F'$ is of a higher level, then $F$ is connected to $F'$. That is,
the edges in $F$ are reoriented so that $F$ is now rooted in the endpoint $v$ of $e_{\min}(F)$, which becomes a child of the other endpoint~$u$.

If the level of $F'$ is lower, then $F$ waits until the level of $F'$ grows (see below, the description of ``test'' messages). The interesting case is when $F$ and $F'$ are of the same level $\level$. If $e_{\min}(F)=e_{\min}(F')$, then $F$ and $F'$ merge to become one fragment, rooted at, say, the highest $\id$ node between $u$ and $v$. The level of the merged fragment is set to $\level+1$.

The remaining case, that (w.l.o.g.)
 $w(e_{\min}(F))> w(e_{\min}(F'))$
 does not need a special treatment.
When $F$ sends $F'$ an offer to merge, $F'$ may have sent such an offer to some $F''$ over
$w(e_{\min}(F'))$. Similarly, $F''$ may have sent an offer to some $F'''$ (over $w(e_{\min}(F'')))$, etc. No cycle can be created in this chain of offers (because of the chain of decreasing weights $w(e_{\min}(F))> w(e_{\min}(F'))>w(e_{\min}(F'')) ...)$. Hence, unless the chain ends with some fragment of a higher level (recall that treating the case that a fragment's minimum edge leads to a higher level fragment was already discussed), some two fragments in the above chain merge, increasing their level by one. This case (for the fragments of the chain, excluding the two merging fragments) now reduces to the case (discussed previously) that a fragment $F$ makes an offer to a fragment of a higher level.

The above describes the behavior of fragments. To implement it by nodes, recall that every fragment always has a root. The root conducts Wave\&Echo  over the fragment to ask nodes to find their own {\em candidate} edges for the minimum outgoing edge. On receiving the wave (called ``find''), each node $v$ selects its minimum edge
$(v,u)$ that does not belong yet to the fragment, and has not been ``tested'' yet (initially, no edge was ``tested''). Node~$v$ sends a ``test'' message to $u$, to find out whether $u$ belongs to $v$'s fragment. The ``test'' includes the $\id$ of $v$'s fragment's root $r$ and its level $\level$. If the level
of $u$'s fragment is at least $\level$ then $u$ answers. In particular, if $u$'s level is $\level$ and $u$'s fragment root is $r$ then $u$ sends a ``reject'' to $v$, causing $v$ to conclude that $(v,u)$ is not outgoing and cannot be a candidate.
(Node $u$ does not answer, until its level reaches $\level$, thus, possibly, causing $v$'s fragment to wait.)
      In the converging wave (called ``found'') of the above ``find'' broadcast, each node $v$ passes to its parent only the candidate edge with the minimum weight (among its own candidate and the candidates it received from its children). Node $v$ also remembers a pointer to whoever sent it the above candidate. These pointers form a route from $F's$ root to the endpoint of $e_{\min}(F)$. The root then sends a message ``change-root'', instructing all the nodes on this route (including itself) to reverse their parent pointers. Hence, $F$ becomes rooted at the endpoint of $e_{\min}(F)$, who now can send an offer to ``connect'' over $e_{\min}(F)$.

\subsection{Algorithm $\ALG$: a synchronous linear time version with optimal memory size}\label{subsec:mst-construction}
The algorithm we now describe is synchronous and  assumes that all the nodes wake up simultaneously at round~$0$.
However, to keep it easy for readers who are familiar with GHS, we tried to keep it as similar to GHS as possible.

Initially, each node is a root of a fragment of {\em level} $0$ that contains only itself.
During the execution of $\ALG$, a node who is not a root, keeps a pointer to its parent. The collection of these pointers (together with all the nodes) defines a forest at all times.
Each node also keeps an {\em estimate} of the $\id$ and the level   of the root of its fragment. As we shall see later, the $\id$ estimate is not always accurate. The level estimate is a lower bound on the actual level.
We use the levels for convenience of comparing the algorithm to that of GHS (and for the convenience of the proof). The levels actually can be computed from the round number, or from the counting procedure defined below. More specifically,
the algorithm is performed in synchronous {\em phases}. Phase $i$ starts at round
$11\cdot 2^i$. Each root $r(F)$ of a fragment $F$ (that is, a node whose parent pointer is null) starts the phase by setting its level to $i$ and then counting the number of nodes in its fragment.

 The counting process, called $\COUNTSIZE$, is defined later, but for now it suffices to say that it consumes
{\em precisely} $2^{i+2}-1$ rounds.
 If the diameter of the fragment is small, then some waiting time is added to keep the precise timing.
On the other hand, if the number of nodes in the fragment  is too large, $\COUNTSIZE$ may terminate before all the nodes in the fragment are counted.
 Specifically, we guarantee that if the counting process succeeds to count all nodes in the  fragment $F$ then the precise number of these nodes  is known to the root $r(F)$ at the end of the counting procedure.
 On the other hand, if the counting process does not count all nodes, then
  the number of nodes in the fragment is at least $2^{i+1}$, and at the end of the $\COUNTSIZE$ process, the root $r(F)$ learns this fact. Moreover, in such a case, as a consequence, $r(F)$
   changes its level to $i+1$.

\begin{definition}
\label{def:fragments}
  A root $r(F)$ is {\em active} in phase $i$ if and only if $|F|\leq 2^{i+1}-1$, where $|F|$ denotes the number of nodes in $F$. Note that if $r(F)$ is active then its level is $i$.
 In particular, all the roots are {\em active} in phase (and level)~$0$.
 A fragment is {\em active} when its root is active.
\end{definition}

\begin{Comment}
\label{com:fragments}
When constructing the marker algorithm in later sections, we use the fragments constructed by algorithm $\ALG$. More specifically, we refer only to the active fragments. As is easy to observe below in the current section, an active fragment is a specific set of nodes that does not change. This is because when the fragment merges with others (or when others join it), it is no longer the same fragment. In particular, when the new set of nodes will be active, it will be in a higher phase.
\end{Comment}

\paragraph{Procedure $\FINDOUT$:} Consider the root $r(F)$ of fragment $F$, who is active in phase~$i$.
 At round $(11+4)\cdot 2^i$, each
such root $r(F)$  instructs the nodes in $F$ to search for the minimum outgoing edge of $F$. This procedure, called $\FINDOUT$, could have been combined with the counting, but we describe it as a separate stage for the sake of simplicity. The method is the same as that of GHS algorithm, except that we achieve an exact timing obtained by not saving in messages.
The search is performed over exactly the same set of nodes which has just been counted.
 This is implemented by a Wave operation initiated by $r(F)$, carrying $r(F)$'s $\id$ and level.
At precisely round $(11+6) \cdot 2^i$,
 each node $v$ who has received the wave, finds the minimum outgoing edge emanating from it. That is, $v$ looks at each of its neighbours $u$ to see whether $u$ belongs to a different fragment of some other root $r(F')\neq r(F)$. We now describe how $v$ identifies this.

Let us   note here two differences from GHS. First, node $v$ tests all of its emanating edges at the same time, rather than testing them one by one (as is done in GHS). Moreover, it does not reject any edge, and will test all its edges in the next searches too.
Intuitively, the above mentioned one by one process was used in GHS in order to save messages. We do not try to save messages, and the simultaneous testing allows us to keep an exact timing on which we rely heavily.
 Second, in GHS, node $u$'s estimate of its level may be lower than that of node $v$.
 In GHS, $v$ then needs to wait for $u$ to reach $v$'s level, before $v$ knows whether edge $(v,u)$ is outgoing. The main reason this action is useful in GHS is to  save on message complexity. Here,
again, we do not intend to save messages.

Recall that the root of $v$'s fragment $F$ is active at phase $i$, hence, $|F|< 2^{i+1}$.
(We shall show that no additional nodes joined $F$ after the counting.)
Hence,  at round $(11+6) \cdot 2^i$, {\em all the nodes} in $F$ have already received the wave, and set their $\id$ estimates to $\id(r(F))$.
The big gain from that, is that at round $(11+6) \cdot 2^i$, the $\id$s of the roots of $u$ and $v$ are different if and only if  the edge $(v,u)$ is outgoing at $v$.
The minimum outgoing edge in the fragment of $r(F)$ is then computed during the convergecast, using the standard Wave\&Echo technique.
 Thus, Procedure $\FINDOUT$ (composed of the aforementioned Wave\&Echo) lasts at most $2(2^{i+1}-1)$ round units,
  hence (having been started at round $(11+4) \cdot 2^i$), it is completed by round
  $(11+8) \cdot 2^i-1$.

\paragraph{Merging and reorienting edges:}
Let $(w,x)$ be the chosen minimum outgoing edge from the fragment $F$, such that $w\in F$. Later, we refer to it as the {\em candidate} edge.
At round $(11+8) \cdot 2^i$, an active root $r(F)$ of $F$ starts the process of re-orienting the edges in $F$ towards  $w$. (For a more thorough description of the root transfer refer to \cite{ghs}.)
This takes at most $2(2^{i}-1)$ rounds.

Node $w$ then conducts a handshake with $x$, referred to as the {\em pivot} of $F$.
This takes a constant time, but, to keep the total numbers simple, we pad this time to $2^{i}$.
One case is that $w$  is, at that time exactly, a pivot of the fragment of $x$, and also $\id(x) < \id(w)$. In this case,
node $x$ will become the child of $w$.
In every other case, $w$  hooks upon the other endpoint $x$ (sets its parent pointer to point at $x$). The hooking is performed exactly
at round $(11+11) \cdot 2^i-1$, ending phase $i$.
Since the next phase starts at round  $22\cdot2^i$  there is no overlap between phases.

\paragraph{Procedure $\COUNTSIZE$:} To complete the description of a phase, it is left to describe the counting process, namely, Procedure $\COUNTSIZE$. To count, a root starts a Wave\&Echo, attaching a {\em time-to-live}~=~$2^{i+1}-1$ counter to its broadcast message. A child $c$ of a node $y$ accepts the wave only if the time-to-live is above zero. Child $c$ then copies the wave broadcast message, decrementing the time-to-live (by 1). If, after decrementing, the value of time-to-live is zero, then $c$ is a leaf who needs to start the echo. The number of the nodes (who copied the broadcast message) is now counted during the echo in the standard way.
Finally, if the count covers the whole graph, this can be easily detected at the time of the echo. The algorithm then terminates.

To sum up, phase $i$ of the MST construction algorithm is composed of the following components.
\paragraph{\underline{Phase $i$} }
\begin{smallitemize}
\item Starts at round $11\cdot 2^i$~~;
\item Root $r(F)$ of each fragment $F$ initiates Procedure $\COUNTSIZE$.

At the end of the procedure, we have:\\  $|F|\leq 2^{i+1}-1$~  iff ~
(1) $r(F)$  is {\em active} and
(2) all nodes in $F$ have their $\id$ estimates set to $\id(r(F))$ ~~;
\item At round $(11+4)\cdot 2^i$, each active
root $r(F)$  initiates Procedure $\FINDOUT$~~;
\item
At round $(11+8) \cdot 2^i$, merge fragments and re-orient edges in the newly created fragments.
\end{smallitemize}
\bigskip
 The proof that  the collection of parent pointers forms a forest (or a tree) at all times is the same as in GHS.
 Let us now analyze the round complexity. Observe that each phase $i$ takes $O(2^i)$ time. Hence, the linear time complexity of the algorithm follows from the lemma below.

\begin{lemma}
\label{lem:disappear}
The size of a fragment $F$ in phase $i$ (and in level $i$) is at least $2^i$.  Moreover,  $|F|<2^{i+1}$ if and only if  $r(F)$ is active by round $(11+4)\cdot 2^i$.
\end{lemma}

\begin{proof}
Let us first prove the second  part of the lemma.
Before deciding whether to be active, a root $r(F)$ of level~$i$ counts the number of nodes in its fragment, by employing Procedure $\COUNTSIZE$. If the count amounts to $2^{i+1}$ or more, then the level of $r(F)$ is set to $i+1$. Otherwise, we have $|F|\leq 2^{i+1}-1$ and the root of $F$ becomes active. Since Procedure $\COUNTSIZE$ is terminated by round $(11+4)\cdot 2^i$, the second part of the lemma follows.
To prove the  first part of the lemma, we need to show that the size of a level $i$ fragment  is at least $2^{i}$. We prove this by induction on $i$.

Intuitively, the induction says that each fragment at the beginning of phase $i-1$, is of size at least $2^{i-1}$. During phase $i-1$, by the second part of the lemma, at time
$(11+4)\cdot 2^{i-1}$, all the non-active fragments are already
of size at least $2^{i}$ and are also of level $i$ (as a result of the count). As for active fragments, each such fragment is combined to at least one other fragment, so the resulting size is at least $2\cdot 2^{i-1}=2^{i}$.

In more details,
note that the claim holds for phase $i=0$.
For a larger phase $i$,  assume that the lemma holds for phases up to $i-1$ including. Consider a root $r(F)$ of a fragment $F$ of level $i$. It was a root also at level $i-1$. First, assume that at phase $i-1$, some other root $r(F_1)$ hooked upon $r(F)$'s tree. To do so,  $r(F_1)$ had to be active at phase $i-1$. By the induction hypothesis (the first part of the lemma), the size of fragment $F_1$, as well as the size of fragment of $F$ at that point in time, was at least $2^{i-1}$. The claim, in this case, follows.

Now, assume that no other fragment hooked upon fragment $F$ in phase $i-1$. Note that $F$ at level $i-1$ does not span the graph (otherwise, no root would reach level $i$, by the 2nd part of the inductive hypothesis, and since the counting process on trees is a known art and is known to be correct).
Hence, it has a minimum outgoing edge $e=(w,x)$, where $w\in F$ and $x$ belongs to some other fragment $F_x$.
We claim that the search process $\FINDOUT$ does find that edge $e$. Recall that in the fragment of an active root, the counting reaches all the nodes in the fragment. Hence, each of them knows the $\id$ of their root $r(F)$ at the time the search in their fragment starts. Moreover, since a hooking is performed only at times of the form $(11+11) \cdot 2^j-1$, no new nodes (or fragments) join until the last time step of the phase (which is after the search, because of what we established about the size of the fragment).

We claim that either (a) Edge $e=(w,x)$ was not the minimum outgoing edge of $x$'s fragment $F_x$, or, alternatively, (b) the root $r_x$  of $x$'s fragment $F_x$ had level $i$ at that time, or
(c) $\id(x)>\id(w)$.
Assume the contrary. By the inductive hypothesis, node $r_x$   is at least at level $i-1$. Since we assumed that (b) does not hold, $r_x$ is exactly at level $i-1$.
This means (by the correctness of the counting) that  $r_x$  is active at phase $i-1$.
  Similarly, this also means that the size of $F_x$ is less than $2^i$. Hence, and by the
  induction hypothesis,
  the counting, the searching, the root transfer, and the handshake end in $x$'s fragment at the same time
they end in $w$'s fragment.
(These processes use known art, and we shall not prove them here.)
Then $x$ hooks upon $w$, contrary to our assumption.

We have just established that the conditions for $w$ to hook upon $x$ hold.
Hence, $w$ hooks upon $x$.
Similarly to the previous case, the size of $w$'s fragment,
as well as the size of $x$'s fragments at that point in time is at least $2^{i-1}$. Thus, the size of the combined fragment is at least $2^i$. This concludes the proof of the first part of the lemma.
\end{proof}

\begin{corollary}\label{cor:time-linear}
The synchronous MST construction algorithm computes an MST in time $O(n)$.
\end{corollary}


 \paragraph{Implementing the algorithm in the shared memory model with $O(\log n)$ memory:}
 Each node $v$ keeps its fragment level and root $\id$. Node $v$ also remembers whether $v$ is in the stage of counting the number of nodes, or searching for the minimum outgoing edge. It also needs to remember whether it is in the wave stage, or has already sent the echo. Node $v$ needs to remember the candidate  (for being $e_{\min}(F))$ edge that $v$  computed in the echo  (``found'') stage of the convergecast. If this candidate was reported by a child, then $v$ also remembers a pointer to that child.
 Clearly, all the above variables combined need $O(\log n)$ bits of memory.

At a first glimpse it may look as if
 a node must also remember the list of pointers to its children. The list is used for (1) sending the wave (e.g., the ``find'' message of the search, using GHS terms) to all the children, and (2) knowing when all the children answered the echo
 (e.g., the ``found'' of the search).
 Note that a node does not need to store this list itself. Node $v$ can look at each neighbour
 $u$ to see whether the neighbour is $v$'s child. (For that purpose, if $u$ is a child of $v$, then $u$ stores $v$'s $\id$ as $u$'s $\PARENTID$.) Clearly, this can be implemented using $O(\log n)$ bits per node.

To implement (1), node $v$ posts its wave broadcast message (e.g., the ``find'') so that every neighbour can read it. However, only its children actually do. To allow the implementation of (2), a precaution is taken before the above posting. Node $v$ first posts a request for its children to reset their $\ECHO$ variables, and performs the posting of ``find'' only when it sees that $\ECHO$ has been reset for every neighbour $w$ whose parent pointer points at $v$.

To implement (2), node $v$ further reads its neighbours repeatedly. It knows all its children echoed the wave in an iteration when it has just finished rereading all its neighbours, and every node $u$ pointing at $v$ (as its parent), also sets its $\ECHO$ variable to some candidate edge (or to some default value if it has no candidate edge).

\begin{observation}
\label{obs:spce-sync}
The space required by the linear time synchronous algorithm is $O(\log n)$ bits per node.
\end{observation}
The theorem below follows from Observation~\ref{obs:spce-sync} and Corollary \ref{cor:time-linear}.

\begin{theorem}\label{thm:ALG}
The synchronous algorithm $\ALG$ computes an MST in time $O(n)$ and memory size $O(\log n)$.
\end{theorem}

\bigskip

\section{Representing and verifying a hierarchy}
\label{sec:section5}
We are now ready to describe our proof labeling scheme $(\cM,\cD)$ for MST. The goal of this section is to construct some part of the marker $\cM$, and the corresponding part of the verifier $\cD$, which are relatively easy to construct.
 The techniques used in this section
bear similarity to the techniques presented in \cite{KKP10}.
  Hence, we only expose the main ideas behind this part of the proof labeling scheme, leaving out some of the technicalities.
Nevertheless, since the notion of proof labeling schemes can sometimes be confusing, this section may help the reader to get accustomed to the notion
and the difficulties that may arise.

As a warm up,
we
 first note that using the 1-proof labeling scheme described in Example $\SP$, we may  assume that $H(G)\equiv T$ is a spanning tree of $G$ rooted at some node $r$, and
that each node knows which
of its neighbours in $G$ are its children in $T$ and which is its parent.
 Moreover, using the 1-proof labeling scheme described in Example $\NK$, we may also assume that each node knows $n$. The 1-proof labeling schemes described in Examples $\SP$ and $\NK$ use $O(\log n)$
 memory size and can be constructed using $O(n)$ time.
 Hence, using them does not violate the desired  complexity constrains of our scheme.
Thus, from now on, let us fix  a spanning tree $T=(V(G),E(T))$ of a
graph $G=(V(G),E(G))$,  rooted at some node $r(T)$.
The goal of the rest of the verification scheme is to verify that $T$ is in fact, minimal. Before we continue, we need a few definitions.

\begin{definition}
\label{def:hierarchy}
A {\em hierarchy} $\cH$ for $T$
is a collection of fragments of $T$ satisfying the following two properties.
\begin{enumerate}
\item
$T\in\cH$ and,
for every $v\in V(G)$,
there is an $F_v\in\cH$ such such
$V(F_v) = \{v\}$ and $E(F_v) = \emptyset$.
\item For any two fragments $F$ and $F'$ in $\cH$, if $F\cap F'\neq \emptyset$ then either $F\subseteq F'$ or $F'\subseteq F$. (That is, the collection of fragments is a laminar family.)
\end{enumerate}
\end{definition}
Please recall (Definition \ref{com:fragments} and Comment \ref{com:fragments})
that when we construct a hierarchy according to Definition \ref{def:hierarchy}, the fragments referred to are the active fragments constructed in $\ALG$.

The {\em root} of a fragment $F$ is the node in $F$ closest to the root of $T$.
For a fragment $F\in\cH$, let $\cH(F)$ denote the collection of
fragments in $\cH$ which are  {\em strictly} contained  in $F$.
Observe that a
hierarchy $\cH$ can be viewed as a rooted tree, whose root is  the
fragment $T$, and whose leaves are the singleton  fragments in
$\cH$.
A child of a non-singleton fragment $F\in \cH$ is a fragment
$F'\in\cH(F)$ such that no other fragment $F''\in \cH(F)$ satisfies
$F''\supset F'$. Note that the rooted tree induced by a hierarchy is
unique (if the children are unordered). To avoid confusion with tree
$T$, we use the name {\em hierarchy-tree} (or, sometimes even just {\em hierarchy})  for the above mentioned tree induced by a hierarchy. We associate  a {\em level}, denoted  $lev(F)$, with each fragment
$F\in \cH$. It is defined as the height of the node corresponding to
$F$ in the hierarchy-tree induced by $\cH$, i.e., the maximal number of fragments on a simple path in $\cH$ connecting $F$ to a singleton fragment. In particular, the level
of a singleton fragment is 0. The level of the fragment $T$ is
called the {\em height} of the hierarchy, and is denoted by $\ell$. Figure \ref{fig:HierarchyTree1} depicts a hierarchy $\cH$ of a tree $T$.

\begin{figure}
\centering
\includegraphics[scale=.8]{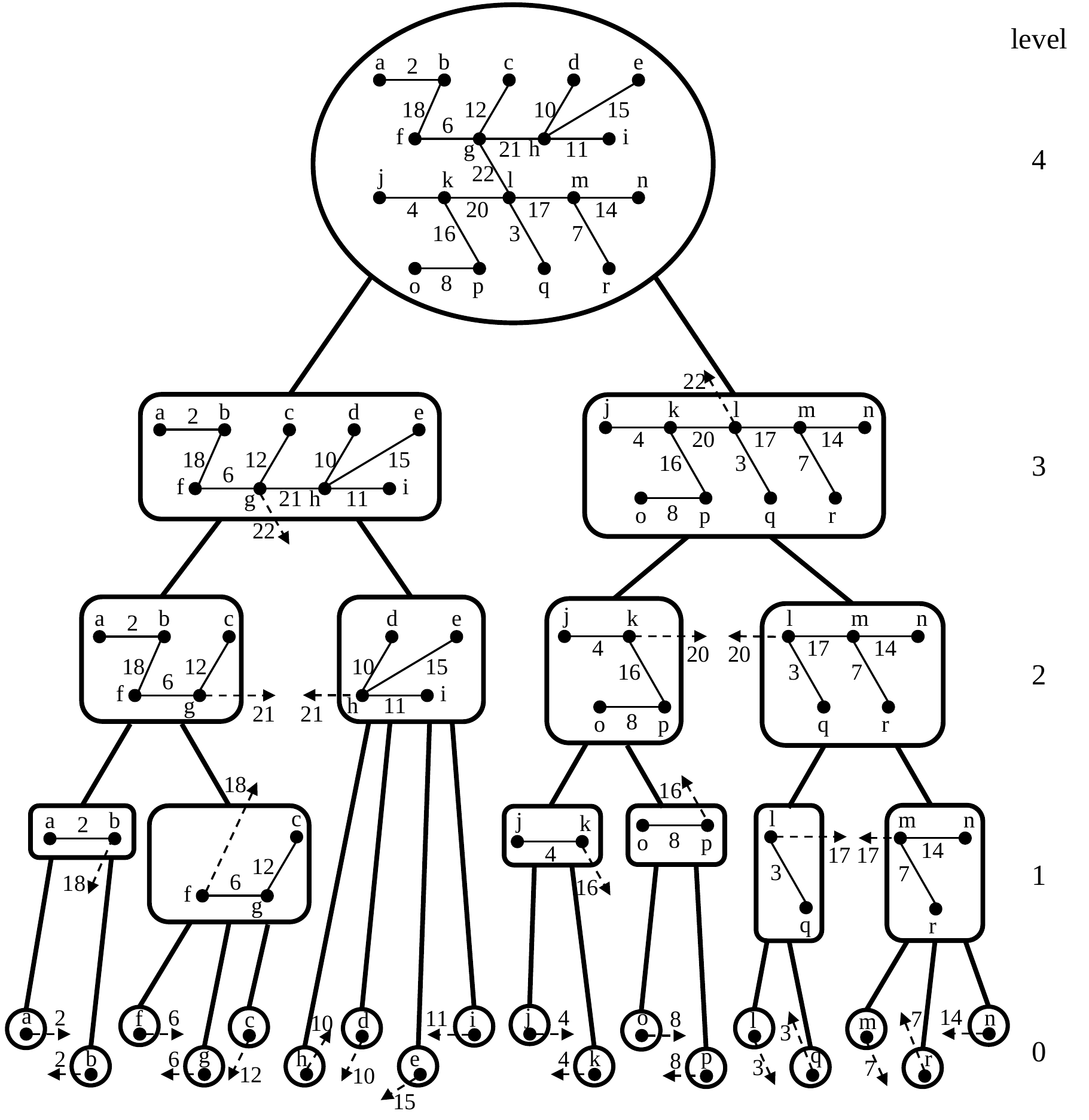}
\caption{A hierarchy $\cH$ of a tree $T$. The root node of $\cH$ represents $T$ (where non-tree edges are omitted). Each fragment that is not a leaf fragment is a parent, in the hierarchy, of the fragments that were merged to form it.
The broken arrow from each fragment is the outgoing edge of the fragment that is used to form a higher level (parent) fragment.}
\label{fig:HierarchyTree1}
\end{figure}

\begin{definition}
Given a hierarchy  $\cH$ for a spanning tree $T$, a function
$\chi:\cH\setminus\{T\} \longrightarrow E(T)$
is called a {\em candidate function} of  $\cal{H}$ if it satisfies $E(F)=\{ \chi(F') | F'\in \cH(F)   \}$ for every $F\in\cH$.
 (Less formally, $F$ is precisely the union of the candidate edges $\chi(F')$ of all fragments $F'$ of $\cH$ strictly contained in $F$).
\end{definition}
\bigskip
The proof of the following lemma is similar, e.g., to the proof of   \cite{ghs}.
\begin{lemma}\label{lem:construction}
Let $T$ be a spanning tree of a graph $G$. If there exists a candidate function $\chi$ for a hierarchy $\cH$ for~$T$,
such that for every $F\in \cH$, the candidate edge $\chi(F)$ is a minimum outgoing edge from $F$, then $T$ is an MST of $G$.
\end{lemma}
 \begin{proof}
  We prove the claim that
  each fragment $F\in\cH$ is a subtree of an MST of $G$, by induction on the level $lev(F)$ of fragment $F$. Note that the
claim obviously holds for any fragment $F$ with $lev(F)=0$ since $F$ is a
 singleton fragment.

 Now consider a fragment $F$ with $lev(F)=k$ under the inductive assumption
 that the claim holds for every fragment $F'$ with $lev(F')<k$.   Let
 $F_1, F_2, \ldots , F_a$ be the child fragments of $F$ in $\cH$.
 Since for each $i\in [1, a]$, fragment $F_i$ satisfies $lev(F_i)<k$, the induction hypothesis implies that $F_i$
 is a subtree of the MST.  It also follows from the facts that
 $E(F)=\{ \chi(F') | F'\in \cH(F)   \}$
  and
   $E(F_i)=\{ \chi(F') | F'\in \cH(F_i)   \}$
   for each
  $i\in [1, a]$ that the fragment $F$ is obtained by connecting
  $F_1, F_2, \ldots , F_a$ with their minimum outgoing edges.
In the case that a fragment $F'$ is a fragment of an MST (as is the case here for $F_1, F_2, \ldots , F_a$, by the induction hypothesis), it is known that
  the union of $E(F')$ with the minimum outgoing edge of $F'$  is a fragment of the MST (the ``safe edge'' theorem).
   (See e.g.,
  \cite{even-book}.)
   Thus, fragment $F$, which is obtained by connecting fragments $F_1,F_2, \ldots, F_a$ with their minimum outgoing edges, is a subgraph of an
  MST.
  \end{proof}
\bigskip
\bigskip

Informally, suppose that we are given
 distributed structures that are claimed to be a tree $T$, a ``legal'' hierarchy~$\cH$ for the tree, and
 a ``legal'' candidate function for the hierarchy.
 The goal obtained in the current section is to
 verify the following properties of
 these structures.
 First, verify
 that this indeed is a hierarchy for $T$
of height $\ell\leq \lceil \log n\rceil$ and a candidate function $\chi$ for $\cH$.
 Moreover,
 verify that each node $v$ ``knows''
to which levels of fragments $v$ belongs and which of its neighbours in $T$
share the same given fragment. (Note that
this section does not guarantee that knowledge for neighbours in $G$ who are not
neighbours in $T$.)
In addition, each node is verified to ``know'' whether
it is adjacent to a candidate edge of any of the fragments it belongs to. Put more formally, this section establishes the following lemma.

\bigskip
\begin{lemma}\label{lem:simple-proof}
There exists a 1-proof labeling scheme with memory size $O(\log n)$ and construction time $O(n)$ that verifies the following:
 \begin{smallitemize}
 \item
$H(G)\equiv T$ is a spanning tree of $G$ rooted at some node $r$, and each node knowns $n$.

\item
 The cartesian product of the data-structures indeed implies a  hierarchy  $\cH$ for $T$ of height
   $\ell\le\lceil \log n\rceil$
  and a candidate function  $\chi$ for $\cH$. Furthermore, the data-structure at each node $v$ allows it to know,
\begin{smallitemize}

\item
Whether  $v$ belongs to a fragment $F_j(v)$ of level $j$ in $\cH$ for each $0\leq j\leq \ell$, and if so:
\begin{smallitemize}
\item
 Whether $v$ is the root of $F_j(v)$. 
\item
Whether $v$ is an endpoint of the
(unique) candidate edge of $F_j(v)$, and if so, which of the edges adjacent to $v$ is the candidate edge.
\end{smallitemize}
\item
Given the data-structure of a node $u$ which is a neighbour of $v$  in $G$, i.e., $(v,u)\in E(G)$, node $v$    can find out whether they are neighbours in $T$ as well, i.e., whether $(u,v)\in E(T)$, and if so, for  each $1\leq j\leq \ell$, whether $u$ belongs to $F_j(v)$.
\end{smallitemize}
\end{smallitemize}
 \end{lemma}

\subsection{Hierarchy $\cH_{\cM}$ and candidate function $\chi_{\cM}$}
On a correct instance, i.e., when $T$ is indeed an MST, the marker $\cM$ first  constructs a particular hierarchy $\cH_{\cM}$ over $T$ and a candidate function $\chi_{\cM}$ for that hierarchy.
Hierarchy  $\cH_{\cM}$ and candidate function $\chi_{\cM}$ are designed so that indeed each candidate of a fragment  is a minimum outgoing edge from that fragment. The marker then encodes
hierarchy  $\cH_{\cM}$ and candidate function $\chi_{\cM}$ in one designated part of the labels using $O(\log n)$ bits per node.
Note, however, that these bits of information  may be corrupted by the adversary. We will therefore need to employ another procedure that verifies that indeed a hierarchy
 $\cH$ and a candidate function $\chi$ are represented by the cartesian product of the encodings of all nodes.
By Lemma \ref{lem:construction}, it is not necessary that
the verifier checks that $\cH$ is, in fact, the particular hierarchy $\cH_{\cM}$ constructed by the marker, or that  the candidate function $\chi$ is $\chi_{\cM}$.
However, as is clear from the same lemma, we do need to show that  $\cH$ and candidate function $\chi$ satisfy that indeed each candidate of a fragment  is a minimum outgoing edge from that fragment.
This task is the main technical difficulty of the paper, and is left for the following  sections.

The  hierarchy $\cH_{\cM}$ and Candidate function $\chi_{\cM}$ built by the marker algorithm are  based on  $\ALG$, the new MST construction algorithm
described in Section~\ref{sec:mst-construction}. Since we assume that the MST is unique, Algorithm $\ALG$ will in fact construct the given MST.
(Recall that we describe here the labels assigned by the marker to a correct instance, where the given subgraph $T$ is indeed an MST.)
The hierarchy and candidate function we define for $T$  follow the merging of {\em active} fragments in algorithm
 $\ALG$. More precisely, the nodes in $\cH_{\cM}$ are the active fragments defined during the execution of $\ALG$. Recall from Section \ref{sec:mst-construction}, that an active fragment $F$ joins some fragment $H$ of $T$, through its minimal outgoing edge $e$.
(It is possible that at the time $F$ joins $H$, $H$ itself was an active fragment that joined $F$ through its own minimal outgoing edge that is also $e$.)
Note that with time, some other fragments join the resulted connected component, until, at some point, the resulted connected component becomes an active fragment $F'$. In the hierarchy tree $\cH_{\cM}$, fragment $F$ is defined as the child of $F'$, and the candidate edge  of $F$ is $e$, i.e., $\chi(F)=e$.

As proved in Lemma~\ref{lem:disappear}, after performing the algorithm for level $i$, the size of every fragment is at least $2^{i}$.
Thus, in particular, the  height of the hierarchy $\cH$ is at most $\lceil\log n\rceil$.
 The  candidate function $\chi_{\cM}$ chosen by the marker for
 $\cH_{\cM}$ is defined by the minimum outgoing edges selected by the algorithm, i.e., for each $F\in\cH_{\cM}$, the candidate edge $\chi(F)$ is the selected edge of $F$.
 Thus, under $\chi_{\cM}$, each candidate of a fragment is, actually,
a minimum outgoing edge.

\subsection{Representing  a hierarchy distributively and verifying it locally}
\label{sec:roots}
\label{app:hierarchy}

\paragraph{Representing a hierarchy:}
 Let
$\ell \le \lceil\log n\rceil$.
Given a hierarchy of fragments $\cH$ of height $\ell$ over the rooted  tree $T=H(G)$, we now describe how we represent it in a distributed manner.
Each node $v$ keeps a string named $\ROOTS(v)$ of length $\ell+1$, where each entry in that string is either ``1'', ``0'', or ``*''. To be consistent with the levels, we enumerate the entries of each string from left to right, starting at position 0, and ending at position $\ell$.
 Fix $j\in[0,\ell]$.
Informally, the $i$'th entry of $\ROOTS(v)$, namely, $\ROOTS_i(v)$, is interpreted as follows.
\begin{smallitemize}
\item $\ROOTS_i(v)=1$ indicates that $v$ is the root of the level $i$ fragment it belongs to.
\item $\ROOTS_i(v)=0$ indicates that $v$ is not the root of the level $i$ fragment it belongs to.
\item $\ROOTS_i(v)=*$ indicates that there is no level $i$ fragment that $v$ belongs to.
\end{smallitemize}
See Table \ref{tab:table1} for an example of $\ROOTS$ strings of nodes corresponding to
Figure \ref{fig:HierarchyTree1}.

\begin{table}
\begin{tabular}{|c|ccccc|p{1cm}|c|ccccc|}
\cline{1-6}
\cline{8-13}
\texttt{Roots}  & 0 & 1 & 2 & 3 & 4 & &
\texttt{EndP}   & 0 & 1 & 2 & 3 & 4\\
\cline{1-6}
\cline{8-13}
a & 1 & 0 & 0 & 0 & 0 & & a & up & none & none & none & none\\
b & 1 & 1 & 0 & 0 & 0 & & b & down & up & none & none & none\\
c & 1 & 0 & 0 & 0 & 0 & & c & up & none & none & none & none\\
d & 1 & * & 0 & 0 & 0 & & d & up & * & none & none & none\\
e & 1 & * & 0 & 0 & 0 & & e & up & * & none & none & none\\
f & 1 & 0 & 0 & 0 & 0 & & f & up & down & none & none & none\\
g & 1 & 1 & 1 & 1 & 0 & & g & down & none & down & up & none \\
h & 1 & * & 1 & 0 & 0 & & h & down & * & up & none & none\\
i & 1 & * & 0 & 0 & 0 & & i & up & * & none & none & none \\
j & 1 & 0 & 0 & 0 & 0 & & j & up & none & none & none & none\\
k & 1 & 1 & 1 & 0 & 0 & & k & down & down & up & none & none\\
l & 1 & 1 & 1 & 1 & 1 & & l & down & down & down & down & none\\
m & 1 & 1 & 0 & 0 & 0 & & m & down & up & none & none & none\\
n & 1 & 0 & 0 & 0 & 0 & & n & up & none & none & none & none\\
o & 1 & 0 & 0 & 0 & 0 & & o & up & none & none & none & none\\
p & 1 & 1 & 0 & 0 & 0 & & p & down & up & none & none & none\\
q & 1 & 0 & 0 & 0 & 0 & & q & up & none & none & none & none\\
r & 1 & 0 & 0 & 0 & 0 & & r & up & none & none & none & none\\
\cline{1-6}
\cline{8-13}
\end{tabular}

\vspace{4mm}

\begin{tabular}{|c|ccccc|p{1cm}|c|ccccc|}
\cline{1-6}
\cline{8-13}
\texttt{Parents}  & 0 & 1 & 2 & 3 & 4 & &
\texttt{Or-EndP}  & 0 & 1 & 2 & 3 & 4\\
\cline{1-6}
\cline{8-13}
a & 1 & 0 & 0 & 0 & 0 & & a & 1 & 0 & 0 & 0 & 0\\
b & 0 & 1 & 0 & 0 & 0 & & b & 1 & 1 & 0 & 0 & 0\\
c & 0 & 0 & 0 & 0 & 0 & & c & 1 & 0 & 0 & 0 & 0\\
d & 1 & 0 & 0 & 0 & 0 & & d & 1 & 0 & 0 & 0 & 0\\
e & 0 & 0 & 0 & 0 & 0 & & e & 1 & 0 & 0 & 0 & 0\\
f & 1 & 0 & 0 & 0 & 0 & & f & 1 & 1 & 0 & 0 & 0\\
g & 0 & 0 & 0 & 1 & 0 & & g & 1 & 1 & 1 & 1 & 0\\
h & 0 & 0 & 1 & 0 & 0 & & h & 1 & 0 & 1 & 0 & 0\\
i & 0 & 0 & 0 & 0 & 0 & & i & 1 & 0 & 0 & 0 & 0\\
j & 1 & 0 & 0 & 0 & 0 & & j & 1 & 0 & 0 & 0 & 0\\
k & 0 & 0 & 1 & 0 & 0 & & k & 1 & 1 & 1 & 0 & 0\\
l & 0 & 0 & 0 & 0 & 0 & & l & 1 & 1 & 1 & 1 & 0\\
m & 0 & 1 & 0 & 0 & 0 & & m & 1 & 1 & 0 & 0 & 0\\
n & 0 & 0 & 0 & 0 & 0 & & n & 1 & 0 & 0 & 0 & 0\\
o & 1 & 0 & 0 & 0 & 0 & & o & 1 & 0 & 0 & 0 & 0\\
p & 0 & 1 & 0 & 0 & 0 & & p & 1 & 1 & 0 & 0 & 0 \\
q & 1 & 0 & 0 & 0 & 0 & & q & 1 & 0 & 0 & 0 & 0\\
r & 1 & 0 & 0 & 0 & 0 & & r & 1 & 0 & 0 & 0 & 0\\
\cline{1-6}
\cline{8-13}
\end{tabular}
\caption{\texttt{Roots, EndP, Parents} and \texttt{Or-EndP} for Figure \ref{fig:HierarchyTree1}.}
\label{tab:table1}
\end{table}

\paragraph{Verifying a hierarchy:}
Observe, the  $\ROOTS$ strings assigned for a correct instance satisfy the following.\\

\noindent{\bf The $\ROOTS$ strings (RS) conditions:}
 \begin{smallitemize}
\item {\bf (RS0)} The prefix of the $\ROOTS$ string at every node is in [1,*]$^*$ and its suffix is in [0,*]$^*$,\\
(*because each node is a root of a level 0 fragment and continues being a root in its fragment  until some level when it stops (if it does stop); when the node stops being a root, it never becomes a root again*)
\item {\bf (RS1)} the length of each  $\ROOTS$ string is $\ell+1$, \\
(*because there cannot be more than $\ell+1$ levels *)
\item {\bf (RS2)} the $\ROOTS$ string of the root $r$ of $T$ contains only ``1''s and ``*''s,  and its  $\ell$'th entry  is ``1'',\\
(*because a zero in the $i$th position would have meant that $r$ is not the root of its fragment of level $i$; the second part follows from the fact that the whole tree is a fragment of level $\ell$ and $r$ is its root *)
\item {\bf (RS3)} the first entry (at position 0) of every $\ROOTS$ string is  ``1'',\\
 (*because every node $v$ is the root of a singleton fragment
  of level 0 containing only node $v$ *),
  \item
{\bf (RS4)} the $\ell$'th entry of every non-root node is ``0'',\\
 (*because only $r$ is the root of a fragment of level $\ell$, since that fragment is the whole tree *)
\item
{\bf (RS5)} if the $j$'th entry of $\ROOTS(v)$ is ``0'' for some node $v$ and $j\in[0,\ell]$, then the $j$'th entry of the $\ROOTS$ strings
of $v$'s parent in $T$ is not ``*''.\\
(*because  node $v$ belongs to a fragment $F$ of level $j$, but is not $F$'s root; hence, $v$'s parent belongs to $F$ of level $j$ too *)

\end{smallitemize}
It is easy to see that for any   assignment of $\ROOTS$ strings  $\cI$  obeying rules RS1--RS5
there exists a unique hierarchy  whose distributed representation is $\cI$.
 Hence, we say that an assignment of $\ROOTS$ strings to the nodes of $T$ is {\em legal} if
 the strings obey the
five
 $\ROOTS$ strings conditions above, namely
  RS1--RS5.
 For a given legal assignment of $\ROOTS$ strings
 $\cI$, we refer to its induced hierarchy as the {\em $\ROOTS$-hierarchy} of $\cI$.
  Recall that at this point, we may assume that each node $v$ knows the value of  $n$, 
and that each node knows whether it is the root of $T$. Hence,
verifying that a given assignment of  $\ROOTS$ strings is a legal one can be done locally, by letting each node look at its own string
and the string of its parent only.

\paragraph{Identifying tree-neighbours in the same fragment:}
Obviously, for correct instances, the marker produces a legal assignment of $\ROOTS$ strings.
For a general instance, if the verifier at some node finds that
the assignment of $\ROOTS$ is not legal then it raises an alarm. Thus, (if no node raises an alarm), we may assume that hierarchy {\em $\ROOTS$-hierarchy}
exists, and that each node knows  (by looking at its own label and the labels of its neighbours in the tree $T)$, for every level $0\leq j\leq \ell$,
\begin{enumerate}
\item  whether it belongs to a fragment
$F_j$ of level $j$, and if so:
\item  which of its  neighbours in $T$ belongs to $F_j$.
\end{enumerate}

\subsection{Representing and verifying a candidate function for the $\ROOTS$-hierarchy}\label{sec:candidate}
Having discussed the proof labeling for the hierarchy, we now describe the proof labeling scheme for the candidate function.
Consider now a correct instance $G$ and the hierarchy $\cH_{\cM}$ produced by the
marker algorithm.  Recall, the  candidate function $\chi_{\cM}$ is given by the selected outgoing edges, which are precisely the minimum outgoing edges of the corresponding fragments, as identified by the construction algorithm $\ALG$. We would like to represent  this candidate function $\chi_{\cM}$ distributively, and to verify that this representation indeed forms a candidate function. Moreover, we would make sure that each node $v$ be able to know, for each fragment $F$ containing it, whether it is an endpoint
of the selected edge of $F$, and if so, which of its edges is the selected edge.

\paragraph{Representing a candidate function:}

Given a correct instance, and its corresponding legal assignment of $\ROOTS$ strings, we augment it by adding, for each node $v$, an additional string of  $\ell+1$ entries named $\EndP(v)$.
Intuitively, $\EndP(v)$ is used by the marker algorithm to mark the levels of the fragments for which $v$ is the endpoint of the minimum outgoing edge. Moreover, in a sense, $\EndP(v)$ also is a part of the marking of the specific edge of $v$ that is the minimum outgoing edge in that level (in the case that $v$ is indeed the endpoint). Let us now give the specific definition of that marking.

 Each entry  in $\EndP(v)$ is one of four symbols, namely,
``up'', ``down'', ``none'' and ``*''.
The entries of  $\EndP(v)$
are defined as follows.
Fix an integer $j\in [0,\ell]$ and a node $v$.
If $v$ does not belong to a fragment of level $j$ in $\cH_{\cM}$, then the $j$'th entry in $\EndP(v)$ is ``*''.
Consider now $j\in [0,\ell]$ such that $v$  does belong to a fragment $F\in \cH_{\cM}$ of level $j$.
If $v$ is not an endpoint of the candidate $\chi_{\cM}(F)$ of $F$,
then the $j$'th entry of  $\EndP(v)$ is ``none''. Otherwise,
node $v$ is an endpoint of $\chi_{\cM}(F)$, i.e.,
$\chi_{\cM}(F)=(v,u)$  (for some $u$ that is not in $F)$.
Consider two cases. If $u$ is $v$'s parent in $T$ then the $j$'th entry of  $\EndP(v)$ is set to ``up''.
If, on the other hand, $u$ is a child of $v$ in $T$, then the $j$'th entry of  $\EndP(v)$ is set to ``down''. See Table 1 for an example of $\EndP$ strings of nodes corresponding to
Figure \ref{fig:HierarchyTree1}.

Consider now  a node $v$ that belongs to a level $j$ fragment $F\in \cH_{\cM}$. By inspecting its own label, node $v$  can find out
whether it is an endpoint of a candidate of $F$ (recall, from the previous subsection, that it also knows whether or not it belongs to a level $j$ fragment).
Moreover, in this case, we would like
$v$ to actually be able to identify in one time unit, which of its incident (tree) edges is the candidate.
Obviously, if  the $j$'th entry in $\EndP(v)$ is ``up'', then
the candidate $e$ is the edge leading from $v$ to its parent in $T$.
Intuitively,
in the case that the entry is ``down'', we would like to store this information in $v$'s children to save space in $v$ (since $v$ may be the endpoint of minimum outgoing edges for several fragments, of several levels, and may not have enough space to represent all of them).
 Hence, we attach to each node~$x$ another
string called $\PARENTS(x)$, composed also of $\ell+1$ bits.
For $j\in[0,\ell]$, the $j$'th bit in $\PARENTS(x)$ is ``1''
 if and only if $(y,x)$ is the candidate of the level $j$ fragment that contains $y$ (if one exists), where $y$ is the parent of~$x$.
See Table 1 for an example of $\PARENTS $ strings of nodes corresponding to
Figure \ref{fig:HierarchyTree1}.
Now, to  identify~$u$, node $v$ needs only
to inspect the $\PARENTS$ strings of its children.
In either of the above cases for the $\EndP(v)$ entry (``up'' or ``down''), we call $e$ the  {\em induced candidate} of $F$.

\paragraph{Verifying a candidate function:} Given a legal assignment of $\ROOTS$ strings, we say that  assignments of $\EndP$ and  $\PARENTS$ strings are {\em legal} if the following conditions hold:

\begin{smallitemize}
\item {\bf (EPS0)}
If the $j$'th entry of $\PARENTS(v)$ is ``1''
 and $u$ is the parent of $v$, then the $j$'th entry of $\EndP(u)$ is ``down'', \\
 (* because if $v$ indicates the minimum outgoing edge of $u$'s fragment (of level $j$) leads from $u$ to $v$, then $v$'s parent $u$ indicates this edge leads to one of $u$'s children *)
\item {\bf (EPS1)}
for each fragment $F$ of level $0\leq j<\ell$ in the $\ROOTS$-hierarchy, there exists precisely one node $v\in F$ whose $j$'th entry in $\EndP(v)$
is either ``up'' or ``down'',\\
(*because  only one node $v$ in each fragment $F$ of level $j$ is the endpoint of the outgoing edge of $F$ *)
\item
{\bf (EPS2)} for each node $v$, if the $j$'th entry in $\EndP(v)$ string is ``down'' then there exists
precisely one child $u$ of $v$ such that the  $j$'th entry in  $\PARENTS(u)$  is ``1'',\\
(*because the j'th entry in $\EndP(v)$ being ``down'' indicates its minimum outgoing edge leads to {\em one} of $v$'s children (only {\em one}, since there is only one minimum outgoing edge of the fragment $F$ of level $j$ containing~$v$); to remember which child, we mark this child $u$ by $\PARENTS(u)=1$ *)
 \item {\bf (EPS3)}
 for each node $v$, and for each $0\leq j< \ell$, if the $j$'th entry in $\EndP(v)$ string is ``up'' then:
\begin{enumerate}
\item
 the  $j$'th entry
of $v$'s $\ROOTS$-string is  ``1'',\\
(*because node $v$ belongs to a different fragment $F_v$ of level $j$ than the level $j$ fragment of $v$'s parent; hence, $v$ is the highest (closest to the root of the whole tree) in $F_v$, that is, $v$ is $F_v$'s root *)
\item\label{item-contain}
 for every $i>j$, the  $i$'th entry
of $v$'s $\ROOTS$-string is not ``1'',\\
(*because fragment $F_v$ of $v$ in level $j$ merges with the fragment (of level $j$) of $v$'s parent; hence, $v$ is not the highest in its fragments of levels $i>j$*)
\end{enumerate}
\item {\bf (EPS4)}
 if  the  $j$'th entry in  $\PARENTS(v)$  is ``1''  then:
\begin{enumerate}
\item the  $j$'th entry of $v$'s $\ROOTS$-string is  not ``0'' ,\\
(*because node $v$ is not in the fragments of level $j$ of $v$'s parent (see EPS2); hence, either $v$
is the root of its fragment of level $j$ (see EPS3, part 1), or $v$ does not belong to a fragment of level $j$ *)
\item  for every $i>j$, the  $i$'th entry
of $v$'s $\ROOTS$-string is not ``1'',\\
(*See EPS3 part 2 *)
\end{enumerate}
\item\label{item-every} {\bf (EPS5)}
for every non-root node $v$, there exists an index integer $j\in[0,\ell]$, such that either the $j$'th entry in $\PARENTS(v)$ is 1
or the  $j$'th entry in $\EndP(v)$ is ``up''.\\
(*because  every node is the root of a fragment of level 0; at some level, $v$'s fragment merges with the fragment of $v$'s parent *)
\end{smallitemize}

\begin{lemma}\label{lem:correct-candidate}
Consider a $\ROOTS$-hierarchy $\cH$ given by a legal assignment of $\ROOTS$ strings.
The conditions  EPS1--EPS5  above imply that
 legal assignments of $\EndP$ and $\PARENTS$ strings
(with respect to  $\cH$) induce a candidate function
$\chi:\cH\setminus\{T\} \longrightarrow E(T)$.
\end{lemma}

\begin{proof}
Condition EPS1 implies that for each fragment $F\neq T$, there is precisely one node ``suspected'' as an endpoint of the induced candidate of $F$.
Condition EPS2 together with the previous one, implies that there is precisely one induced candidate edge $\chi(F)$ for each fragment $F\neq T$.
That is, these two conditions induce a function  $\chi:\cH\setminus\{T\} \longrightarrow E(T)$. Our goal now is to show that $\chi$ is, in fact, a candidate function. That is, we need to show that
 for every fragment $F\in\cH$,
we have $E(F)=\{ \chi(F') | F'\in \cH(F)   \}$.
(Recall, $\cH(F)$ denotes the set of fragments in $\cH$ which are strictly contained in $F$.)

It follows by the second items  in Conditions EPS3 and EPS4, that
for every fragment $F\in\cH$,
we have
\begin{equation}
E(F)\supseteq \{ \chi(F') | F'\in \cH(F)   \}
\end{equation}
In particular, we have
$E(T)\supseteq \{ \chi(F') | F'\in \cH(T)   \}$.
Now,
by  Condition EPS5, we get that each edge of $T$ is an induced candidate of some fragment.
That is, we have:
\begin{equation}
E(T)=\{ \chi(F') | F'  \in \cH(T)   \}
\end{equation}
The first items in Conditions EPS3 and EPS4 imply that for every fragment $F\in \cH\setminus\{T\}$, the edge $\chi(F)$ is outgoing from $F$.
This fact, together with part
$(2)$ in the definition of a hierarchy, implies that for every fragment $F\in\cH$,
\begin{equation}
\{ \chi(F') | F' \not \in \cH(F)   \} \bigcap E(F) ~=~\emptyset.
\end{equation}
Equations (1), (2), and (3)  imply that for every fragment $F\in\cH$, $\{ \chi(F') | F'\in \cH(F)   \}$.
In other words, $\chi$ is a candidate function for $\cH$, as desired.
\end{proof}

\begin{Comment}
Condition EPS0
is not required in order to prove the above lemma.
  If the labels were assigned by our MST construction algorithm, condition EPS0 holds too. Even though adding the condition seems redundant, we decided to add it because
  we believe  it makes the reading more intuitive.
\end{Comment}
Now, to verify that  assignments of $\EndP$ and $\PARENTS$ strings are legal with respect to a given legal assignment of $\ROOTS$ strings,
 we need to verify the five conditions above. Conditions EPS2--EPS5  can be  verified easily,
in 1 unit of time, while the first condition EPS1 needs additional information at each node to be verified in 1 unit of time. Specifically, verifying the rule amounts to verifying that exactly one of the nodes in a fragment of level $i$ has its $i$'th position in $\EndP$ equal to 1. This is easy to do in a scheme that is very similar to Example $\NK$ in Section \ref{sub:examples}.
Hence, we omit this simple description (nevertheless, it is demonstrated
in Table \ref{tab:table1} in the example of the $\AND$ strings of nodes corresponding to
Figure \ref{fig:HierarchyTree1}).

\subsection{The distributed marker algorithm} \label{sec:labels-construction}
\label{app:marker}
A natural method for assigning the labels of the 1-proof labeling scheme described above (composed of the representation of $\cH_{\cM}$ and $\chi_{\cM}$, and the strings $\ROOTS$, $\PARENTS$, $\EndP$, and $\AND$), is to follow the construction algorithm of the MST, namely $\ALG$
(see Section \ref{sec:mst-construction}), which, in particular, constructs the hierarchy $\cH_{\cM}$ and the candidate function $\chi_{\cM}$.
Recall that the complexity of $\ALG$ is $O(n)$ time and $O(\log n)$ bits of memory per node.

Essentially, assigning the labels is performed by adding some set of actions to $\ALG$. These actions do not change the values of any of the variables of the original algorithm. Also, we do not change the algorithm's flow of control, except for adding these actions. Since each action is just a new assignment to a new variable (of logarithmic size), the addition of these actions cannot violate the correctness of $\ALG$, nor change its time and memory complexities (except by a constant factor).
We note that adding these actions on top of $\ALG$  is not complicated, and can be realized using standard techniques. Hence, we omit it here.
Hence, we obtain the following.

\begin{lemma}\label{lem:lineal-marker}
There exists a distributed marker algorithm assigning the labels of the 1-proof labeling scheme described in Section~\ref{sec:section5},  running  in $O(n)$ time and using $O(\log n)$ bits of memory per node.
\end{lemma}
The lemma above together with Lemma \ref{lem:correct-candidate} establishes Lemma \ref{lem:simple-proof}.

\section{Distributing pieces of information} 
\label{sec:proof}
In the previous section, we described the verification that (1) a tree exists, (2) it is decomposed into a hierarchy of fragments, and (3) edges emanating from the fragments compose a candidate function (so that the tree is the collection of these edges). That verified the Well-Forming property. It is left to verify  the Minimality  property. That is, it is left to show that each edge of the candidate function is the minimum outgoing edge of some fragment in the hierarchy. The current section describes a part of the marker algorithm responsible for marking the nodes for this verification.

 Informally, to perform the verification, each node must possess some information regarding each of the fragments $F$ containing it.
 The information regarding a fragment $F$ contains the weight of the selected edge of the fragment as well as the fragment identity, hence, it can be encoded using $O(\log n)$ bits.
 (The fragment identity is needed to identify the set $O_F$ of outgoing edges from $F$, and the weight of the selected edge is needed for comparing it to the weight of the other edges of $O_F$; this is how we detect that the selected edge is indeed the minimum).
  However, as mentioned, each node participates in $O(\log n)$ fragments, and hence, cannot hold at the same time all the information relevant for its fragments. Instead, we distribute this information among the nodes of the fragments, in a way that will allow us later to deliver this information efficiently to all nodes of the fragment. In this section, we show how to distribute the information regarding the fragments. In the next section, we explain how  to  exploit this distribution of information  so that during the verification phase, relevant information can be delivered to nodes relatively fast and without violating the $O(\log n)$ memory size.

\paragraph{The piece of information $\id(F)$:} As mentioned in Section \ref{intuition:P2},
a crucial point in
the scheme is letting each node $v$ know, for each of its incident edges
$(v,u) \in E$
 and for each level $j$, whether $u$ and $v$ share the same level $j$ fragment.
(Note, in the particular case where $u$ is also a neighbour of $v$ in $T$, this information can be extracted by $v$ using $u$'s data-structure, see Lemma \ref{lem:simple-proof}.) Intuitively, this is needed in order to identify outgoing edges.
For that purpose,
we assign each fragment a unique identifier, and $v$ compares the identifier of its own level $j$ fragment with the identifier of $u$'s level $j$ fragment.
The identifier of a fragment $F$ is $\id(F):=\id(r(F))\circ lev(F)$, where $\id(r(F))$ is the unique identity of the root $r(F)$ of $F$, and $lev(F)$ is $F$'s level.
We also need each node $v$ to know the weight $\omega(F)$ of the minimum outgoing edge of each  fragment $F$ containing~$v$. To summarize, the {\em piece of information} $\Info(F)$ required in each node $v$ per fragment $F$ containing $v$
 is $\Info(F):=\id(F)\circ\omega(F)$. Thus, $\Info(F)$ can be encoded using  $O(\log n)$ bits.

\comment
When describing the construction of hierarchy $\cH_{\cM}$ below, we view it as a tree (and call it the {\em  fragment-tree}) whose nodes are fragments of $T$, and in particular, the  root of the fragment-tree is
$F^{\ell}$, the level $\ell$ fragment which is the whole tree $T$. We shall make sure that $\ell=O(\log n)$.
We refer to the nodes of the fragment-tree $\cH_{\cM}$ as {\em fragment nodes}.
Consider a fragment node $F$, corresponding to a fragment  of $T$.
Whenever no ambiguity arises, we  refer to $F$ sometimes as a node of $\cH_{\cM}$ and sometimes as a fragment of $T$. For example, we may refer either to $F^{\ell}$ or to $T$, whichever simplifies the notations better.

\subsubsection{Constructing  hierarchy $\cH_{\cM}$ and candidate function $\chi$}
 \label{subsub:phases}

.

Unfortunately, we did not manage to utilize one of the known decompositions of a tree, or an MST.
One famous such decomposition is the set of fragments created in the algorithm of Gallager, Humblet, and Spira \cite{ghs} as building blocks of the MST in constructed by their seminal algorithm. First, many of their fragments are very large. We could not spread the fragment's $\Info$ over such a large fragment, since this would take the information far beyond the desired locality radius.
Our construction below groups such fragments into an object we there term ``phase 1'' and replicate $\Info$ over phase 1 as needed in order not to violate the locality radius constraint.
The rest of the fragments are grouped into phases as well. This is motivated by
 our intention to show later (in Section
\ref{sec:distributed-implementation}) a distributed implementation of the labeling schemes presented here.
Intuitively, the implementation needs to use some memory per an object, and let the marker algorithm distribute this memory over the nodes of the object
(we shall see the details in  Section
\ref{sec:distributed-implementation}). Had there been too many objects, we would have violated the memory constraint. Hence, we needed to group fragments together, leading the to ``phases'' structure described below.

The marker constructs  the fragment-tree $\cH_{\cM}$ in $O(\log^* n)$ iterations. Each iteration
constructs a forest (a collection of subtrees of $\cH_{\cM}$) termed {\em phase}. Specifically,
each iteration $i\geq 1$ that constructs a phase $P$, starts with a collection of fragment nodes,
which then become
the roots of the subtrees of phase $P$. We term them
the {\em root fragments} of phase $P$.
For an iteration  $i> 1$, each such root fragment (of phase $P$) is a child
 (in $\cH_{\cM}$) of a fragment leaf of the phase $P^-$ constructed in the previous iteration.
  The first iteration starts with the root fragment of $\cH_{\cM}$, namely, $F^{\ell}$.
  For a fragment $F$ in some phase $P$, let $F^*$ denote
 the root fragment of $P$ that is an ancestor of $F$ in $\cH_{\cM}$.
 We refer to $F^*$ as the {\em root fragment}
 of $F$.

Let $F$ be a fragment of some phase constructed in iteration $i$.
We define parameter $\zeta(F)$ depending on the iteration $i$ as follows.
For $i=1$, $\zeta(F)=\tau/15$ and for $i>1$, $\zeta(F)=\log^2 |F^*|$.
 Two crucial properties of the hierarchy we construct are that for every leaf fragment $F$, we have (1)   $|F|=\Omega(\zeta(F))$, and (2) each child $F_i$ of $F$ (a root fragment of the next phase) satisfies $|F_i|=O(\zeta(F))$.

Each iteration consists of steps, each starting  with a fragment already in the hierarchy, and already in the current phase $P$.
In particular, the first step of the first phase starts with
 the whole $T= F^{\ell}$ and the first step of every other phase starts with a root fragment of that phase.
Let us now describe a step of an iteration. Consider  a fragment $F$. If $F$ is a singleton fragment then $F$ in a leaf of the hierarchy $\cH_{\cM}$.
If $|F|= 2$ then $F=\{u,v\}$ and we let the singleton fragments $\{u\}$ and $\{v\}$ be the two children of $F$ in the hierarchy $\cH_{\cM}$
 (each being a leaf of $\cH_{\cM}$). Consider now a
 fragment $F$  such that $|F|>2$. Fragment $F$  is first decomposed (by the removal of one node- the {\em separator}) into
 multiple subtrees $T_1, T_2, T_3, ..., T_j$.
 In
 traditional separator decompositions, the separator is chosen such that
  $|T_i|<|F|/2$ for every $1\leq i\leq j$,
and each such $T_i$ becomes a child fragment of $F$ in the hierarchy. (The ``removed'' node is added to one of the $T_i$'s arbitrarily besides belonging to $F$).
 We start by partitioning $F$ the same way the traditional decomposition does. However, we change this decomposition (for the construction of $\cH_{\cM}$) to satisfy the two crucial properties mentioned above.

Assume, without loss of generality, that after removing the chosen separator $s$, the resulted subtrees
$T_1,T_2,\cdots, T_j$ are ordered by their size, i.e., such that  $|T_l|\leq |T_r|$ for $l<r$. Let $T'_1$ be the subtree
$T_1$ merged with the separator $s$.
  First of all, if we either have $|F|<6\zeta(F)$ or
that  each  subtree $T_1,T_2,\cdots, T_j$ is smaller than $\zeta(F)$ then
 $F$ becomes a leaf of the current phase $P$ and
the subtrees $T'_1, T_2, T_3,\cdots, T_j$ become root fragments for the phase constructed in the next iteration.
Each such root fragment is nevertheless considered a child of $F$ in the
hierarchy $\cH_{\cM}$ we build.

Consider now the case that $|F|\ge 6\zeta(F)$.
If
 $|T_i| \ge \zeta(F)$ holds for each $i$,
then each of the subtrees $T'_1, T_2, T_3,\cdots, T_j$ becomes a child fragment of $F$, still in the current phase $P$.
We are left with the case that there exists some integer $1<k<j$, such that
$|T_1|, |T_2|, ... , |T_k| < \zeta(F)$, and $|T_{k+1}|, |T_{k+2}|, ..., |T_j| \ge \zeta(F)$.
Consider two subcases:
\begin{enumerate}
\item
 $|T_1|+ |T_2| + ... + |T_k| \ge \zeta(F)$.
 Merge $T_1, T_2, ..., T_k$ adding the separator $s$ of $F$. The result becomes one child fragment $F'$ of $F$ in the same phase.
 Each of the ``large'' subtrees, $T_{k+1}, T_{k+2}, ..., T_j$ also becomes a child of $F$ in the same phase.
In the next step, when considering $F'$, choose $s$ as its separator, to guarantee that $F'$ would be a leaf of the
current phase (note, after removing $s$, fragment $F'$ breaks into the ``small'' subtrees $T_1, T_2, ..., T_k$ which are
root fragments of the next phase).

\item
$|T_1|+ |T_2| + ... + |T_k| < \zeta(F)$. Merge them with $T_j$ and add the separator node. The resulting fragment becomes a child fragment of $F$, in the same phase.  Each of the ``large'' subtrees, $T_{k+1}, T_{k+2}, ..., T_{j-1}$ also becomes a child of $F$ in the same phase.
\end{enumerate}

\begin{lemma}
\label{lem:hierarchy-properties}
Let $F$ be a leaf fragment  of some phase constructed in
iteration $i$. \\The following {\bf  hierarchy properties} hold:
 \begin{enumerate}
\item If $|F|>1$ then $F$ is not a leaf of the hierarchy $\cH_{\cM}$, and
 each child $F_i$ of $F$  satisfies $|F_i|\leq 3\zeta(F)\leq \tau/5$.
 \item $|F|=\Omega(\zeta(F))$.
 \item There are $O(\log^* n)$ phases.
 \item The height of $F$ in the phase of  $F$ is $O(\log |F^*|)$.
\item  The height of hierarchy $\cH_{\cM}$ is $O(\log n)$.
 \end{enumerate}
\end{lemma}

\begin{proof}
Consider a leaf fragment $F$ of some phase constructed in
iteration $i$. It follows trivially by the construction that if $|F|>1$ then $F$  is not a leaf of $\cH_{\cM}$.
Let $F'$ be a child of $F$ in $\cH_{\cM}$. Either $F$ satisfies
$|F|<6\zeta(F)$, in which case we have $|F'|\leq 3\zeta(F)$, or that $|F|\ge 6\zeta(F)$ but
 $|F'|<\zeta(F)$. property (1) follows. property (3) follows from property (1), since for iteration $i>1$, we have $\zeta(F)=\log^2 |F^*|$.
property (2) is satisfied since we guarantee that a fragment $F'$ is a child of a fragment $F$ in the same phase
only if $|F'| \ge \zeta(F)$. Let us now show that Properties (4) and (5) hold as well.

Let $P$ be a phase and let $F$ be a fragment in $P$.
We first claim that if fragment $F'$, which is a child of fragment $F$, satisfies $|F'|\ge 2|F|/3$ then  $F'$ is a leaf fragment of phase $P$. Assume by contradiction that $F'$ is not a leaf of $P$ and yet $|F'|\ge 2|F|/3$. Consider first the case that $F'$ was  constructed in the first subcase above.
The fact that $F'$ is not a leaf means that $F'$ was not constructed  by merging $T'_1, T_2, ..., T_k$.
Hence, the separator $s$ of $F$ is the one selected by the traditional separator decomposition.
However, if $F'$ is one of the subtrees $T_i$
obtained by removing $s$ from $F$ then $|F'|< |F|/2$. (If $F'=T'_1$ then
$|F'|\leq |F|/2$). A contradiction.
Finally, consider the case that $F'$ was constructed in the second subcase above, by merging $T'_1, T_2, ..., T_k$ together with $T_j$.
Observe that $|T_j|< |F|/2$ and that $|T_1|+ |T_2| + ... + |T_k| < \zeta(F)$. This means that
$|F'|<|F|/2+\zeta(F)$.
If $|F'|\geq 2|F|/3$ then we have $2|F|/3<|F|/2+\zeta(F)$ which implies that $|F|<6\zeta(F)$. However, in this case, $F$ is a leaf of $P$, a  contradicting. This establishes our claim. property (4) now follows immediately, and property (5) follows by property (3).

\end{proof}

The candidate function $\chi_{\cM}$ for $\cH_{\cM}$ is defined as follows. For each $F\in \cH_{\cM}\setminus\{T\}$, let $F'$ be the parent of $F$  in $\cH$.
Then $\chi_{\cM}(F)$ is simply the edge in $T$ connecting $F$ with $H$.

\commentend

At a very high level description, each node $v$ stores {\em permanently}  $\Info(F)$ for a constant number of  fragments~$F$. Using that,
  $\Info(F)$ is ``rotated''  so that each node in $F$ ``sees'' $\Info(F)$ in
  $O(\log n)$  time.
We term  the mechanism that performs this rotation a {\em train}. A crucial point is having each node participate in only few trains.
Indeed, one train passing a node could delay the other trains that ``wish'' to pass it. Furthermore, each train utilizes some (often more than constant) memory per node. Hence, many trains passing at a node would have
violated the $O(\log n)$ memory constraint. In our solution, we let each node participate in  two trains.

Let us recall briefly the motivation for {\em two} trains rather than one. As explained in Section \ref{subsub:overview-proof}, one way to involve only one train passing each node would have been to partition the nodes, such that each fragment would have intersected only one part of the partition. Then, {\em one} train could have passed carrying the pieces of information for all the nodes in the part. Unfortunately, we could not construct such a partition where the parts were {\em small}. A small size of each part is needed in order to ensure that a node sees all the pieces (the whole train) in a short time.

Hence, we construct {\em two} partitions of the tree. Each partition is composed of a collection of  node-disjoint subtrees called {\em parts}.
For each partition, the collection of parts covers all nodes. Hence, each node belongs to precisely two parts, one part per partition. For each part,
we  distribute the information regarding some of the fragments it intersects, so that each node holds at most a constant number of such pieces of information.
 Conversely, the information regarding a fragment is distributed to nodes of one of the two parts intersecting it.
 Furthermore, for any node $v$, the two parts corresponding to it encode together the information regarding all fragments containing $v$. Thus, to deliver all relevant information, it suffices to utilize one train per part (and hence, each node participates in two trains only). Furthermore, the partitions are made so that the diameter of each part is $O(\log n)$, which allows each train to pass in all nodes in short time, and hence to deliver the relevant information quickly. The mechanism of trains and their synchronization is described in the next section.
 The remaining of this current section is dedicated to the construction of  the two partitions, and to explaining how the information regarding fragments is distributed over the parts of the two partitions.

\subsection{The two partitions}
 \label{sec:partitions}
Consider a correct instance, and fix the corresponding hierarchy tree $\cH=\cH_{\cM}$.
We now describe {\em two} partitions of the nodes in $T$, called  $\Top$ and $\Bottom$. (The distributed algorithm that constructs the partitions
is described later.)  We also partition
  the fragments into two kinds, namely, {\em top} and {\em bottom} fragments.
  
 \paragraph{Top and bottom fragments:}  Define  the {\em top} fragments to be precisely those fragments whose number of nodes is at least $\log n$. 
   Observe that the top fragments correspond to a  subtree of the hierarchy tree~$\cH$. Name that subtree $T_{\Top}$.
  All other fragments are called {\em bottom}. See the left side of Figure \ref{fig:Top} for an illustration of the top fragments and the subtree $T_\Top$.

\subsubsection{Partition $\Top$}
Let us first describe partition $\Top$. We first need to define three new types of fragments.

\paragraph{Red, blue, and large fragments:}
   A leaf fragment in subtree $T_{\Top}$ is colored red.
   A fragment not in $T_{\Top}$ which is a sibling in $\cH$ of a fragment in $T_{\Top}$ is colored blue. (Equivalently, a blue fragment is a fragment not in $T_{\Top}$, whose parent fragment in $\cH$ is a non-red fragment in $T_{\Top}$). The following observation is immediate.
   \begin{observation}
   \label{obs:red-blue-partition}
  The collection of red and blue fragments forms a partition ${\cal{P}}'$ of the nodes of~$T$. See Figure \ref{fig:Top}  for an illustration of partition ${\cal{P}}'$.
  \end{observation}
  To emphasize the fact that each non-blue child fragment of an internal fragment in $T_{\Top}$  contains at least $\log n$ nodes, we call internal fragments  in $T_{\Top}$  {\em large}. Note, the large fragments are precisely the (strict) ancestors of the red fragments in $\cH$.
 Since the ancestry relation in $\cH$ corresponds to an inclusion relation between the corresponding (active) fragments in $T$, we obtain the following observation.
  \begin{observation}
  \label{obs:large-red-blue}
Each large fragment $F_{large}$ is composed of at least one red fragment $F_{red}$ as well as  one or more blue ones, and does not contain any additional nodes (of course, the part  may contain also the {\em edges} connecting those fragments).
\end{observation}

\paragraph{Partition ${\cal{P}}''$:}
\label{par:red-blue-merge}
 Our goal now is to partition the nodes to parts such that each part contains {\em precisely} one red fragment and possibly several blue ones, and no additional nodes. Such a partition exists, since,  it is just a coarsening of the partition $\cP'$ of the nodes to red and blue fragments. Moreover, the construction of {\em some} such a partition is trivial, following Observation \ref{obs:large-red-blue} and the fact that the tree is a connected graph. The following procedure produces such a partition ${\cal{P}}''$  that has an additional property defined below. (A less formal description of the procedure is as follows: let pink parts be either red fragments, or the results of a merge between a red fragment and any number of blue ones.
 Now repeat the following as long as there are unmerged blue fragments: consider a blue fragment $F_{blue}$ who has a sibling pink fragment and, moreover touches that sibling; merge $F_{blue}$ with one of its sibling pink fragments  it touches).

 \paragraph*{Procedure Merge}
 \begin{enumerate}
 \item
 Initialize the set $\tilde{\cP}$ of parts to include precisely the set of red parts.
 (*  $\tilde{\cP}$ is not yet a partition *)
 \item
 Repeat while there are blue fragments not merged into parts of $\tilde{\cP}$
 \begin{enumerate}
 \item
 Let $F_{large}$ be a top fragment that contains a node $u$ that is {\em not} in any part
 of $\tilde{\cP}$, where
  all the nodes of every child fragment of $F_{large}$ belong to parts of
 $\tilde{\cP}$.
    \item
    Let $F_{blue}$ be the blue fragment containing $u$. (Note that we have $u\in F_{blue}\subset F_{large}$.) Let
    $\tilde{P} \in \tilde{\cP}$ be some part that touches $F_{blue}$.
    \item
    Merge $F_{blue}$ with one such $\tilde{P}$. (This also removes $\tilde{P}$ from $\tilde{\cP}$ and inserts, instead, the merged part $F_{blue} \bigcup \tilde{P}$
  \end{enumerate}
  \item When the procedure terminates,  ${\cal{P}}''\gets \tilde{\cP}$.
\end{enumerate}
 See Figure \ref{fig:Top} for an illustration of partition ${\cal{P}}''$.
It is easy to see (e.g., by induction on the order of merging in the above procedure) that partition ${\cal{P}}''$ is constructed in the following way: let $F_{red}$ be the red fragment in a part $\tilde{P}$. Then all the nodes in $\tilde{P}$ belong to ancestor fragments of $F_{red}$. This leads to the following observation.

\begin{figure}
\centering
\includegraphics[scale=0.8]{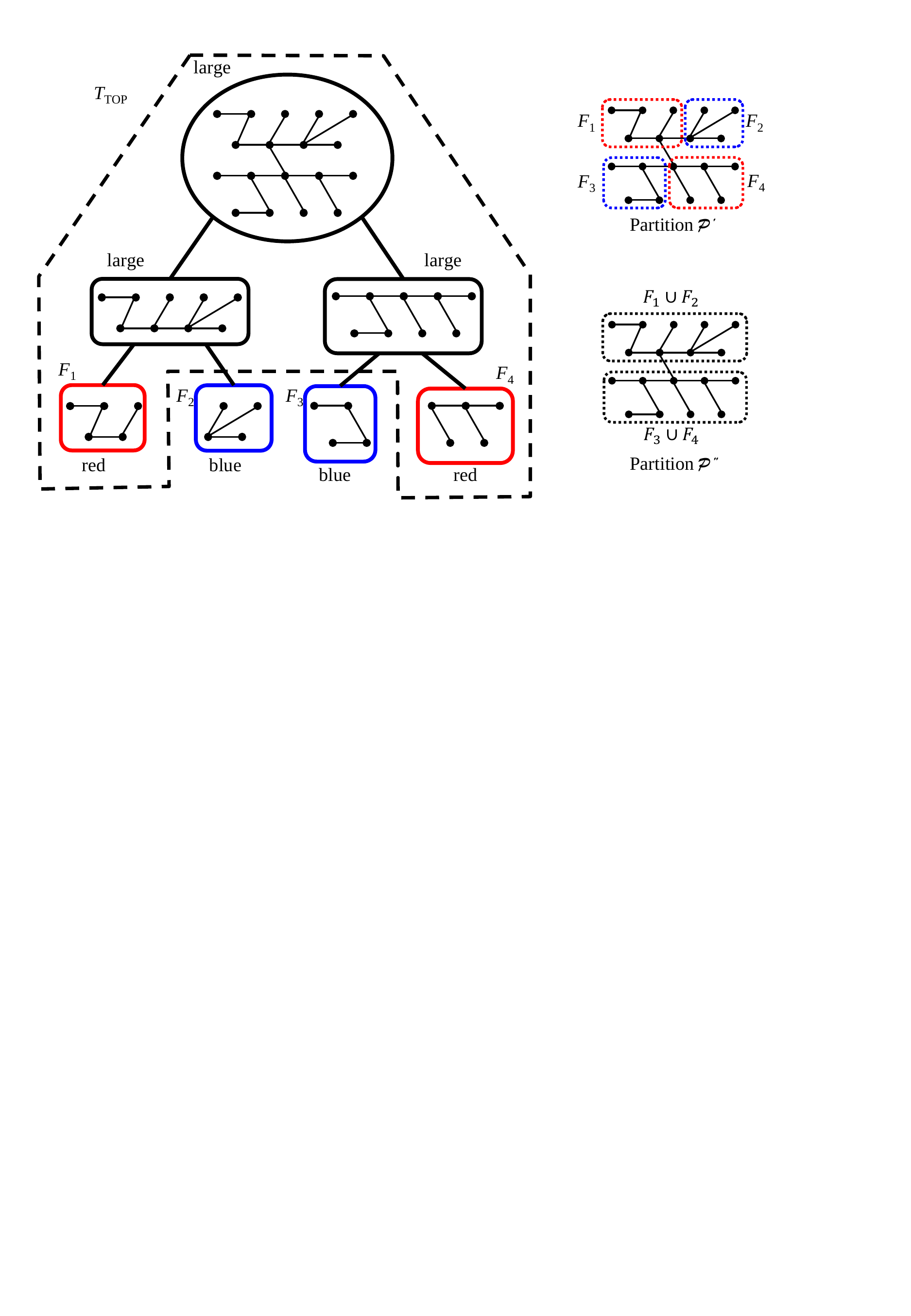}
\caption{On the left: the top fragments and  $T_\Top$; On the right: partition ${\cal{P}}'$ (above) and  partition ${\cal{P}}''$ (below)}
\label{fig:Top}
\end{figure}

 \begin{claim}
 \label{cla:red-blue-intersect}
  Each part  $P\in {\cal{P}}''$
intersects at most one level~$j$ top fragment, for every~$j$.
\end{claim}
The property captured in the above claim is very useful. As can be seen later, this property means that the train in each part $\tilde{P}$ needs to carry only one piece of information for each level.

\paragraph{Partition $\Top$:}
    We would  like to
pass a train in each part  $P$ of ${\cal{P}}''$.
Unfortunately,  the diameter of $P$ may be too large. In such a case, we partition $P$ further to
   {\em neighbourhoods}, such that each neighbourhood is a subtree of $T$ of size at least $\log n$
   and of diameter
   $O(\log n)$. The resulted partition is called $\Top$. The lemma below follows.

 \begin{lemma}\label{lem:top}
 For every part $P$ in partition $\Top$, the following holds.
 \begin{smallitemize}
 \item
 $|P|\geq\log n$,
 \item
  $D(P)=O(\log n)$, where $D(P)$ is the diameter of $P$.
  \item $P$
intersects at most one level $j$ top fragment, for every $j$ (in particular, it intersects at most $\ell=\lceil\log n\rceil$
top fragments).
 \end{smallitemize}
   \end{lemma}

\subsubsection{Partition $\Bottom$}
The bottom fragments are precisely those with less than $\log n$ nodes.
The parts of the second partition $\Bottom$  are the following: (1) the blue fragments, and (2) the children fragments in $\cH_{\cM}$ of the red fragments. By Observation \ref{obs:red-blue-partition}, this collection of fragments is a indeed a partition. Observe that each part of $\Bottom$ is a bottom fragment.
Thus, the size, and hence the diameter, of each  part $P$ of  $\Bottom$,  is less than $\log n$.
Figure~\ref{fig:Bottom} illustrates the bottom fragments and partition $\Bottom$.
 Observe also that a part $P\in \Bottom$ contains all of $P$'s descendant fragments in $\cH$ 
 (recall, $P$ is a fragment, and the collection of fragments are a laminar family), and does not intersect   other bottom fragments.
 Hence, we get the following.

\begin{figure}
\centering
\includegraphics[scale=0.8]{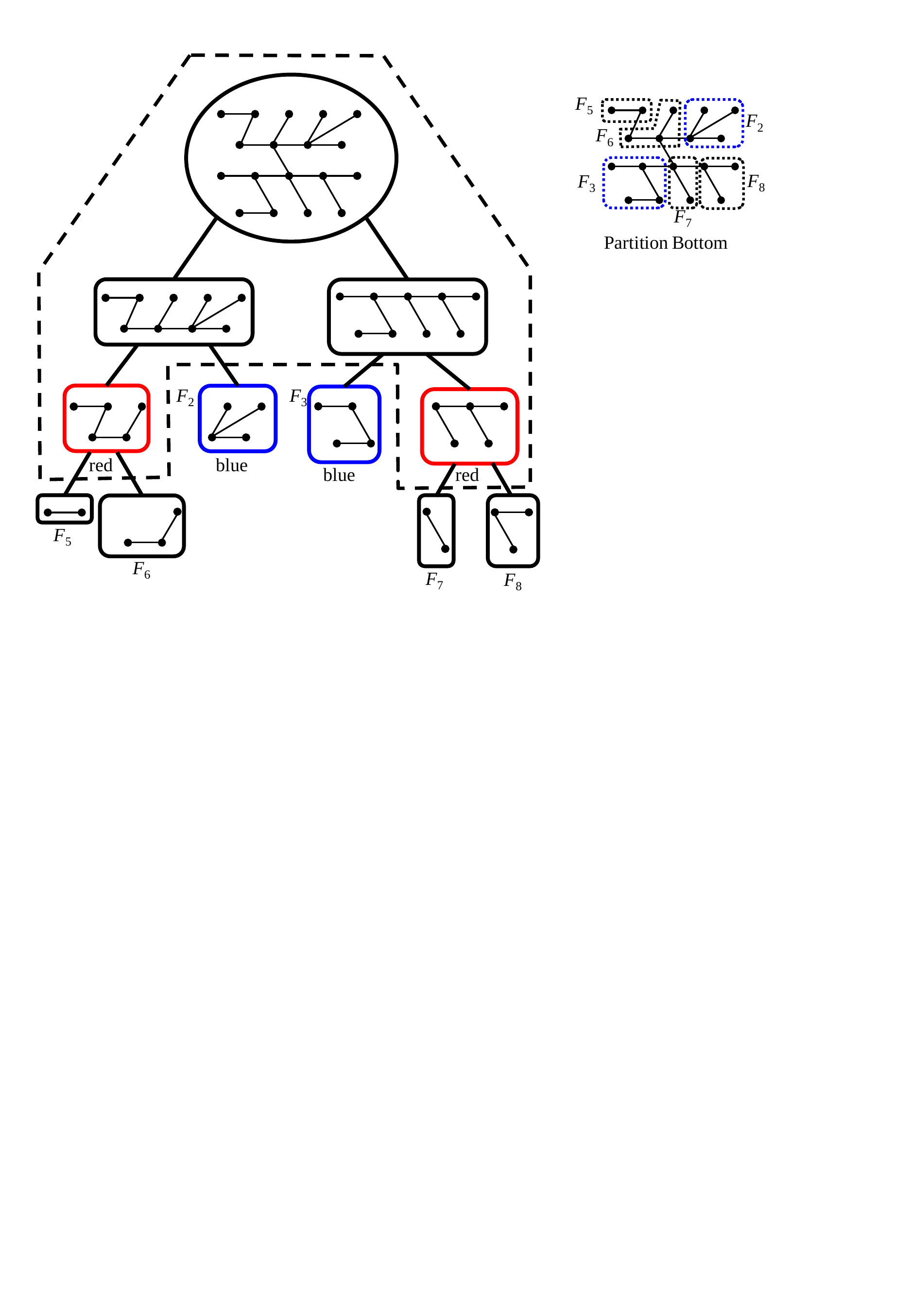}
\caption{The bottom fragments and partition $\Bottom$.}
\label{fig:Bottom}
\end{figure}

  \begin{lemma}\label{lem:bottom}
For every part $P$ of partition $\Bottom$, 
the following holds:
 \begin{smallitemize}
 \item
   $|P|<\log n$, and
   \item $P$
intersects at most $2|P|< 2\log n$ bottom fragments.
 \end{smallitemize}
   \end{lemma}

\subsubsection{Representations of the partitions}
In Section~\ref{app:neighbourhoods}, we show that the above partitions $\Top$ and  $\Bottom$  can be constructed by a distributed algorithm that uses $O(\log n)$ memory
and linear time. Each part $P$ of each of the two partitions is represented  by encoding  in a designated part of the label of each node in $P$, the identity $\id(r(P))$ of the root of $P$  (the highest node of part $P$). Recall that a node participates in only two parts (one of each partition), so this consumes $O(\log n)$ bits per node.
Obviously, given this representation, the root of a part can identify itself as such by simply comparing the corresponding part of its label with its identity.
In addition, by consulting the data-structure of a tree neighbour~$u$, each node $v$ can detect whether $u$ and~$v$ belong to   the same part (in each of the two partitions).

A delicate and interesting point is that the verifier does not need to verify directly that the partitions $\Top$   or $\Bottom$  were constructed as explained here. This is explained in Section \ref{sub:3.3}.

\subsection{Distributing the information of fragments}
\label{subsec:train-initial}

Fix a part $P$ of  partition $\Top$  (respectively, $\Bottom$). 
Recall that $P$ is a subtree of $T$  rooted at $r(P)$.
Let $F_1,F_2,\cdots, F_k$ be the top (resp., bottom) fragments intersecting~$P$, for some integer $k$.
 By Lemma \ref{lem:top} (resp., Lemma \ref{lem:bottom}), we know that $k\leq \min\{2|P|,2\log n\}$.
Assume w.l.o.g., that the indices are such that
the level of $F_i$ is, at least, the level of $F_{i-1}$, for each $1< i\leq k$.

The information concerning part $P$ is defined as $\Info(P) =\Info(F_1)\circ\Info(F_2)\circ,\cdots,\circ\Info(F_k)$. We distribute this information over the nodes of $P$ as follows.
We break $\Info(P)$ into~$|P|$ pairs of pieces. Specifically, for $i$ such that $1~\leq~i~\leq~\lceil k/2\rceil\leq |P|$,  the $i$'th pair,  termed $\piece(i)$, contains $\Info(F_{2i-1})\circ\Info(F_{2i})$
(for odd $k$, $\piece(\lceil k/2\rceil)=\Info(F_k)$).

The process of storing the pieces permanently at nodes of a part of the partition is referred to as the {\em initialization of the trains}.
The distributed algorithm that implements  the initialization of the trains using $O(\log n)$ memory size and linear time is described next. It is supposed to reach the same result of the following non-distributed algorithm (given just in order to define the result of the distributed one).

This non-distributed algorithm is simply the classical Depth First Search (DFS) plus the following operation in every node visited for the first time.
 Consider a DFS traversal over $P$  that starts at $r(P)$  and let  $\DFS(i)$ denote the  the $i$'th node visited in this traversal.
 For each $i$,  $1\leq i\leq \lceil k/2\rceil$, $\DFS(i)$   stores  permanently the $i$'th pair of  $\Info(P)$,  namely, $\piece(i)$.

\subsection{Distributed implementation}
\label{app:neighbourhoods}

Before describing the distributed construction of the two partitions, namely $\Top$ and $\Bottom$, we need to describe a tool we use for efficiently executing several waves\&echoes operations in parallel.
This {\em $\MULTI$} primitive (described below) performs a Wave\&Echo in every fragment in $\cH$ of level $\level$, for
$\level= 0, 1, 2, \cdots,\ell$.
 Moreover, the $i+1$th Wave\&Echo
 is supposed to start after the $i$th Wave\&Echo terminates.
Furthermore, all this is obtained in time $O(n)$.

 \subsubsection{The $\MULTI$ primitive}
\label{sub:multi-wave}

We shall use this primitive only after the $\ROOTS$ string is already set, so that every node knows for each level, whether it is the root of a fragment of that level.
Let us first present a slightly inefficient way to perform this.
The root of the whole tree starts $\ell+1$
{\em consecutive} waves and echoes, each for the whole final tree.
(By consecutive we mean that the $\level+1$th wave starts when the $\level$th wave terminates.)
Let the level $\level$ wave be called $\WAVE^{I_1}$(T,$\level)$ since it carries some instruction $I_1$, is sent over the whole tree $T$, and carries the information that it is meant for level~$\level$.
A root $v_{\level}$ of a fragment $F^{\level}$ of level $\level$, receiving
$\WAVE^{I_1}$(T,$\level)$, then starts its own Wave\&Echo $\WAVE^{I_2}$($F^{\level}, \level)$
over its own fragment only. (Here, $I_2$ is some instruction possibly different than  $I_1$.)
A node who is not a root of a level $\level$ fragment can echo $\WAVE^{I_1}$(T,$\level)$ as soon as all its children in the final tree (if it has any) echoed. A root~$v_{\level}$ echoes $\WAVE^{I_1}$(T,$\level)$ only after its own wave
$\WAVE^{I_2}$($F^{\level}, \level)$ terminated (and, of course, after it also received the echoes of $\WAVE^{I_1}$(T,$\level)$ from all its children).
The following observation follows immediately from the known semantics of Wave\&Echo.

\begin{observation}
\label{obs:multi-wave-semantics}
Consider a fragment $F^{\level}$ of level $\level$ rooted at some $v_{\level}$.
The wave initiation by $v_{\level}$ starts after all the waves involving its descendant fragments 
terminated (at the roots of the corresponding fragments).
\end{observation}
The ideal time complexity of performing the above  collection of $\ell$ waves is $\Theta(n \log n)$. In the case that the size of a level $\level$ fragment $F^{\level}$ is
$ 2^j \leq |F^{\level}| < 2^{j+1}$,
we can achieve the semantics of Observation \ref{obs:multi-wave-semantics} somewhat more time efficiently. The primitive that achieves this is termed a $\MULTI$. When invoking it, one needs to specify which instructions it carries.
Informally, the idea is that the roots $R_0$ of level $0$ fragments perform the wave (for level 0) in parallel, each in its own fragment of level $0$ (a single node).
Recall that a fragment $F_1$ of level 1 contains multiple fragments of level zero. The roots of these fragments of level zero report the termination of the level $0$ wave to the root of $F_1$. Next, the roots $R_1$ of level $1$ fragments perform the wave (for level 1) in parallel, each in its own fragment of level $1$. The terminations are reported to level $2$ fragment roots, etc., until the $\MULTI$ terminates.

The $\MULTI$ is started at the root of the final tree $T$ by a wave termed
$\MULTI(T, I_1, I_2)$.
 Each node $v$ who receives $\MULTI(T, I_1, I_2)$ acts also as if $v$ has initiated a  $\WAVE^{I_2}(F^{0},0)$ on a tree containing only itself. Ensuring the termination and
 the semantics for level $\WAVE^{I_2}(F^{0},0)$ is trivial. We now define the actions of levels higher than zero in an inductive manner.
Every node $v$ who received (and forwarded to its children if it has any) $\MULTI(T, I_1, I_2)$, simulates the case that it
received (and forwarded to its children)  $\WAVE^{I_2}(F^{\level}, \level)$  for every level $\level$. However, $v$ is not free yet to echo   $\WAVE^{I_2}(F^{\level}, \level)$   until an additional condition holds as follows:
  When some wave $\WAVE^{I_2}(F^{\level}, \level)$ terminates at the root $v_{\level}$ of $F_j$, this root initiates an informing wave
  $\WAVE^{\FREE-I_2}(F^{\level}, \level)$
 to notify the nodes in $F^{\level}$ that the wave of level $\level$ in their subtree terminated, and thus they are free to echo $\WAVE^{I_2}(F^{\level+1}, \level+1)$.
 That is, a leaf of a $F^{\level+1}$ fragment can echo $\WAVE^{I_2}(F^{\level+1}, \level+1)$ immediately when receiving  $\WAVE^{\FREE-I_2}(F^{\level}, \level)$, and a non-leaf of
 $\WAVE^{I_2}(F^{\level+1}, \level+1)$ may echo $\WAVE^{I_2}(F^{\level+1}, \level+1)$ when it receives echoes from all its children in $F^{\level+1}$.

Specifically, the  convergecast  is performed to the containing $\level+1$ fragment as follows:
 a leaf of a level $\level+1$ fragment
who receives $\WAVE^{\FREE-I_2}(F^{\level}, \level)$
sends a message
 $\READY(\level+1,I_2)$ to its parent.
A parent node  sends message  $\READY(\level+1, I_2)$ if it is not a root of a level $\level+1$ fragment, and only after receiving $\READY(\level+1, I_2)$ from all its children. When a root of a level $\level+1$ fragment receives the $\READY(\level+1, I_2)$ message from all of its children, it starts  $\WAVE^{\FREE-I_2}(F^{\level+1}, \level+1)$. The $\MULTI$  terminates at the root of the final tree when the wave for level $\ell$ terminates at that root.
  The informing wave   $\WAVE^{\FREE-I_2}(F^{\level}, \level)$ itself needs no echo.

\begin{observation}
\label{obs:multi-wave-semantics2}
The efficient implementation
of the multi-wave simulates the multiple waves analyzed in Observation
\ref{obs:multi-wave-semantics}. That is, it obtains the same result for the instructions $I_1$ and $I_2$ in every node.
\end{observation}

\begin{proof}
Consider an alternative algorithm (for $\MULTI$) in which, when a root  of fragment $F^{\level+1}$ receives $\MULTI(T, I_1, I_2)$, it starts a wave
$\WAVE-\READY(\level+1, I_2)$. Assume further, that the $\READY(\level+1, I_2)$ messages are sent as echoes of
$\WAVE-\READY(\level+1, I_2)$. Moreover, assume that an echo $\READY(\level+1, I_2)$ is sent by a node only after it received  $\WAVE^{\FREE-I_2}(F^{\level}, \level)$. The claim for such an alternative algorithm would follow from
 Observation \ref{obs:multi-wave-semantics} and the known properties of Wave\&Echo. Now, it is easy to verify that
 the $\MULTI$ described simulates that alternative algorithm. That is,
  (1) $\MULTI(T, I_1, I_2)$ is sent by a node $ v_{\level+1}$ who belongs to a fragment of level $\level+1$ to its child $u $ in the same fragment exactly when it would have sent the imaginary $\WAVE-\READY(\level+1, I_2)$.
  This is easy to show by induction on the order of events. Moreover, at that time, the child $u $ knows the information carried by $\WAVE-\READY(\level+1, I_2)$ since it knows (from its $\ROOTS)$  which fragments it shares with its parent (and for each one of them we simulate the case $u$ now receives  $\WAVE-\READY)$.
 \end{proof}

\begin{observation}
The ideal time complexity of performing a multi-wave on the hierarch $\cH_{\cM}$ is $O(n)$.
\end{observation}

\begin{proof}
The wave started by the root consumes $O(n)$ time. Recall that hierarchy $\cH_{\cM}$ corresponds to active fragments during the construction of the MST by algorithm $\ALG$.
Hence, Lemma \ref{lem:disappear} implies that in hierarchy $\cH_{\cM}$,
the size of a level $\level$ fragment $F^{\level}$ satisfies
$ 2^j \leq |F^{\level}| < 2^{j+1}$. Thus,
 each wave started by a root of a fragment $F^{\level}$ of level $\level$ takes $O(2^{\level})$ time, and starts at time $O(2^{\level})$ after the initiation of the multi-wave.
\end{proof}


\subsubsection{Distributed construction of partition ${\cal{P}}'$}

The construction of partition ${\cal{P}}'$ is performed in several stages. Each of the tasks below is performed using  the $\MULTI$ primitive.
This is rather straightforward, given that the usage of Wave\&Echo as a primitive is very well studied. Below, we give some hints and overview.

First, we need to identify red fragments. It is easy to count the nodes in a fragment using Wave\&Echo to know which fragment has more than $\log n$ nodes.
 However, a large fragment that properly contains a red fragment is not red itself. Hence, the count is performed first in fragments lower in the hierarchy, and only then in fragments that are higher. Recall that the  $\MULTI$ primitive indeed completes  first waves in fragments that are lower in the hierarchy, before moving to fragments that are higher.
Hence, one execution of the  $\MULTI$ primitive allows to identify red fragment. At the end of this execution, the roots of fragments  know whether they are the roots of red fragments or not. A similar technique can be applied to identify blue fragments.

A second task is to identify a large  fragment $F_{large}$  that is not red, but has a child fragment who is red. It is an easy exercise to perform the construction using the $\MULTI$ primitive.
\commcomm
  In such a case, if the level of $F_{large}$ is $j$,
we set the variable
$\rchild(j)$  (initialized to $\False$) at a root of  $F_{large}$ to $\True$.
This task can be performed by the same set of Wave\&Echoes as above. Recall that in the $\MULTI$ primitive, the Wave\&Echo in $F_{large}$ is performed after the Wave\&Echo in its red child fragment $F_{red}$ terminated in
its root $v_{red}$, and
identified that $F_{red}$ was red.
Recall also, that $F_{large}$ contains $F_{red}$, so the Wave\&Echo in $F_{large}$ passes in $v_{red}$. Thus, by setting the proper flags, it is easy to set the value of
$\rchild(j)$ in $F_{large}$ correctly. We also make use of an output $\cred(j)$ (initialized to $\False$),  that is set to $\True$ if $F^j$ is either red or contains a red descendant fragment. This output of the first two tasks is  useful for performing a third task: identifying the blue fragments.

\begin{observation}
\label{obs:1st-partition}
Consider a fragment $F_{large}$ of some level $j$ which has a red child. When the above two tasks are completed successfully, the following holds.
\begin{enumerate}
\item
The value of $\rchild(j)$ in the root $r_{large}$ of $F_{large}$ is $\True$.
\item
Consider each child fragment $F^i$ of some level $i$ of $F_{large}$.
Assume further that $F^i$ is either red itself, or  it contains a red fragment. Then the value of $\cred(i)$ in the root $r_i$ of $F_i$ is $\True$.
 \item
 A child fragment $F^i$ of $F_{large}$ whose root $r_i$ does {\em not}
 have $\True$ in $\cred(i)$ is a blue fragment.
\end{enumerate}

\end{observation}
Based on the last part of the above observation, it is easy to have a root detect it is the root of a blue fragment, using another Wave\&Echo (that starts by the root of the parent fragment $F_{large}$).
 As before, the root of the final tree initiates a $\MULTI$ to facilitate that detection.

third task: identifying the blue fragments.

\commend

The third task is that of identifying the blue fragments.
A fourth task is to let each node in a blue fragment, and each node in a red fragment, know the color of their fragments. Again, designing these tasks is an easy exercise given the example of the first task above, and the $\MULTI$ primitive.

   \commcomm
   is easy to do using a Wave\&Echo from a root who already knows this color, as established above. This whole set of Wave\&Echoes is started by the root of the final tree (again, using a $\MULTI$).
\commend

\subsubsection{Constructing partition ${\cal{P}}''$}
  It is rather straightforward to use waves\&echos to implement procedure Merge to generate partition
   ${\cal{P}}''$. The red fragments use waves\&echos to annex roots of sibling blue fragments. They become pink parts (in the terminology of paragraph \ref{par:red-blue-merge}. Then this is repeated in the parent fragment, etc. Since this process goes from a lower level fragment to higher and higher levels, the $\MULTI$ primitive handles this well.

\commcomm
Hence, we construct ${\cal{P}}''$ inductively, using $\ell+1$
Wave\&Echoes, $\TREECOARSEN(T,j)$, for  $j=0, 1, 2, ... \ell$.
To define the inductive property, let us define the following.
   Call a partition in which each part consists of one red fragment and zero or more blue fragments a {\em red-centered partition}. Notice that the process defined in the previous paragraph, establishes a red-centered partition of
   $F_{large}$.
   This establishes the base for the following inductive property:\\

   \noindent {\bf Invariant 1:} Consider the time $\TREECOARSEN(T,j)$ is started. Let $F^i$ of level $i<j$ be a fragment who contains a red fragment (or is red itself). Then, $F^i$ is already partitioned into a red-centered partition. Moreover, each part is rooted in its highest (in the final tree) node.\\

   It follows from Observation \ref{obs:1st-partition}, that any fragment child of $F_{huge}$ that does not contain a red fragment (and is not red itself) is blue. By the induction assumption, when  $\TREECOARSEN(T,j)$ is started in some $F_{huge}$ of level $j$, any fragment child of  $F_{huge}$ is either red-centered partitioned, or is blue. The process is very similar to the one described for $F_{large}$ in the first paragraph of this construction. That is, the root of each red-centered part starts yet another Wave\&Echo (in $F_{huge}$ only) carrying its own $\id$ with invitations to blue fragments to join. Consider the root $v_{blue}$ of a blue fragment (a child fragment of $F_{huge}$) who has not yet joined a part of  ${\cal{P}}''$. Root $v_{blue}$ joins (together with its whole blue fragment) the part constructed by the root corresponding to the first such wave $v_{blue}$ receives.
\commend

\subsubsection{Constructing partition $\Top$}
Upon receiving the echoes for the $\MULTI$ primitive constructing
partition ${\cal{P}}''$, the root of the final tree instructs (by yet another Wave) each $\PARTLEADER$ of
 ${\cal{P}}''$ to start partitioning its part into parts of
partition $\Top$. That is, each part of ${\cal{P}}''$ is partitioned into subtrees, each of diameter  $O(\log n)$ and of size $\Omega(\log n)$.
This task is described in \cite{kdom}. When it is completed, each part of $\Top$ is rooted at its highest node. Moreover, every node in that part is marked by
the name of its part leader, in its variable called $\Top-\Roo$. 
 (Since $\Top$ is a partition, each node belongs only to one part; hence, this does not violate the $O(\log n)$ bits constraint.)

\subsubsection{Constructing partition $\Bottom$}
Recall that the parts of the second partition $\Bottom$ are (1) the blue fragments and (2) the child fragments of red fragments.
Let us term the latter {\em green fragments}.
We already established that members and roots of blue fragments know that they are members and roots of blue fragments.
The green fragments are notified in a similar way the blue ones were. That is, the root of the final tree starts a Wave\&Echo instructing the roots of the red fragments to notify child fragments that they are green.

\begin{claim}
\label{claim:top-bottom-assignment}
The two partitions $\Top$ and $\Bottom$ described in Section \ref{sec:partitions} can be assigned in time $O(n)$ and memory size $O(\log n)$.
\end{claim}

\subsubsection{Initializing the trains information}
\label{sub:init-trains}

First we describe a primitive that a root  of a part $P$ uses for storing $\Info(F)$ of one given fragment $F\in P$.
This is a well known distributed algorithm, so we do not describe it in detail. We use a distributed Depth First Search (DFS), see, e.g.~\cite{liuba-oded, chlamtac, awerbuch-dfs}. Initially, all the nodes in a part $P$ are marked $\vacant(P)$. When the root of the part wants to store the $\Info(F)$ of some fragment $F$, it sends this $\Info(F)$ (with a token) to perform a DFS traversal of part $P$. The first time that token reaches a node marked $\vacant(P)$,
it sets $\vacant(P)$ to $\False$ and stores $\Info(F)$ in that node.
It is left to describe how the root of a part gets $\Info(F)$ for each $F$ whose $\Info(F)$ should reside in that part.

\subsubsection{Storing $\Info$ in  partition $\Top$}
\label{subsub:storing-1st}
A part $P''$ in  partition $\cP''$ contains precisely one red fragment $F_{red}$. Hence, we call such fragments {\em red-centered}.
Consider a part $P$ in  partition $\Top$ that was created from a red-centered part  $P''\in\cP''$.
Recall that such a part $P$ should store only the $\Info$ of top fragments it intersects. Since each such top fragment is an ancestor fragment of $F_{red}$, we let  part  $P$ store the $\Info$ of all ancestor fragments of $F_{red}$. Hence, the set of $\Info$ stored at $P$ includes the  $\Info$  of all fragments $P$ intersects, but may include more $\Info$'s (of fragments
intersecting $P''$ but not $P$). Nevertheless, note that,  for every $j$, these other $\Info$'s correspond to at most one fragment in level $j$. This follows simply from the fact that  $F_{red}$ intersects at most one level $j$ fragment (see Claim \ref{cla:red-blue-intersect}). Recall also that the root of the part knows it is a root of a part (by comparing its $\Top-\Roo$ variable with its identity), and every node knows which part it belongs to (again, using its $\Top-\Roo$ variable)
as well as who are its parent and children in the part. (The latter information a node can deduce by reading each tree neighbour.)

The root of the final tree $T$ starts a $\MULTI$  over $T$. Fix a level $j$.
The $j$th wave of the multi-wave, which we term $\SENDANCESTORINFO(T,j$), signals the roots of every top fragment $F^j$ of level $j$ to obtain the information $\Info(F^j)$ and to send it to the roots of the parts of partition $\Top$ intersecting $F^j$.
 Consider a root $v_j$ of such a fragment $F^j$ who receives the signal of $\SENDANCESTORINFO(T,j)$.
First, to obtain $\Info(F^j)$, node $v_j$ must find the weight of
 $e$, the minimum outgoing edge of $F^j$. Recall that the endpoint $u\in F_j$ of $e=(u,v)$ can identify it is the endpoint
 using the $j$'th position in  $\EndP_u$, and can identify which of its incident edges is $e$.
So, node $v_j$ starts another Wave\&Echo bringing the weight of $e$ to $v_j$.
(Note that $v_j$ is the root of a single fragment in level $j$, though it may be the root of other fragments in other levels; hence, at the time of the wave of level $j$, it handles the piece of only one fragment, namely, $F_j$; hence, not congestion arises).
When this wave terminates, $v_j$ sets  $\Info(F^j)=(\id(v_j), \omega(e))$ and starts  another Wave\&Echo, called $\ANCESTORINFO(F^j, j)$, conveying $\Info(F^j)$ to roots of  the parts of partition $\Top$ intersecting $F^j$.
 To implement  this, $v_j$, the root of $F_j$, first broadcasts $\Info(F^j)$ to the nodes of $F^j$. At this point, each node in $F_j$ knows $\Info(F^j)$.  Next, our goal is to deliver $\Info(F^j)$ to the roots of parts in partition $\Top$ intersecting $F^j$.
However, note that since $F_j$ is a subtree, and all parts are subtrees, the roots of the parts of partition $\Top$ intersecting $F^j$ are all contained in $F_j$, except maybe the root $u_j$ of the part containing $v_j$. So, by now, all roots of parts in partition $\Top$ intersecting $F^j$, except maybe $u_j$, already know $\Info(F^j)$.
To inform $u_j$ it suffices to deliver $\Info(F^j)$ up the tree, from $v_j$ to $u_j$. Since, all roots of parts in $\Top$, and  $u_j$ in particular, know they are roots, this procedure is trivial. Finally, to complete the wave at level $j$, a root of a $\Top$ part receiving   $\Info(F_j)$ stores it in its part as described at top of this section (that is, storing each piece at a node in the part, following a DFS order).

Note, that since the diameter of a part in $\Top$ is of length $O(\log n)$, the wave of level $j$ can be implemented in $O(|F_j|+\log n)=O(2^j+\log n)$ time. Altogether, the $\MULTI$  over $T$ is completed by time $$O(\sum_{j=1}^{\log n} 2^j+\log n)=O(n).$$

\subsubsection{Storing $\Info$ in partition $\Bottom$}
\label{subsub:storing-2nd}

Recall that a part in partition $\Bottom$ is a fragment of size $O(\log n)$. The root of such a part $P$ collects the $\Info$ of fragments in $P$ of each level $i$ by issuing a Wave\&Echo for level $i$.
The weight of the minimum outgoing edge of each fragment $F^i$ of level $i$ is then collected by the root of $F^i$. This ensures that the $\Info(F^i)$ for each fragment $F^i$ of level $i$ in the fragment arrives at $F^i$'s root. Finally, the Wave\&Echo collects the $\Info$s from the roots of the fragments in the $\Bottom$ part to the root of the part. It is easy to see the following.

\begin{claim}
\label{claim:initial-trains}
The initialization of  the trains information described in Section \ref{app:neighbourhoods} can be done in time $O(n)$ and memory size $O(\log n)$.
\end{claim}
The next corollary that summarizes  this section follows from Lemma \ref{lem:lineal-marker} and Claims \ref{claim:top-bottom-assignment} and   \ref{claim:initial-trains}.
\begin{corollary}\label{cor:marker-time}
The marker algorithm $\cM$ can be implemented using memory size $O(\log n)$ and $O(n)$ time.
\end{corollary}

\section{Viewing distributed information}
\label{sec:viewing}
\label{sub:3.2}

We now turn to the verifier algorithm of part of the proof labeling scheme that verifies  the Minimality property.

Consider a node $v$ and a fragment $F_j(v)$ of level $j$ containing it.
Recall that $\Info(F_j(v))$ should reside permanently in some node of a part $P$ to which $v$ belongs.
This information should be compared at $v$ with the information $\Info(F_j(u))$ regarding a neighbour~$u$ of $v$, hence both these pieces
must somehow  be ``brought'' to $v$.  The process handling this task contains several components.
The first is called a ``train'' and is responsible for moving the pieces' pairs $\piece(i)$ through $P$'s nodes, such that each node does not hold more than $O(\log n)$ bits at a time, and such that in short  time, each node in $P$ ``sees'' all pieces, and in their correct order.
(By short time, we mean $O(\log n)$ time in synchronous networks, and $O(\log^2 n)$ time asynchronous networks.)

Unfortunately, this is not enough, since $\Info(F_j(v))$ may arrive at $v$ at a different time than $\Info(F_j(u))$ arrives at $u$, hence some synchronization must be applied. Further difficulties arise
 from the fact that
the neighbours of a node~$v$ may belong to different parts, so different trains pass there.  Note that~$v$ may have many neighbours, and we would not want to synchronize so many trains.

A first idea  to obtain synchronization would have been to utilize  delays of trains.
However,  delaying trains at different nodes could accumulate, or could even cause deadlocks.
 Hence, we avoid delaying trains almost completely. Instead, each node $v$ repeatedly samples a piece from its train, and synchronizes the comparison of this piece with pieces sampled by its neighbours, while both trains advance without waiting.  Perhaps not surprisingly, this synchronization turns out to be easier in synchronous networks, than in asynchronous ones.

   This process presented below assumes that no fault occurs. The detection of faults is described later. 
   
\subsection{The trains}
\label{sub:the-trains}
For simplicity, we split the task of a train into two subtasks, each performed repeatedly --
the first, {\em convergecast}, moves (copies of) the pieces one at a time {\em pipelined}  from their permanent locations to $r(P)$, the root of part $P$,  according to the DFS order. (Recall, $\DFS(i)$   stores  permanently the $i$'th piece of  $\Info(P)$.)

\begin{definition}
 A {\em cycle} is a consecutive delivery of  the $k$ pairs of pieces $\piece(1), \piece(2),\cdots,\piece(k)$ to $r(P)$.
\end{definition}
Since we are concerned with at most
  $k\leq 2\log  n$  pairs  of pieces, each cycle   can be performed in
 $O(\log n)$ time.
The second subtask, {\em broadcast}, broadcasts each piece from $r(P)$ to all other nodes in $P$ (pipelined).
This subtask can be performed in $D(P)=O(\log n)$ time, where $D(P)$ is the diameter of $P$.
We now describe these two subtasks (and their stabilization) in detail.

Consider a part  $P$ (recall, a part is a subtree).
The (pipelined) broadcast in $P$ is the simpler subtask. Each node contains a broadcast buffer for the current broadcast piece, and the node's children (in the part) copy the piece to their own broadcast buffer. When all these children of a node acknowledge the reception of the piece, the node can copy the next piece into its broadcast buffer.
Obviously, this process guarantees that the broadcast of each piece is performed in $D(P)=O(\log n)$ time, where $D(P)$ is the diameter of $P$.

We now describe the convergecast subtask.
Informally, this is a recursive process that is similar to a distributed DFS. The subtask starts at the root. Each node $v$ which has woken-up, first wakes-up its first child (that is, signals the first child to start). When the first child $u_1$ finishes (delivering to $v$ all the pieces of information in $u_1$'s subtree), then $v$ wakes-up the next child, and so forth.

Each node holds  two buffers of $O(\log n)$ bits each for two pieces of the train, besides its own piece (that it holds permanently). The node uses one of these buffers, called the {\em incoming car},
 to read a piece from one of its children, while the other buffer, called the {\em outgoing car} is used to let its parent (if it has one) read the piece held by the node.
A node $v\neq r(P)$ participates in the following simple procedure whenever signaled by its parent to wake-up. Let   $u_1,u_2,\cdots, u_{d}$ denote the children of $v$ in $P$ (if any exists), ordered according
 to their corresponding port-numbers at $v$ (i.e., for $i<j$, child $u_i$ is visited before $u_j$ in the DFS tour).\\

 \noindent
{\bf Train Convergecast Protocol} (performed at each node $v\neq r(P)$)\\
(*Using two buffers: incoming car and outgoing car *)
\begin{enumerate}
\item
Copy $v$'s (permanent)  piece into $v$'s outgoing car
      \item
            For $i=1$ to $d$ (*$d$ is the number of $v$'s children*)
      \begin{enumerate}
         \item
         Signal $u_i$ to start performing the train algorithm; (*wake-up $u_i$*)
         \item
         Repeat until $v$ receives a signal ``finished'' from $u_i$
             \begin{enumerate}
             \item
          Copy the  piece from the outgoing car of $u_i$ to $v$'s incoming car
             \item
             Wait until $v$'s outgoing car is read by its parent (*to accomplish that, $v$ reads the incoming car of its parent and compares it with its outgoing car *)
             \item
             Move the piece from the incoming car to outgoing car (and, subsequently, empty the content of the incoming car);
       \end{enumerate}
      \end{enumerate}
   \item
      Report ``finished'' to parent;

\end{enumerate}
The train Convergecast protocol of the  root $r(P)$ is slightly different. Instead of waiting for its parent to read each piece, it waits for the train Broadcast protocol (at the root) to read the piece to its own buffer. Instead of reporting ``finished'' to its parent, it generates a new start to its first child.
\begin{theorem}\label{thm:train}
Let  $t_0$ be some time when
the root $r=r(P)$ of $P$ initiated the ``For'' loop of the train  Convergecast protocol.
Each node in $P$ sees the pieces in the cycle  $\{\piece(1), \piece(2), \cdots, \piece(\phi(P))\}$ in $O(\log n)$ time  after~$t_0$ in synchronous networks and in $O(\log^2 n)$ time  after~$t_0$ in asynchronous networks.
\end{theorem}

\begin{proof}
First observe that the train broadcast in a leaf node of the part who received a piece from its parent, does not need to pass that piece to any further children.
Hence the train process does not incur a deadlock.

As mentioned before, once the root sees a piece, the broadcast protocol guarantees that this piece is delivered to all nodes in the part in $D(P)=O(\log n)$ time.
Let $\tau'$ denote the maximal time period between two consecutive  times that the broadcast protocol at the root  reads the buffer of the convergecast protocol to take a new piece
 (a piece is actually taken only if the convergecast has managed to bring there a new piece, after the broadcast process took the previous one).
Now, denote $\tau=\max\{1,\tau'\}$.

\begin{observation}
In synchronous networks, we have $\tau=1$. In asynchronous networks, we have $\tau\leq D(P)=O(\log n)$.
\end{observation}
The first part of the observation is immediate. To see why the second part of the observation holds, note that
by the definition of time,  it takes $O(D(P))$ for a chain of events that transfer a piece to a distance of $D(P)$, in the case that all the buffers on the way are free; note that there is no deadlock and no congestion for information flowing down the tree, away from the root; this can be seen easily by induction on the distance of a broadcast piece from the furthest leaf; clearly, if the distance is zero, the piece is consumed, so there is already a room for a new piece; the rest of the induction is also trivial.\\

We shall measure the time
in {\em phases}, where  each phase consists of $\tau$  time units.   Let us start counting the time after time $t_0$, that is, we say that phase 0 starts at time $t_0$.
Our goal now is to show that (for either synchronous or asynchronous networks), for each $1\leq i < \phi(P)$, piece  $\piece(i)$ arrives at the root within $O(\log n)$ phases.

We say that a node $v$ is {\em holding} a piece at a given time  if either (1) $v$ keeps the piece permanently, or (2)  at the given time,
 the piece resides in either $v$'s incoming car or its outgoing car.
 Consider now phase $t$.   For each~$i$, where
$1\leq i \leq \phi(P)$,  if the root $r$ held  $\piece(i)$ at some time between $t_0$ and
 the beginning of the phase~$t$, then we say that $i$ is {\em not $t$-relevant}. Otherwise,
$i$ is {\em $t$-relevant}.
For any $t$-relevant $i$, where
$1\leq i \leq \phi(P)$,
let $d_t(i)$
denote the smallest  DFS number of a node $v$ holding  $\piece(i)$
at the beginning of phase $t$.
That is,
$d_t(i) =
\min \{ \DFS(u) |~ u \;\; {\rm holds} \;\; \piece(i) \;\; {\rm at \;\; time} \;\; t \}$.
For any  $i$ that is not $t$-relevant, let $d_t(i)=0$.
The following observation is immediate.

\begin{observation}\label{obs:d}
At any time $t$,
\begin{itemize}
\item
For any  $1\leq i \leq \phi(P)$,  we have $d_t(i)\geq d_{t+1}(i)$ (in other words, $d_t(i)$ cannot increase with the phase).
\item
For any  $1\leq i < \phi(P)$, we have  $d_t(i)\leq d_t(i+1)$.
\end{itemize}
\end{observation}
Informally, the following lemma gives a bound for the delay of a piece as a result of processing previous pieces.

\begin{lemma}
\label{lem:train}
Let $x$ and $i$ be two integers such that
  $1\leq i\leq  x\leq \phi(P)$.
  Then,  $d_t(x)\leq i-1$ ~for~ $t\geq 3x-i$.
\end{lemma}
To prove the lemma,
first observe that the condition holds for the {\em equality} case, that is, the case where $i=x$. Indeed,
for each  $1\leq  x\leq \phi(P)$, the node holding $\piece(x)$ permanently is at distance at most $x-1$ from the root.
Hence,
$d_{0}(x)\leq x-1$. Now, the condition follows since, by Observation~\ref{obs:d}, $d_t(x)$ cannot increase with the phase.

We now prove the lemma using a double induction.
The first induction is on $x$. The basis of the induction, i.e., the case $x=1$, is trivial, since it reduces to the equality case $i=x=1$.

Assume by induction that the condition holds for $x-1$ and any $i$, such that $1\leq i\leq x-1\leq \phi(P)$.
 We now prove that the condition holds for $x$ and any  $1\leq i\leq x$. This is done using a   reverse induction on $i$.

 The basis of this (second) induction, i.e., the case $i=x$, is an equality case and hence, it is already known to
 satisfy the desired condition.  Now assume by induction, that the condition holds for $x$ and $i$, where  $2\leq  i\leq x$, and let us show that it holds also for  $x$ and $i-1$.

Let us first consider the case $i=2$. By the (first) induction hypothesis (applied with values $x-1$  and $i=1$), we know that

$$d_{t'}(x-1)\leq 0 \mbox{~~~~where~~~~} t'= 3x-4.$$
 Thus, at phase $t'=3x-4$,  piece
 $\piece(x-1)$ is not $t'$-relevant. That is, at that time,
  piece $\piece(x-1)$  is either in the outgoing car of the root $r(P)$ or in the root's incoming car.
  In the first case, the incoming car of the root is already empty at $t'$. Otherwise, recall that,
 by definition, the broadcast process at the root consumes a piece from $r(P)$'s outgoing car every phase (if there is a new piece there it has not taken yet).
 Hence, the outgoing car at $r(P)$ is consumed by phase $t'+1$. By that phase, the root notices the piece is consumed,
  deliverers the content of its incoming car (namely, piece $\piece(x-1)$)
 to its outgoing car, and empties its incoming car.

 On the other hand, by the second induction hypothesis, $d(x) \le 1$ at the beginning of phase $3x-2=t'+2$. That is, $\piece(x)$ is at some child $v$ of the root. By the second part of Observation \ref{obs:d}, node $v$ is the child the root reads next, and, moreover no piece other than $\piece(x)$ is at the outgoing car of $v$.
 If at the beginning of phase $t'+2$, piece $\piece(x)$ is at the outgoing car of $v$, then the piece reaches the incoming car of the root already at phase $t'+2$. Otherwise, by at most phase $t'+3$, node $r(P)$ has a copy of $\piece(x)$ in its incoming car.
   This means that $d_{t''}(x)\leq 0$, where $t''=t'+3\leq 3x-(i-1)$, as desired.

 Now consider the case that $2<i$.  By the second induction hypothesis, we have $d_{t'}(x)\leq i-1$, where $t'=3x-i$.
If $d_{t'}(x)\leq i-2$ at the beginning of  phase $t'$ then we are done. Otherwise,
let $v$ be the node holding $\piece(x)$ at the beginning of phase  $t'$
such that the distance (on the tree) of $v$ from $r$ is $i-1$.
 Let $u$ be $v$'s parent. Our goal now is to show that $u$ holds $\piece(x)$ by phase $t'+1$.
The (first) induction hypothesis implies  that the condition holds for the pair $x-1$ and $i-2$. That is,
$$d_{t'-1}(x-1)\leq i-3.$$ 
 Thus, $\piece(x-1)$ has already been copied to $u$'s parent $w$. The only reason
 $\piece(x-1)$ may be stuck at $u$ (perhaps at both  the incoming and outgoing cars of $u$) at phase $t'-1$,
  is that  $u$ has not observed yet that its parent $w$ actually already copied $\piece(x-1)$. This is observed by $u$ by phase $t'$
 (when $u$ observes this, it empties the content of its incoming car).
Thus, by phase $t'+1$, node $u$ has a copy of $\piece(x)$, as desired. This concludes the proof of the lemma.

The theorem  now follows from the lemma and from the fact that $\phi(P)=O(\log n)$.
\end{proof}


\paragraph{Recognizing membership to arriving fragments: }
Consider now the case that a piece containing $\Info(F)$ carried by the broadcast wave arrives at some node $v$.
Abusing notations, we refer to this event by saying that fragment $F$ {\em arrives} at $v$.
Recall that $v$ does not~have~enough memory to remember the identifiers of all the fragments  containing it. Thus, a mechanism for letting~$v$ know~whether the arriving fragment $F$ contains $v$ must be employed. Note that the level $j$ of $F$ can be extracted from~$\Info(F)$, and recall that it is
already ensured that $v$ knows whether
it is contained in some level $j$ fragment. Obviously,  if $v$ is not contained in a level $j$ fragment then   $v\notin F$. If $F_j(v)$ does exist, we now explain
how to let  $v$ know whether~$F=F_j(v)$.

Consider first a train in a part  $P\in \Top$. Here, $P$ intersects  at most one level $j$ top fragment, for each level~$j$ (see Lemma~\ref{lem:top}). Thus, this train carries at most one level $j$  fragment $F_j$. Hence, $F_j=F_j(v)$ if and only if $F_j(v)$ exists.

Now consider a train in a part  $P\in \Bottom$.
In this case, part $P$ may intersect several bottom fragments of the same level.
 To allow a node $v$ to detect whether a fragment $F_j$ arriving  at $v$ corresponds to fragment $F_j(v)$, we slightly refine the above mentioned  train  broadcast mechanism as follows. During the broadcast wave, we attach a flag to each  $\Info(F)$, which can be either ``on'' or ``off'', where initially, the flag is ``off''.
 Recall that $\Info(F)$  contains the identity $\id(r(F))$ of the root $r(F)$ of $F$.
When the broadcast wave reaches  this root  $r(F)$ (or, when it starts in $r(F)$ in the case that $r(F)=r(P)$), node $r(F)$  changes the
flag to ``on''.
In contrast, before transmitting the broadcast wave from a leaf $u$ of $F$ to $u$'s children in $T$ (that do not belong to $F$), node $u$ sets the flag to ''off''.
That way,  a   fragment $F$ arriving at a node $v$ contains $v$ if and only if the corresponding flag is set to ``on''.
(Recall that the data structure lets each node know whether it is a leaf of a level $j$ fragment.) This process allows each node $v$ to detect whether $F=F_j(v)$.

To avoid delaying the train beyond a constant time, each node multiplexes the two trains passing via it.
 That is, it passes one piece of one train, then one piece of the other.

\subsection{Sampling and synchronizing}
\label{sub:Sampling and synchronizing}
Fix a partition (either $\Top$ or $\Bottom$), and a part $P$ of the partition.
Node $u\in P$ maintains two variables: $\Ask(u)$ and $\Show(u)$, each for holding one piece $\Info(F)$. In $\Ask(u)$, node $u$ keeps $\Info(F_j(u))$
for some $j$, until $u$ compares the piece $\Info(F_j(u))$ with the piece $\Info(F_j(v))$, for each of its neighbours $v$.
Let $\cE(u,v,j)$ denote the event that node $u$ holds $\Info(F_j(u))$ in $\Ask(u)$ and sees $\Info(F_j(v))$ in $\Show(v)$.
(For simplicity of presentation, we consider here the case that both $u$ and $v$ {\em do} belong to some
 fragments
 of level $j$;
 otherwise,
  storing and comparing the information for a non-existing
 fragments is trivial.)
 For any point in time~$t$, let $C(t)$ denote the minimal  time interval $C(t)=[t,x(t)]$  in which
 every event of the type $\cE(u,v,j)$ occurred.
 For the scheme to function, it is crucial that  $C(t)$ exists for every time~$t$. Moreover,
to have a fast scheme, we must ensure that $\max_t |C(t)|$ is small.

Recall that the train (that corresponds to $P$)  brings the pieces $\Info(F)$  in a cyclic order.
When $u$ has done comparing  $\Info(F_j(u))$ with  $\Info(F_j(v))$ for each of its neighbours $v$, node
$u$ waits until it receives (by the train)
the first piece $\Info(F)$ following $\Info(F_j(u))$ in the cyclic order, such that $F$ contains $u$ (recall that $u$ can identify this~$F$). Let us denote the level of this next fragment $F$ by $j'$, i.e.,
$F=F_{j'}(u)$.
 Node $u$ then removes $\Info(F_{j}(u))$ from  $\Ask(u)$ and stores $\Info(F_{j'}(u))$  there instead,
 and so~forth. Each node  $u$ also stores some piece $\Info(F_i(u))$ at
$\Show(u)$ to be seen by its neighbours.  (Note that the value at $\Show(u)$ may be different than the one at $\Ask(u)$.)

Let us  explain the comparing mechanism.
Assume that everything functions correctly. In particular, assume that the partitions and the distribution of the information are as described above, and the trains function correctly as well.
Let us first focus our attention on the simpler and seemingly more efficient synchronous case.

\subsubsection{The comparing mechanism in synchronous networks}
Fix  a node $v$.
 In a synchronous network, node $v$ sees $\Show(u)$ in every pulse, for each neighbour $u$.
 Let every node $u$ store in $\Show(u)$ each piece that arrives in the train (each time, replacing the previous content of $\Show(u)$).
 Hence, by Theorem \ref{thm:train}, given a level $j$, node $v$ sees  $\Info(F_{j}(u))$ (if such exists) within $O(\log n)$ time.
Put differently, if $v$ waits some  $O(\log n)$ time
(while $\Info(F_j(v))$ is in $\Ask(v)$), node $v$  sees $\Info(F_j(u))$ in $\Show(u)$ for each neighbour $u$.
(We do not assume that $u$ keeps track of which neighbours $v$ has already seen $\Info(F_j(u))$; node $v$ simply waits
sufficient time -- to allow one cycle of the train, while looking at its neighbours, looking for their $\Info$ for level $j$.)
Subsequently, node $v$ waits another $O(\log n)$ rounds until the train brings it $\Info(F_{j'}(v))$ and stores it in  $\Ask(v)$, and so forth.
In other words, we have just established that event  $\cE(v,u,j')$ occurs within $O(\log n)$ time after $v$ stores $\Info(F_{j'}(v))$ in $\Ask(v)$, which happens
$O(\log n)$ time after event $\cE(v,u,j)$, and so forth.
 The time for at most $\log n+1$ such events to occur (one per  level $j$)  is $O(\log^2 n)$.

\begin{lemma}
\label{lem:synchronous-time}
In a synchronous environment,  for each node $v$ and its neighbour $u$, all events of type $\cE(v,u,j)$ (for all levels $j$) occur within time
 $O( \log^2 n)$.
\end{lemma}

\subsubsection{The comparing mechanism  in asynchronous networks}
 In an asynchronous network, without some additional kind of a handshake,
 node $u$ cannot be sure that the piece in $\Show(u)$ was actually seen by its neighbours.
 (Intuitively, this is needed, so that $u$ can replace the piece with the next one.)
 Moreover, it is not easy to perform such handshakes with all of $u$'s neighbours, since
 $u$ does not have enough memory to keep track on which of its neighbours $v$
  has seen the piece and which has not yet.
 First, let us describe a simple, but somewhat inefficient handshake solution. A more efficient one is presented later.

 Each node $v$, holding some piece $\Info(F_j(v))$ in $\Ask(v)$, selects a neighbour $u$ and acts as a ``client'': that is, node $v$ writes in its register $\Keep$ the pair $(\id(u),j)$.
 Node $v$ then looks repeatedly at $\Show(u)$ until it sees $\Info(F_j(u))$ there.
 At the same time, each node $u$  also has a second role -- that of a ``server''.
 That is, each node rotates these two roles: it performs one atomic action as a server and one as a client.
  Acting as a server, $u$ selects a client to serve (in a round robin order). If the client has written
 some $(\id(w),i)$ in the client's $\Keep$, for $w\neq u$, then $u$ chooses another client.  On the other hand, if the client wrote $(\id(u),j)$
in the client's $\Keep$, then $u$ waits until it receives by the train $\Info(F_j(u))$ and stores it in $\Show(u)$. A trivial handshake then suffices for $u$ to know that
this value has been read by the client. Node $u$, in its role as a server, can then move to serve its next neighbour, and node $v$, in its role as a client, can move on to the next server. In particular,
 if the client $v$ has already received service from all its neighbours for $\Info(F_j(v))$, then $v$ waits until the train  brings it the next piece $\Info(F_{j'}(v))$  that $v$ needs to compare.

Consider now the time a client $v$ waits to see $\Info(F_j(u))$  for one of its neighbours $u$. Before serving $v$, the server $u$ may serve $O(\Delta)$ neighbours. By Theorem \ref{thm:train} (applied for the asynchronous setting), each service takes $O(\log^2 n)$
time. In addition, the client needs services from $\Delta$ servers, and for $O(\log n)$ values of $j$. The total time for all the required events to happen in this simple handshake mechanism
is, thus,
$O(\Delta ^2 \log^3 n)$.

Let us now describe the more efficient asynchronous comparison mechanism that requires only $O(\Delta \log ^3 n)$ time.
Before dwelling into the details of the comparison mechanism, let us first describe a difference in the way we employ the train.
Recall that in the simple  solution above (as well as in the synchronous case), the movement of  trains was independent from the actions of the comparison mechanism, and hence, by Theorem \ref{thm:train}, each train finishes a cycle in $O(\log^2 n)$
time.
In contrast,  a train here may be delayed by the nodes it passes, in a way to be described.
Crucially, as we show later, the delay at each node is at most some constant time  $c$, and hence,  the time a train finishes a cycle remains asymptotically the same, namely, $O(\log^2 n)$.

As before, node $v$, holding $\Info(F_j(v))$ in $\Ask(v)$ chooses a server $u$ among $v$'s neighbours
and reads $\Show(u)$.
Another small, but crucial addition to the actions taken in the simple procedure, is the following: if, when reading $\Show(u)$, node $v$ reads $\Info(F_j(u))$, then  $\cE(v,u,j)$ occurred, and $v$ moves on to read another neighbour.
This is illustrated in Figures \ref{fig:handshake1} and \ref{fig:handshake2}.

Only in the case that $\Info(F_j(u))$ is not at $\Show(u)$ at that time
(see Figure \ref{fig:handshake3}), node $v$ sets $\Keep(v) \gets (\id(u), j)$
(see Figure \ref{fig:handshake4}).
In this case, we say that $v$ {\em files a request for} $j$ {\em at} $u$. This request stays filed until the value of $\Show(u)$ is the desired one and  $\cE(v,u,j)$ occurs.
 Similarly to the synchronized setting, in the case that~$v$ has just finished seeing
 $\Info(F_j(u))$ in every  neighbour $u$, node $v$
first waits until it gets by the train, the next piece $\Info(F_{j'}(v))$ in the cycle, and then  puts $\Info(F_{j'}(v))$ as the new content of
$\Ask(v)$.

Now consider any node $u$ in its role as a server. It reads all the clients.
(Recall that the ideal time complexity assumes this can be performed in one time unit.)
When node $u$ receives $\Info(F_j(u))$ from the train, it puts this value in
$\Show(u)$. It now delays the train
 as long as it sees any client $v$ whose  $\Keep(v) = (\id(u), j)$. In particular, node $u$
keeps $\Info(F_j(u))$ in $\Show(u)$ during this delay time period.
    If $u$ has not read any neighbour $v$ such that $\Keep(v) = (\id(u), j)$, then $u$ stops delaying the train, waits for receiving the next piece  $\Info(F_{j'}(u))$ from the train, and uses it to replace the content of $\Show(u)$.

\begin{figure}
\centering
\includegraphics[scale=.4]{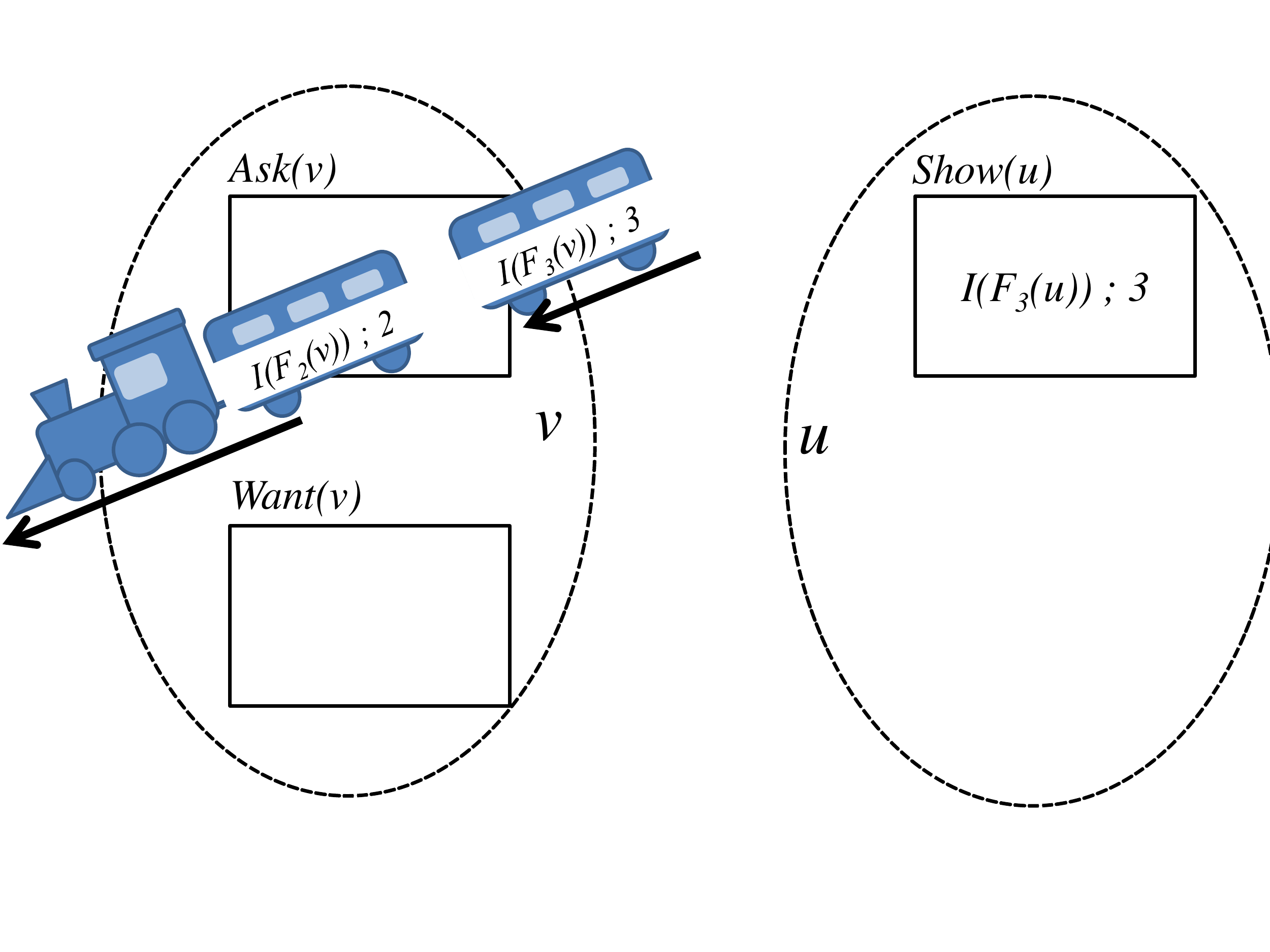}
\caption{Node $v$ receives the next piece (for $j=3$) to compare.
}
\label{fig:handshake1}
\end{figure}

\begin{figure}
\centering
\includegraphics[scale=.4]{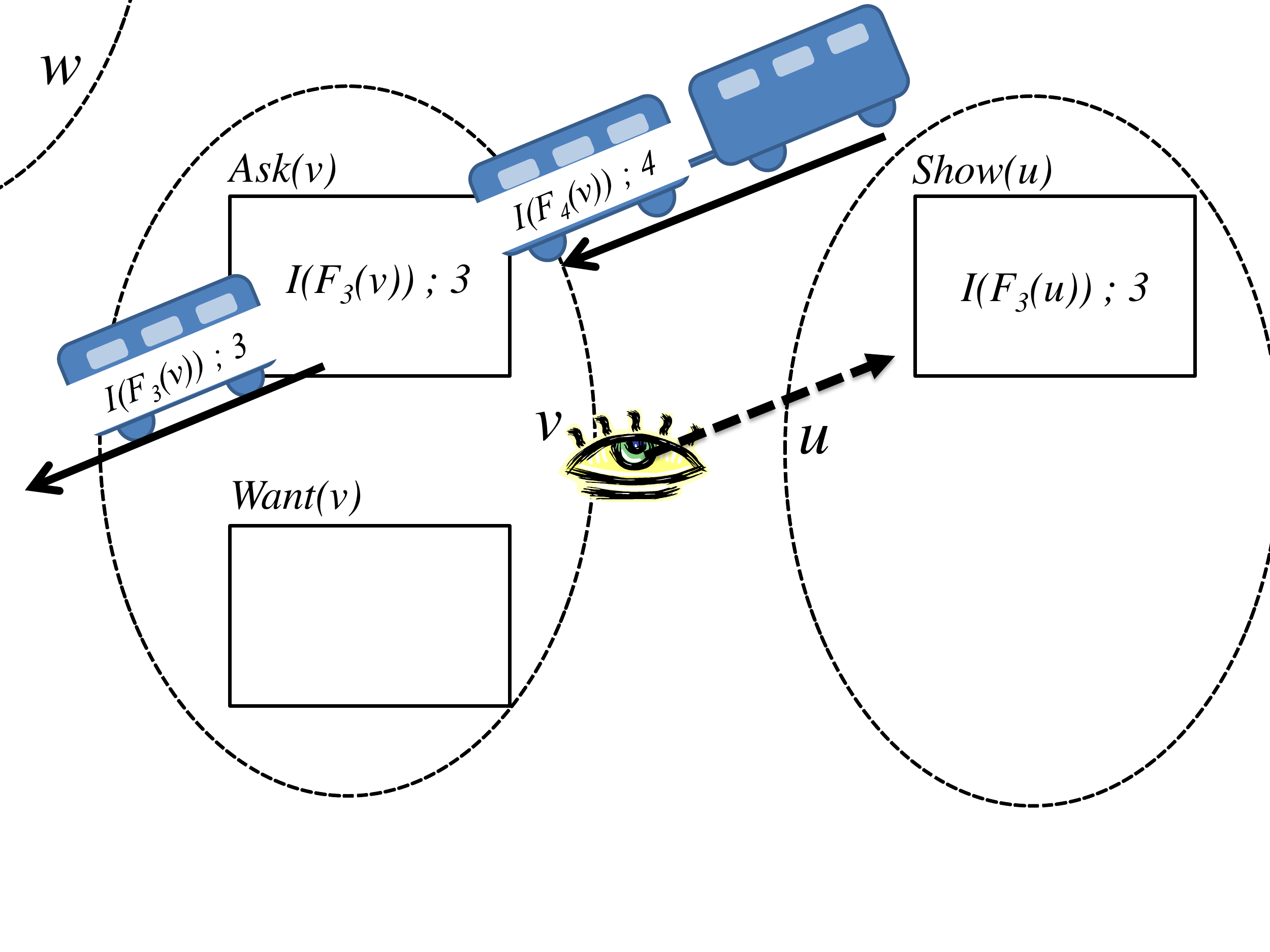}
\caption{First case,  $\cE(v,u,3)$ occurred the first time $v$ reads $\Info(F_3(u))$. Next, $v$ may look at its next node, $w$.}
\label{fig:handshake2}
\end{figure}

\begin{figure}
\centering
\includegraphics[scale=.4]{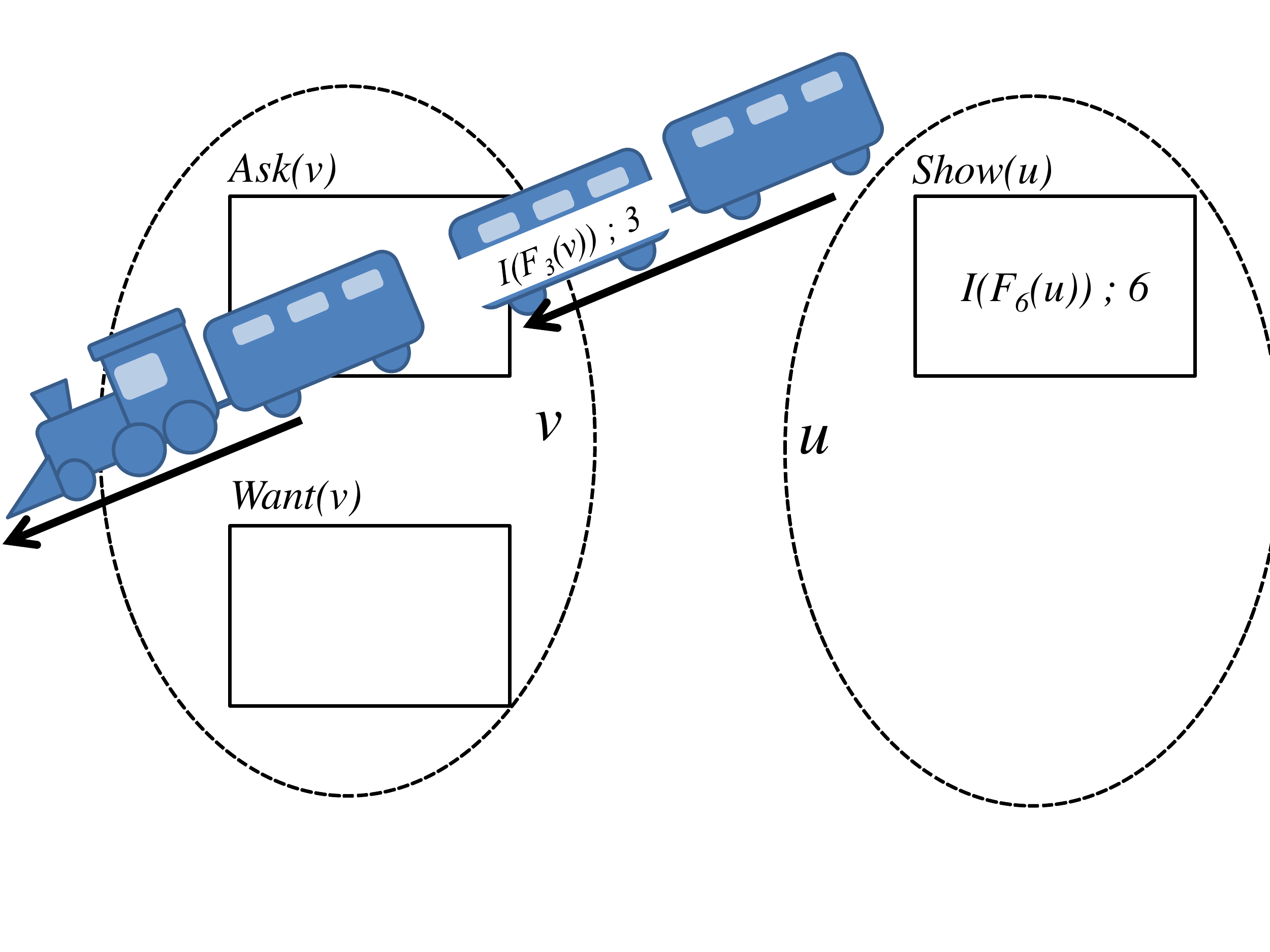}
\caption{Second case,  $\cE(v,u,3)$ does not occur immediately.}
\label{fig:handshake3}
\end{figure}

\begin{figure}
\centering
\includegraphics[scale=.4]{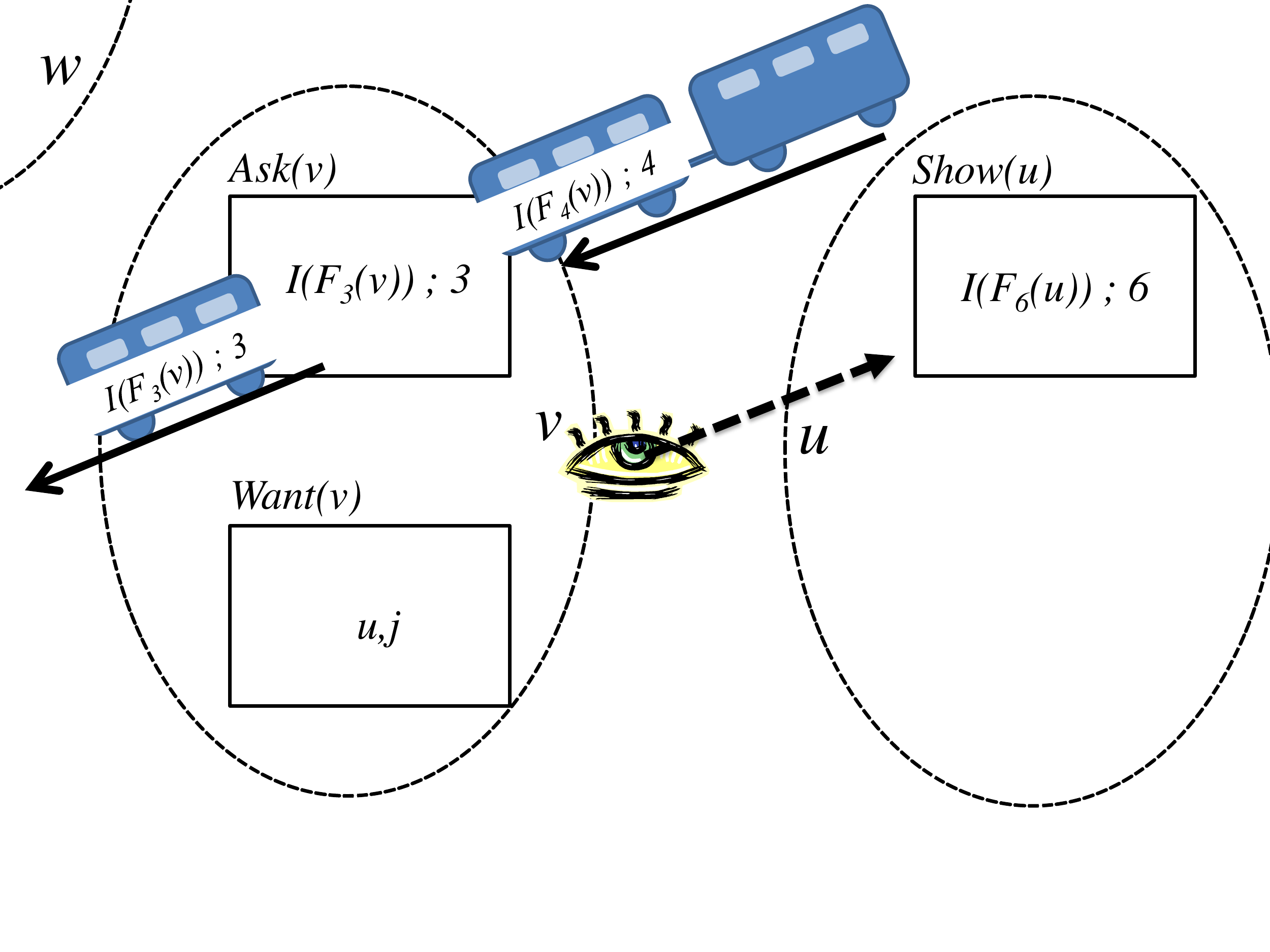}
\caption{Node $v$ files a request at $u$ for $j=3$.}
\label{fig:handshake4}
\end{figure}

\begin{figure}
\centering
\includegraphics[scale=.4]{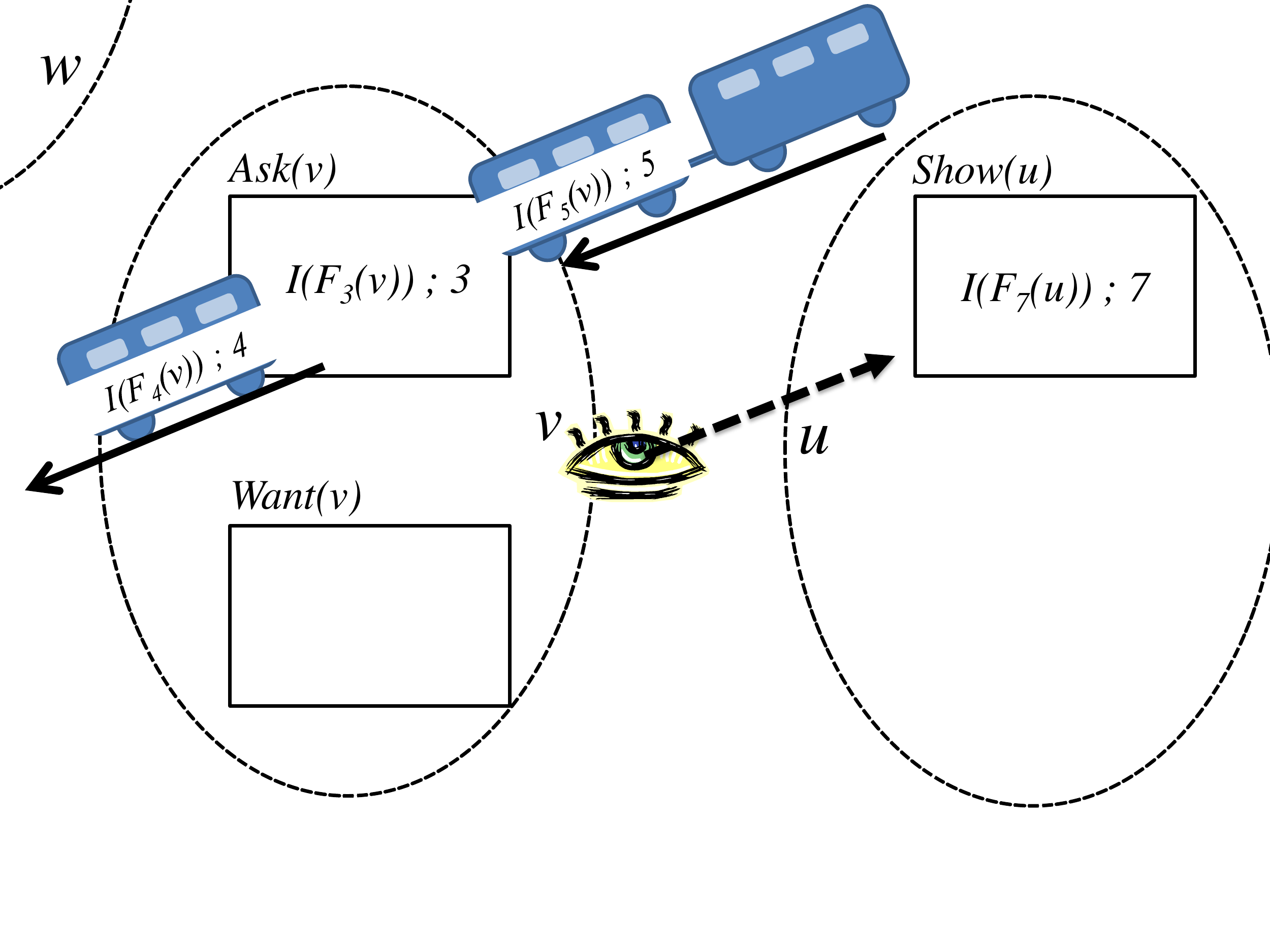}
\caption{The trains at $v$ and at $u$ do not stop.}
\label{fig:handshake5}
\end{figure}

\begin{figure}
\centering
\includegraphics[scale=.3]{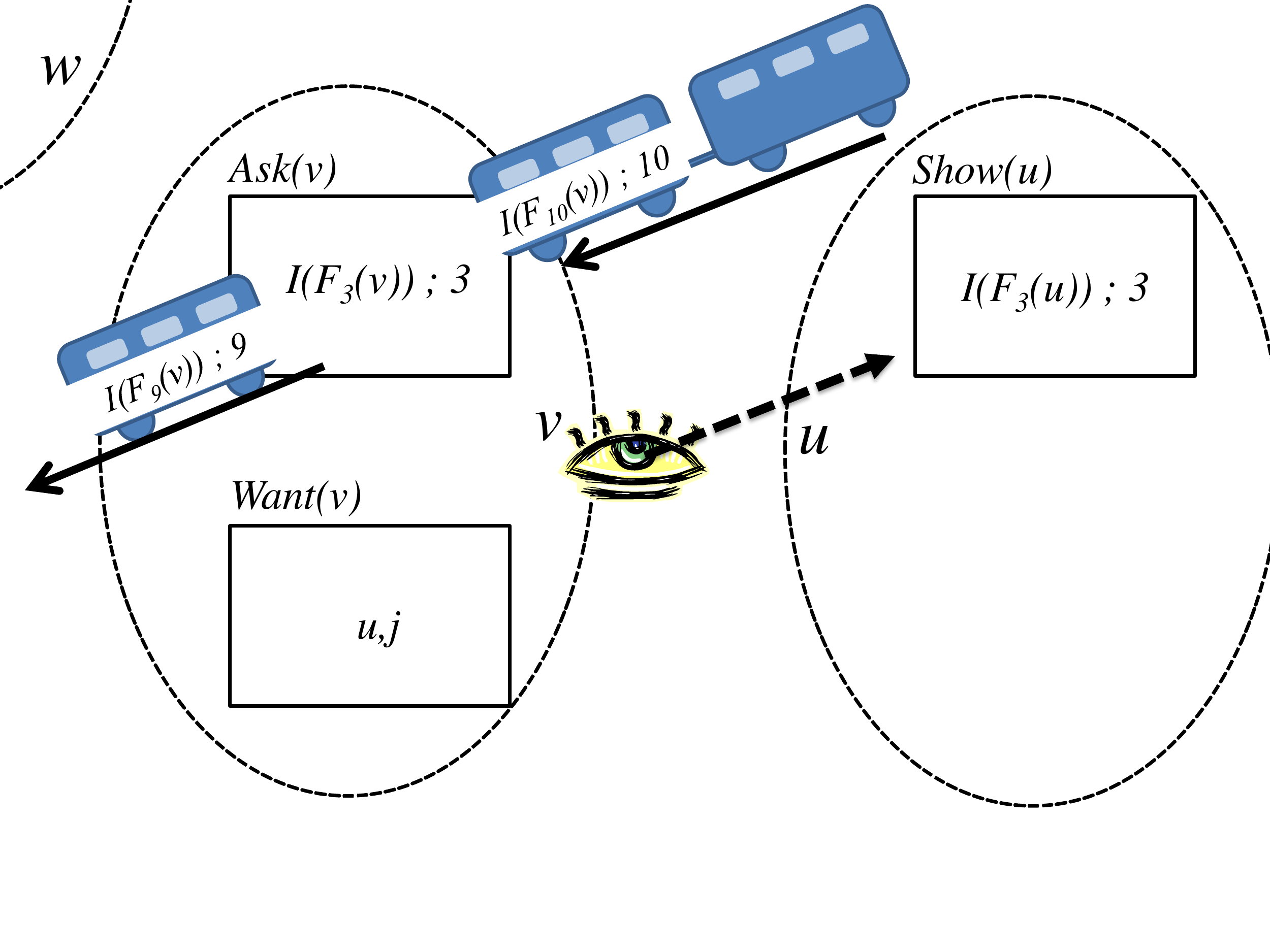}
\caption{The server $u$ received the requested piece at last.}
\label{fig:handshake6}
\end{figure}
 We define the
   $\Ask$ cycle of a node $v$. This is the time interval starting at the time a client $v$  replaces
  the content of $\Ask(v)$ from
  $\Info(F_{j^{max}}(v))$ to
  $\Info(F_{j^{min}}(v))$, and until (and excluding) the time $v$ does that again. Here, ${j^{max}}$ is the highest level of a piece
    in that train, such that
    $F_{j^{max}}(v)$ exists, and ${j^{min}}$ is the smallest level of a piece in that train,
    such that  $F_{j^{min}}(v)$ exists.

 \begin{lemma}
 \label{lem:cycle-time}
 The total length of a $\Ask $ cycle of a node $v$ is
 $O(\Delta \log^3 n)$.
 \end{lemma}
 \begin{proof}
Fix a node $u$ and let $t_u^1$ be some time that $u$ starts storing  $\Info(F_j(u))$ in $\Show(u)$, for some level $j$; moreover, $\Info(F_j(u))$
is stored there until some
time $t_u^2$
when $u$ replaces the content of $\Show(u)$ again.

Recall, node $u$ delays the train and
keeps $\Info(F_j(u))$ in $\Show(u)$  as long as it sees any client $v$ such that  $\Keep(v) = (\id(u), j)$; when it sees that no such neighbour $v$ exists, it stops delaying the train and waits for the train  to deliver it the next piece $\Info(F_{j'}(u))$ to be used for replacing the content of
$\Show(u)$.
We now claim that the delay time period at node $u$ is at most some constant time.
To prove that, we first show that there exists a constant $c$ such that
no client $v$ has  $\Keep(v) = (\id(u), j)$
 in the time interval $[t_u^1+c, t_u^2]$.  Indeed, for each neighbour $v$ of $u$, let $t_v$ be the first time after $t_u^1$ that $v$ reads the value of
$\Show(u)$. Clearly, there exists a constant $c$ (independent of $u$ and $v$) such that $t_v\in [t_u^1, t_u^1+c]$. Right at time $t_v$, the content of $\Keep(v)$ stops being  $(\id(u), j)$
(if it were before), since
  $\Info(F_j(u))$ is  the value of $\Show(u)$
 during the whole time interval $[t_u^1, t_u^2]$.
 Moreover, during the time interval $[t_v, t_u^2]$, node $v$ does not  file a request  for $j$ at $u$, since again,
 whenever it reads $\Show(u)$
 during that time interval, it sees $\Info(F_j(u))$.
Hence, no client $v$ has  $\Keep(v) = (\id(u), j)$
 in the time interval $[t_u^1+c, t_u^2]$.
Now, from time $t_u^1+c$,   it takes at most some constant time to let $u$ observe that none of its neighbours $v$ has $\Keep(v) = (\id(u), j)$.
This establishes the fact that the delay of the train at each node is at most some constant. Hence, as mentioned before,
the time the train finishes a cycle is $O(\log^2 n)$. (It is also easy to get convinces that this delay does not prevent the train from being self-stabilizing.)

Next, consider the time that some node $v$ starts holding $\Info(F_j(v))$ in   $\Ask(v)$. Consider a neighbour~$u$ of~$v$.
The time it takes for $v$ until it  sees $\Info(F_j(u))$ in $\Show(u)$
is $O(\log^2 n)$.
Hence, a client $v$ waits $O(\log^2 n)$ for each request $v$ files at a server $u$ for a value $j$. The total time that $v$ waits for a service of $j$ at all the servers is then
$O(\Delta \log^2  n)$. From that time, $v$ needs to wait additional $O(\log^2 n)$ time  to receive from the train the next piece  $\Info(F_{j'}(v))$ (to replace the content of $\Ask(v)$).
Summing this over the $O(\log n)$ pieces in the cycle, we conclude that the total time of an $\Ask$ cycle of $v$ is
$O(\Delta \log^3 n)$.
 \end{proof}

 \begin{lemma}\label{lem:time-sync}
If (1)  two partitions are indeed represented, such that   each part of each partition  is of diameter $O(\log n)$,
and the number of pieces in a part is $O(\log n)$, and (2) the trains operate correctly, 
then the following holds.
\begin{itemize}
\item
In a synchronous network, $\max_t |C(t)|=O(\log^2 n)$.
\item
In an asynchronous network, $\max_t |C(t)|=O(\Delta\log^3 n)$.
 \end{itemize}

 \end{lemma}

\section{Local verifications}
\label{sub:utilizing}\label{sec:verifying}
\label{sub:3.3}
In this section, we describe the measures taken in order to make the verifier self-stabilizing.
That is, the train processes, the partitions,  and also, the pieces of information carried by the train may be corrupted by an adversary.
 To stress this point and avoid confusion, a piece of information of the form $z\circ j\circ\omega$,
 carried by a train, is termed the {\em claimed} information
 $\hInfo(F)$ of a fragment $F$ whose root $\id$ is $z$, whose level is $j$,
 and whose minimum outgoing edge is $\omega$. Note that such a fragment $F$ may not even exist, if the information is corrupted.
  Conversely, the adversary may also erase some (or even all) of such pieces corresponding to existing fragments. Finally, even correct pieces that correspond to existing fragments may not arrive at a node in the case that the
adversary corrupted the partitions or the train mechanism. Below we explain how the verifier does detect such undesirable phenomena, if they occur. Note that for a verifier, the ability to detect with assuming any
initial configuration
means that the verifier is self-stabilizing, since the sole purpose of the verifier is to detect.
We show, in this section, that if an MST is not represented in the network, this is detected. Since the detection time (the stabilization time of the verifier) is sublinear,  we still consider this detection as local, though some of the locality was traded for improving the memory size when compared with the results of \cite{KormanKutten07,KKP10}.

Verifying that {\em some} two partitions exist is easy. It is sufficient to (1) let  each node verify that its label contains the two bits corresponding to the two partitions; and (2) to have the root $r(T)$ of the tree verify that the value of each of its own two bits is 1.
(Observe that if these two conditions hold then (1)
 $r(T)$ is a root of one part in each of the two partitions; and (2) for a node $v\neq r(T)$, if one of these two bits in $v$ is zero, then $v$ belongs to the same part in the corresponding partition as its parent.)
Note that this module of the algorithm  self-stabilizes trivially in zero time.

It seems difficult to verify that the given partitions are as described in Section \ref{sec:partitions}, rather  than being two
arbitrary
partitions generated by an adversary.  Fortunately, this verification turns out to be unnecessary.
(As we shall see, if the components at the nodes do not describe an MST, no adversarial partitioning can cause the verifier to accept this as representing an MST; if partitions are represented, we just need to verify that a part is not too large for the time complexity).

First, for the given partitions, it is a known art to self-stabilize the train process.
 That is, the broadcast part of the train is a standard flooding, for which the self stabilization has been heavily studied, see, in particular, \cite{BDPV,CPVD}.
 For the convergecast, first, note that pieces are sent up the tree. Hence, they cannot cycle, and cannot get ``stuck''. Moreover, it is easy to get convinced that only pieces that are already in some buffer (either incomming, or outgoing, or permanent) can be sent. Finally, notice that
the order of the starting of the nodes is exactly the DFS order. The stabilization of the DFS process is well understood \cite{CollinDolevDFS}. It is actually easier here, since this is performed on a tree (recall that another part of the verifier verifies that there are no cycles in the tree).

 Finally, composing such self-stabilizing primitives in a self-stabilizing manner is also a known art, see e.g.~\cite{dim-composition, jalfon-composition,protocol-composition,katz-perry}.
In our context, once the DFS part stabilizes, it is easy to see the pieces flow up the tree stabilizes too. This leads to the following observation.

 \begin{observation}
 \label{obs:train-sync}
Starting at a time that is $O(\log n)$ after the faults in synchronous networks, and   $O(\log^2 n)$ time in asynchronous networks, the trains start delivering only pieces that are stored permanently at nodes in the part.
\end{observation}
{\em After} the trains stabilize (in the sense described in Observation \ref{obs:train-sync}), what we want to ensure at this point is
 that the set of pieces stored in a part (and delivered by the train) includes
all the (possibly corrupted) pieces of the form $\Info(F_j(v))$, for every $v$ in the part and for every $j$ such that $v$ belongs to a level $j$ fragment.
  Addressing this,
we shall show that the verifier at each node rejects  if it does not obtain all the required pieces eventually, whether the partitions are correct or not.
Informally, this is done as follows.  Recall  that each node $v$ knows the set $J(v)$ of levels $j$ for which there exists a fragment of level $j$ containing it, namely, $F_j(v)$. Using a delimiter (stored at $v$), we partition $J(v)$ to  $J_{\Top}(v)$ and $J_{\Bottom}(v)$;  where
 $J_{\Top}(v)$ (respectively, $J_{\Bottom}(v)$) is the set of levels $j\in J(v)$ such that $F_j(v)$   is top   (resp., bottom).

 Node $v$ ``expects'' to receive the claimed information
 $\hInfo(F_j(v))$ for $j\in J_{\Top}(v)$ (respectively, $j\in J_{\Bottom}(v)$) from the train
of the part in $\Top$ (respectively, $\Bottom$) it belongs to.

Let us now consider the part $P_{\Top}\in \Top$ containing $v$.
In correct instances, by the way the train operates, it follows that the levels of fragments arriving at $v$ should arrive in a strictly {\em increasing order} and in a {\em cyclical} manner, that is,  $j_1<j_2<j_3<\cdots< j_{a}, j_1<j_2<\cdots j_a,j_1\cdots$ (observe that $j_a=\ell$). Consider the case that the verifier at $v$ receives two consecutive pieces $z_1\circ j_1\circ\omega_1$
and $z_2\circ j_2\circ \omega_2$ such that $j_2 \leq j_1$.
The verifier at $v$ then  ``assumes'' that the event $S$ of the arrival of the  second piece
$z_2\circ j_2\circ \omega_2$ starts
a new cycle of the train. Let the set of pieces arriving at $v$ between two consecutive such $S$ events be named a {\em cycle set}.
 To be ``accepted'' by the verifier at $v$, the set of levels of the fragments arriving at $v$ in each cycle set must contain
$ {J}_{\Top}(v)$. It is trivial to verify this in two cycles after the faults cease. (The discussion above is based implicitly on the assumption that each node receives pieces infinitely often;
this is guaranteed by the correctness of the train mechanism, assuming that at least one piece is indeed stored permanently in $P_{\Top}$;
verifying this assumption is done easily by
the root $r(P_{\Top})$ of $P_{\Top}$, simply by verifying that $r(P_{\Top})$ itself does contain a piece.) Verifying the reception of all the pieces in a part in $\Bottom$ is handled very similarly, and is thus omitted.
 Hence, we can
 sum up the above discussion as follows:
 
\begin{claim}\label{claim:received-info}
If the verifier accepts then each node $v$ receives $\hInfo(F_j(v))$, for every level $j\in J(v)$
(in the time stated in Lemma \ref{lem:time-sync}), and conversely, if a node does not receive $\hInfo(F_j(v))$
(in the time stated in Lemma \ref{lem:time-sync}) then the verifier has rejected.
\end{claim}
 Let $p(v)$ denote the parent of $v$ in $T$.  Recall, that by comparing the data structure of a neighbour $u$ in $T$, node $v$ can know whether
 $u$ and $v$ belong to the same fragment of level $j$, for each $j$. In particular, this is true for $u$ being the parent of $v$ in $T$.
 Consider an event $\cE(v,p(v),j)$. In case $p(v)$ belongs to the same level $j$ fragment as $v$, node $v$
    compares $\hInfo(F_j(v))$ with $\hInfo(F_j(p(v)))$, and verifies
 that  these pieces are equal (otherwise, it rejects). By transitivity, if no node rejects, it follows  that
for every fragment $F\in\cH$, we have that $\hInfo(F)$  is of the form $z\circ j\circ\omega$, and all nodes in $F$ agree on this.
By verifying at the root $r_F$ of $F$ that $\id(r_F)=z$, we obtain the following.
\begin{claim}
If the verifier accepts
then:
\begin{itemize}
\item
 The  claimed identifiers of the  fragments  are compatible with the given hierarchy $\cH$.  In particular,  this guarantees that the identifiers of fragments are indeed unique.
\item
For every $F\in \cH$, all the nodes in $F$ agree on the {\em claimed}  weight
of the minimum outgoing edge of fragment $F$, denoted
 $\hat{\omega}(F)$, and on the  identifier of fragment $F$, namely,
 $\id(F)$.
 \end{itemize}
 \end{claim}
So far, we have shown that each node does receive the necessary information needed for the verifier. Now,
finally, we show how to use this information to detect whether this is an MST. Basically, we
 verify that $\hat{\omega}(F)$ is indeed the minimum outgoing edge ${\omega}(F)$ of $F$  and that this minimum is indeed the candidate edge of $F$, for every $F\in\cH$.
 Consider a time when $\cE(v,u,j)$ occurs. Node $v$ rejects  if any of the checks below is not valid.
\begin{itemize}
\item
{\bf C1:} If $v$ is the endpoint of the  candidate edge $e=(v,u)$  of $F_j(v)$   then $v$ checks that $u$ does not belong to $F_j(v)$, i.e., that ${\id}(F_j(v))\neq {\id}(F_j(u))$,   and that
$\hat{\omega}(F_j(v))=\omega(e)$ (recall, it is already ensured  that $v$ knows whether it is an endpoint, and if so, which of its edges is the candidate);
\item
{\bf C2:} If
${\id}(F_j(v))\neq {\id}(F_j(u))$
then
 $v$ verifies  that $\hat{\omega}(F_j(v))\leq \omega((v,u))$.
\end{itemize}

The following
lemma
now follows from C1, C2 and Lemma~\ref{lem:construction}.

\begin{lemma}
\label{lem:self-stab-verifier}
\begin{itemize}
\item
If
 by some time $t$,
 the events  $\cE(v,u,j)$ occurred for each node $v$ and each neighbour $u$ of $v$ in $G$ and for each level~$j$, and the verifier did not reject, then $T$ is an MST of $G$.
 \item
 If $~T$ is not an MST, then in the time stated in Lemma~\ref{lem:time-sync}  after the faults cease, the verifier rejects.
 \end{itemize}
\end{lemma}
We are now ready for the following theorem, summarizing Sections \ref{sec:mst-construction} to \ref{sec:verifying}.

\begin{theorem}
\label{thm:verification-properties}
The  scheme described in Sections \ref{sec:mst-construction}--\ref{sec:verifying} is a correct proof labeling scheme for MST. Its memory complexity is $O(\log n)$ bits. Its detection time
complexity is $O(\log^2 n)$ in synchronous networks and  $O(\Delta\log^3 n)$ in asynchronous ones.  Its detection distance is $O(f\log n)$ if $f$ faults occurred. Its construction time is $O(n)$.
\end{theorem}
\begin{proof}
The correctness and the specified detection time complexity follow from Lemma~\ref{lem:self-stab-verifier} and Claim~\ref{claim:received-info}.
The space taken by pieces of $\Info$ stored permanently at nodes (and rotated by the trains) was already shown to be
$O(\log n)$ bits.
 In addition,
 a node needs some additional $O(\log n)$ bits of memory for the actions described in Section \ref{sub:utilizing}. Similarly,
 the data-structure at each node and the corresponding $1$-proof labeling schemes (that are used to verify it) consume additional $O(\log n)$ bits.
 Finally, for each train, a node needs a constant number of counters and variables, each of logarithmic size. This establishes the required memory size of the scheme.

 To show the detection distance, let network $G_1$ contain faults.  Consider a (not necessarily connected) subgraph $U$ containing every faulty node $v$, every neighbour $u$ of $v$, and the
 parts, both of $\Bottom$ and of $\Top$ of $v$ and $u$.
First, we claim that no node outside of $U$ will raise an alarm. To see that, assume (by way of contradiction) that some node $w$
  outside $U$ does raise an alarm.
  Now, consider
  a different network $G_2$ with the same sets of nodes and of edges as $G_1$. The state of every node in $G_2\setminus U$
     is exactly the (correct) state of the same node in $G_1$. The states of the nodes
        in $U$ are chosen so that to complete the global configuration to be correct. (Clearly, the configuration can be completed in such a way.)
  Hence, no node should raise an alarm (since we have shown that our scheme is correct). However, node $w$ in $G_2$ receives exactly the same information it receives in $G_1$, since it receives only information from
nodes in the parts to which it or its neighbours belong. Hence, $w$ will raise an alarm. A contradiction.
The detection distance complexity now follows from the fact that the radius of $U$ is $O(f \log n)$.
(Informally, this proof also says that non-faulty nodes outside of $U$ are not contaminated by the faulty nodes, since the
verification algorithm sends information about the faulty nodes only within $U$.)

The construction time complexity required for the more complex part of the proof labeling scheme, that is, the proof  scheme described in Sections  \ref{sec:proof}--\ref{sec:verifying}, is  dominated by the construction time of the MST algorithm $\ALG$. This time is shown to be $O(n)$ in Theorem~\ref{thm:ALG}. The construction time required for the simpler 1-proof labeling scheme described in Section~\ref{sec:section5} is shown to be linear in Lemma \ref{lem:simple-proof}.
\end{proof}


\section{Verification of time lower bound}
\label{sec:lower}
We now show that any proof labeling scheme for MST that uses optimal memory must use at least logarithmic time complexity, even when restricted to  synchronous networks.
The lower bound is derived below from the relatively complex lower bound for 1-proof labeling schemes for MST presented in  \cite{KormanKutten07}, by a not- too- difficult reduction from that problem.
We prove the lower bound on the specific kind of networks used in  \cite{KormanKutten07}.
These networks are a family of weighted graphs termed
$(h,\mu)${\em-hypertrees}. (The name may be misleading; a $(h,\mu)$-hypertrees
is neither a tree nor a part of a hyper graph; the name comes from them being a combination of
$(h,\mu)$-trees (see also \cite{GPPR01,KKKP04}) and  hypercubes.)
We do no describe these hypertrees here, since we use them, basically,
 as black boxes. That is, all we need here is to know certain properties (stated below) of these graphs. (We also need to know the lower bound of \cite{KormanKutten07}).

The following two properties of this family  were observed in
\cite{KormanKutten07}. First, all
$(h,\mu)$-hypertrees are identical if one considers them as
unweighted. In particular, two homologous nodes in any two
$(h,\mu)$-hypertrees  have the same identities. Moreover,
the components assigned to two homologous graphs in  \cite{KormanKutten07} are the same. Hence,
 the
(unweighted) subgraphs $H(G)$ induced by the components of any two
$(h,\mu)$-hypertrees are the same.
The second property that was observed is that this subgraph $H(G)$ is in
fact a (rooted) spanning tree of $G$, the corresponding
$(h,\mu)$-hypertree. Another easy observation that can be obtained by following the recursive
construction of an $(h,\mu)$-hypertree
(see Section 4 of \cite{KormanKutten07}), is that each node in an
$(h,\mu)$-hypertree $G$ is adjacent
to at most one edge which is not in the tree $H(G)$, and that the root of
$H(G)$ is adjacent only to edges in $H(G)$.

Fix an integer $\tau$.
Given a $(h,\mu)$-hypertree $G$, we transform $G$ into a new graph $G'$
according to the following procedure (see Figures \ref{fig:lower1}  and \ref{fig:lower2}).
We replace every edge $(u,v)$ in $G$ where $\id(u)<\id(v)$ with a simple path
$P(u,v)$
containing $2\tau+2$ consecutive nodes, i.e.,  $P(u,v)=(x_1,x_2\cdots,
x_{2\tau+2})$,
where $x_1=u$, and $x_{2\tau+2}=v$. For $i=2,\cdots,2\tau+1$, the port
number at $x_i$ of the port leading to $x_{i-1}$ (respectively, $x_{i+1})$
is
1 (resp., 2).
The port-number at $x_1$ (respectively, $x_{2\tau+2})$ of the
port leading to $x_2$ (resp., $x_{2\tau+1})$ is the same as the port-number
of the port leading from $u$ (resp., $v)$ to $v$ (resp., $u)$ in $G$.
 The weight of the edge $(x_{2\tau+1},x_{2\tau+2})$ is  the weight of $(u,v)$, that is, $\omega(x_{2\tau+1},x_{2\tau+2})=\omega(u,v)$, and the
weight of  all other edges in $P(u,v)$ is 1.
The identities of the nodes in the resulted graph are given according to a
DFS traversal on $G'$. We now describe the component
of each node in $G'$.
Let $(u,v)$ be an edge in $G$ and let $P(u,v)=(x_1,x_2\cdots, x_{2\tau+2})$
be the corresponding path in $G'$, where
$x_1=u$ and $x_{2\tau+2}=v$ (here we do not assume necessarily that
$\id(u)<\id(v))$. Consider
 first the case that in the graph $G$, the edge $(u,v)$ belongs to the tree
$H(G)$. Assume without loss of generality that the component of $u$ in $G$
points at $v$. For each $i=1,2,\cdots,2\tau+1$, we let the component at
$x_i$ point at $x_{i+1}$ (the component at $x_{2\tau+2}$ is the same as the
component of $v$ in $G)$.
Consider now the case that $(u,v)$ does not belong to $H(G)$.
In this case, for $i=2,3,\cdots,\tau+1$, we let the component at $x_i$ point
at $x_{i-1}$ (the component at $x_1=u$ in $G'$ is the same as the component of $u$ in
$G)$ and
for $i=\tau+2,\tau+3,\cdots,2\tau+1$, we let the component at $x_i$ point at
$x_{i+1}$ (similarly, the component at $x_{2\tau+2}$
is the same as the component of $v$ in~$G)$.

\begin{figure}
\centering
\includegraphics[scale=.4]{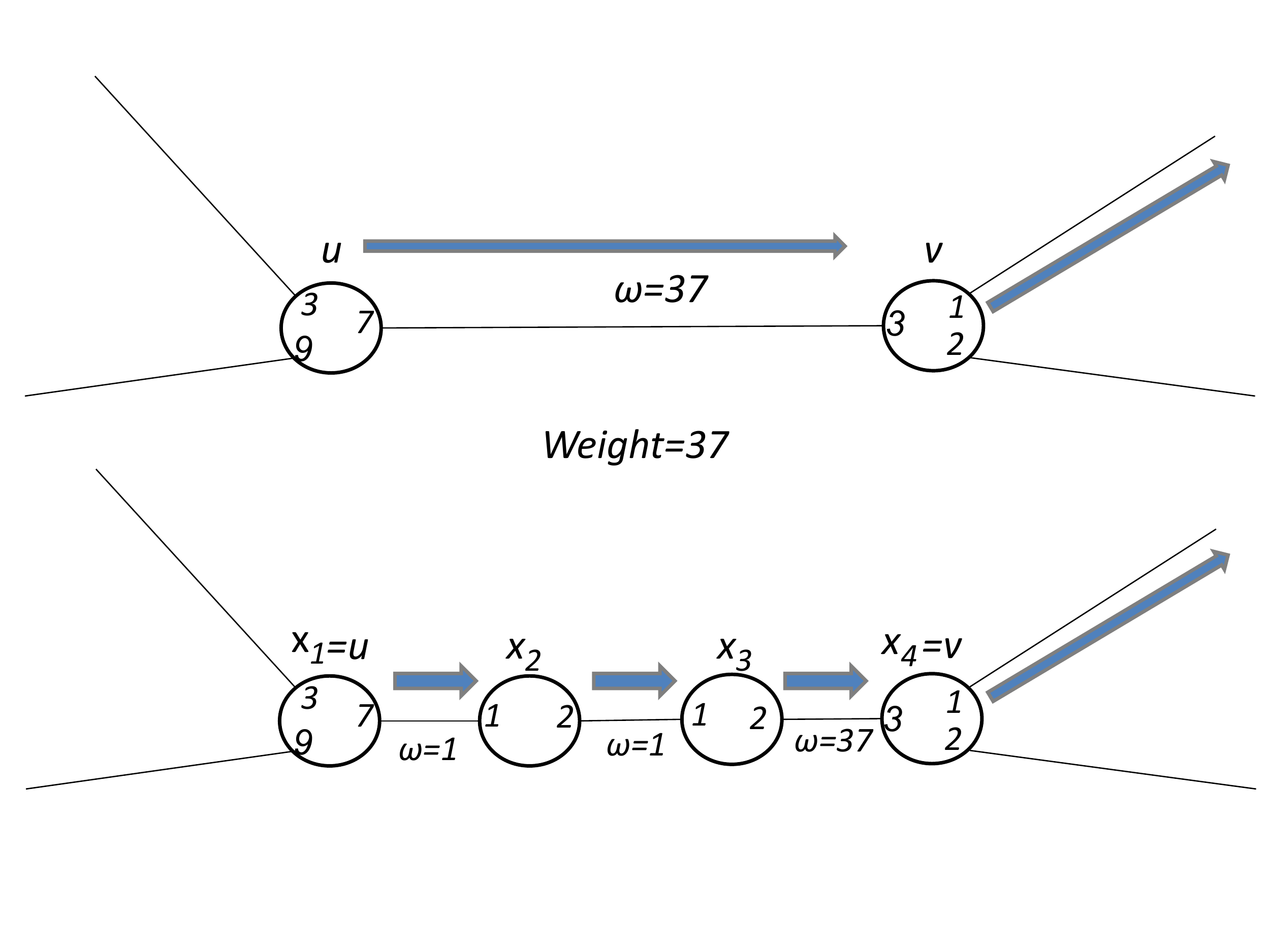}
\caption{Transforming an edge of $G$ (the upper part) to a path of $G'$ (the lower part) for $\tau=1$ and the case that the component of $u$ points at $v$. The component of $v$ points at $v$'s port 1.}\label{fig:lower1}
\end{figure}

\begin{figure}
\centering
\includegraphics[scale=.4]{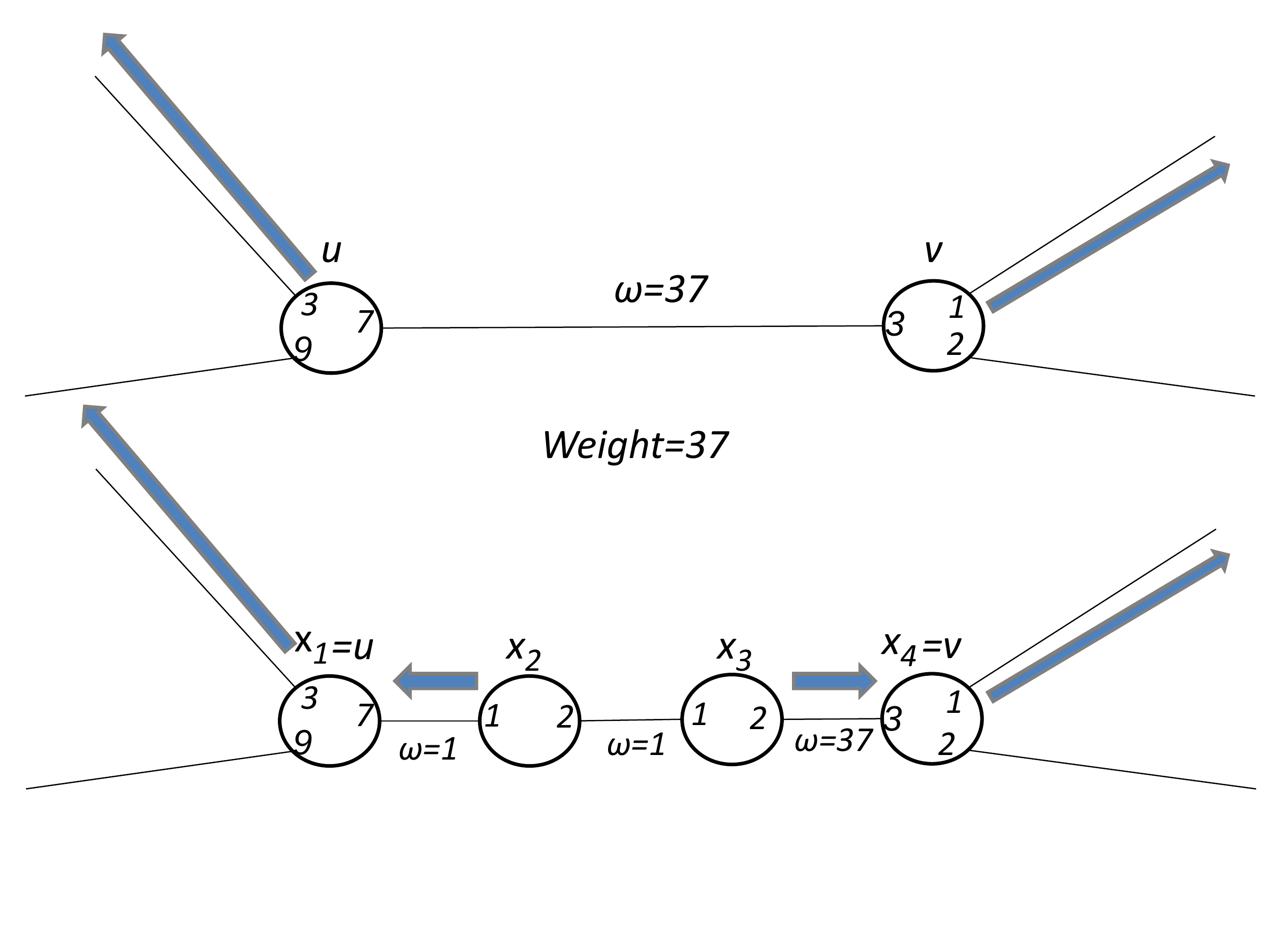}
\caption{The case that the component of $u$ points at $u$'s port 3 that does not lead to $v$, and the component of $v$ points at $v$'s port 1 that does not lead to $u$.}\label{fig:lower2}
\end{figure}

By this construction of $G'$, we get that the subgraph $H(G')$ induced by
the components of $G'$ is a spanning tree of $G'$, and it is an MST
of $G'$ if and only if $H(G)$ is an MST of $G$.
Let $\cF(h,\mu,\tau)$ be the family of all weighted graphs $G'$ obtained by
transforming every $(h,\mu)$-hypertree $G$ into $G'$ using the method
explained above.

\begin{lemma}
If there exists a proof labeling scheme for MST on the family
$\cF(h,\mu,\tau)$ with memory complexity $\ell$ and detection time $\tau$
then there exists a $1$-proof labeling scheme (a proof labeling scheme in the sense of
\cite{KormanKutten07}) for the MST predicate on the family
of  $(h,\mu)$-hypertrees with label size $O(\tau\ell)$.
\end{lemma}
\begin{proof}
Let $(\cM',\cV')$ be a proof labeling scheme for MST and the family
$\cF(h,\mu,\tau)$  with memory complexity $\ell$ and detection time $\tau$.
We describe now a 1-proof labeling scheme $(\cM,\cV)$ for the MST predicate on
the family
of  $(h,\mu)$-hypertrees.
Let $G$ be an $(h,\mu)$-hypertree that satisfies the MST predicate. We first
describe the labels assigned by the marker $\cM$
to the nodes on $G$.
In this lemma, we are not concerned with the time needed for actually assigning the labels using a distributed algorithm, hence, we describe the marker $\cM$ as a
centralized algorithm and not as a distributed one.   (We note that this is consistent with the model of \cite{KormanKutten07} that considers only centralized marker algorithms.)

The marker $\cM$ transforms $G$ to $G'$. Observe that $G'$ must also satisfy
the MST predicate.
$\cM$ labels the nodes of $G'$ using the marker $\cM'$. Note that any label
given by the marker $\cM'$ uses at most $\ell$ bits. Given a node $u\in G$,
let $e_1(u)$ be the edge not in the tree $H(G)$ that is adjacent to $u$
(if one exists) and let $e_2(u)$ be the
edge in $H(G)$ leading from $u$ to its parent in $H(G)$ (if one exists). Let
$P_1(u)=(w_1,w_2\cdots, w_{2\tau+2})$ be the path in $G'$ corresponding
to $e_1(u)$.
If $u$ is not the root of $H(G)$ then $e_2(u)$ exists and let
$P_2(u)=(y_1,y_2\cdots, y_{2\tau+2})$
be the path in $G'$ corresponding
to $e_2(u)$, where $y_1=u$ and $y_{2\tau+2}$ is the parent of $u$. For the
root $r$ of $H(G)$, let $P_2(r)$ be simply $(r)$.  If $u$ is not the root
of $H(G)$ then for each $i\in\{1,2,\cdots, 2\tau+1\}$, the marker $\cM$
copies the labels $\cM'(y_i)$ and
$\cM'(w_i)$ into the $i$'th field in the label $\cM(u)$. (Note that the
labels $\cM'(w_i)$ are copied in the labels given by $\cM$ to
both end-nodes of $e_1(u)$.) If $r$ is the root of $H(G)$ then $e_2(r)$ does
not exist and actually, also $e_1(r)$ does not exist,
as $r$ is not adjacent to any edge not in $H(G)$.
The marker $\cM$ simply copies the label $\cM'(r')$ into the label $\cM(r)$,
where
$r'$ is the corresponding node of $r$ in $G'$.

In the model of proof labeling schemes in \cite{KormanKutten07}, the
verifier $\cV$ at a node $u\in G$  can look at the labels of all nodes $v$
such that $(u,v)$ is an edge of $G$. In particular,
it sees the labels assigned by $\cM'$ to all nodes in $G'$ at distance at
most $2\tau$ from its corresponding node $u'$ in $G'$.
Let $B_\tau(u')$ be the set of nodes at distance at most $\tau$ from $u'$ in
$G'$.
 We let the verifier $\cV$ at
$u$
simulate the operations of the verifier $\cV'$ at each node in
$B_\tau(u')$--this can be achieved as
the information in the $1$-neighbourhood of $u$ (in $G)$ contains the
information in the $\tau$-neighbourhood of $G'$ of any node in $B_\tau(u')$.
Finally, we let $\cV(u)=1$ if and only if
 $\cV'(x)=1$ for all $x\in B_\tau(u')$.
It can be easily observed that $(\cM,\cV)$ is indeed a $1$-proof labeling
scheme for the family of  $(h,\mu)$-hypertrees.
Moreover, each label assigned by the marker $\cM$ uses $O(\tau\ell)$ bits.
(Note that the model in \cite{KormanKutten07} restricts only the sizes of
the labels and not the memory size used by the verifier.)
This completes the proof of the lemma.
\end{proof}

\begin{corollary}
Fix a positive integer $\tau=O(\log n)$. The memory complexity of any proof
labeling scheme for $\cF(n)$ with detection time $\tau$
is $\Omega(\frac{\log^2 n}{\tau})= \Omega(\log n)$.
(Recall $\cF(n)$ represents all connected undirected weighted graphs.)
\end{corollary}
\begin{proof}
In \cite{KormanKutten07} we showed that  the label size of any proof
labeling scheme for the MST predicate and the family of
$(\log n,n)$-hypertrees is $\Omega(\log^2 n)$ bits.
The claim now follows by combining the previous lemma together with the fact
that
the number of nodes in a graph $G'\in\cF(\log n,n,\tau)$ is polynomial in
$n$.
\end{proof}

\section{The self-stabilizing MST construction algorithm}
\label{sec:full-mst-alg}
\label{scheme}

We use
a transformer that inputs a non-self-stabilizing algorithm and outputs a self-stabilizing one.
For simplicity, we first explain how to use the transformer proposed in the seminal paper of Awerbuch and Varghese \cite{awerbuch-varghese}
 (which utilizes the transformer of its companion paper
\cite{APV} as a black box).
This already yields a self-stabilizing MST algorithm with $O(n)$ time and $O(\log n)$
memory per node.
Later, we refine that transformer somewhat to add the property that the
verification time is of $O(\log^2 n)$
in a synchronous network, or $O(\min\{ \Delta\log^3 n ,n\})$ in an asynchronous
one.
We then also establish the property that if $f$ faults occur, then each fault is
detected within its $O(f \log n)$ neighbourhood.

The Resynchronizer of \cite{awerbuch-varghese}  inputs a non-stabilizing synchronous input/output algorithm\footnote{An input/output algorithm is one whose correctness requirement can be specified as a relation between its input and its output.}
$\Pi$
whose running time and memory size  are some $T_\Pi$ and $S_\Pi$, respectively.
Another input it gets is $\hat{D}$, which is an {\em upper bound} on the actual diameter $D$ of the network.
It then yields a self-stabilizing version whose memory size is $O(S_\Pi  + \log n)$ and whose
time complexity is $O(T_\Pi+\hat{D})$.

For our purposes, to have the Resynchronizer yield our desired result, we first need
to come up with such
a bound $\hat{D}$ on the diameter.
 (Recall that we do {\em not}  assume that $D$, or even $n$, are known).
 Second, the result of the Resynchronizer of \cite{awerbuch-varghese} is a synchronous algorithm, while we want an algorithm that can be executed in an asynchronous network.
Let us describe how we bridge these two gaps.

We use known self-stabilizing protocols \cite{afek-bremler,datta2008}
 to compute~$D$, the diameter of the network, in time $O(n)$, using $O(\log n)$ bits of memory per node. We use this computed $D$ as the desired $\hat{D}$. Note that
at the time that~\cite{awerbuch-varghese} was written, the only algorithm for computing a good bound (of $n)$ on the diameter with a bounded memory had time complexity
$\Theta(n^2)$ \cite{AKY}.

To bridge the second gap, of converting the resulting self-stabilizing algorithm for an {\em asynchronous} network, we use a {\em self-stabilizing synchronizer}
that
transforms algorithms designed for synchronous networks to function
correctly in  asynchronous ones. Such a synchronizer was
not known at the time that~\cite{awerbuch-varghese} was written, but several are available
now. The synchronizer
 of \cite{awerbuch-kutten-patt-et-al,dependable-systems} was first described as if it needs unbounded memory. However, as is stated in
 \cite{awerbuch-kutten-patt-et-al}, this synchronizer is meant to be coupled with a reset protocol to bound the memory. That is, to have a memory size of $O(\log n)$ and time
 $O(n)$, it is sufficient to use a reset protocol with these complexities. We use the reset protocol of  \cite{APV}. Similarly, this reset protocol is meant to be coupled with a
 self-stabilizing spanning tree construction algorithm. The complexities of the resulting reset protocol
  are dominated by those of the spanning tree construction. We plug in some spanning tree algorithm with the desired properties (such as  \cite{afek-bremler,datta2008}) whose memory size and time complexities are
  the desired $O(\log n)$ and $O(n)$ in asynchronous networks, respectively.
  (It is easy to improve the time to $O(D)$ in synchronous networks.)
 This  yields the desired reset protocol, and, hence, the desired synchronizer protocol\footnote{An alternative synchronizer can be based on the one of~\cite{graph-theory-helps}, again, coupled  with some additional known components, such as a module to compute $n$.}.

Let us sum up the treatment of the first two gaps: thanks to some new modules developed after \cite{awerbuch-varghese}, one can now use the following version of the main result of~\cite{awerbuch-varghese}.

\begin{theorem}\label{EAV}
{\bf Enhanced Awerbuch-Varghese Theorem,  (EAV):}
Assume we are given a distributed algorithm~$\Pi$ to compute an input/output relation.
Whether $\Pi$ is
synchronous or asynchronous,
let $T_\Pi$ and
 $S_\Pi$ denote $\Pi$'s time complexity and  memory size, respectively, when executed in synchronous networks.
The enhanced Resynchronizer compiler produces an asynchronous (respectively, synchronous) self-stabilizing algorithm whose memory size is $O( S_\Pi+\log n) $ and whose time complexity
is $O(T_\Pi+ n)$ (resp., $O(T_\Pi+ D))$.
\end{theorem}
The EAV theorem differs from the result in \cite{awerbuch-varghese} by (1) addressing also asynchronous algorithms, and (2) basing the time complexity on the actual values of $n$ and $D$ of the network rather than on
an a-priori bound $\hat {D}$ that may be arbitrarily larger than $D$ or $n$.

 Recall from Theorem \ref{thm:ALG} that  in synchronous networks, algorithm $\ALG$ constructs an MST in $O(n)$ time and using $O(\log n)$ memory bits per node. Hence, plugging in algorithm $\ALG$ as $\Pi$ yields the following theorem.

 \begin{theorem}
There exists a self-stabilizing MST construction algorithm that can operate in an asynchronous environment, runs in $O(n)$ time and uses $O(\log n)$ bits of memory per node.
 \end{theorem}

\subsection{Obtaining fast verification}

The Resynchronizer compiler performs iterations forever. Essentially, the first iteration is used to compute the result of $\Pi$, by executing $\Pi$
plus some additional components   needed for the self-stabilization. Each of the later iterations is used to check that the above result is correct. For that, the Resynchronizer executes a {\em checker}. If the result is not correct, then the checker in at least one node ``raises an alarm''. This, in effect, signals the Resynchronizer to drop back to the first iteration. Let us term such a node a {\em detecting node}. Our refinement just replaces the checker, interfacing with the original Resynchronizer by supplying such a detecting node.

We should mention that the original design in \cite{awerbuch-varghese} is already modular in allowing such a replacement of a checker. In fact,
two alternative such checkers are
studied in \cite{awerbuch-varghese}. The first kind of a checker is $\Pi$ itself. That is, if $\Pi$ is deterministic, then, if executed again, it must compute the same result again
(this is adjusted later in \cite{awerbuch-varghese} to accommodate randomized protocols).
This checker functions by comparing the result computed by $\Pi$ in each ``non-first'' iteration to the result it has computed before. If they differ, then a fault is detected.
  The second kind of a checker is a local checker of the kinds studied in
\cite{AKY,APV} or even one that can be derived from local proofs \cite{KormanKutten07,KKP10}. That is, a checker whose time complexity is exactly~1. When using this kind of a checker, the Resynchronizer uses one iteration to
execute~$\Pi$, then the Resynchronizer executes the checker repeatedly until a fault is detected.
It was argued in \cite{awerbuch-varghese} that the usage of such a checker (of time complexity exactly $1$) is easy, since such a checker self-stabilizes trivially.
We stress  that it was later shown that such a checker (whose time complexity is 1) must use $\Omega(\log^2 n)$ bits~\cite{KKP10}.
Hence, plugging such a checker into the Resynchronizer compiler cannot yield an optimal memory self-stabilizing algorithm.

The door was left open in \cite{awerbuch-varghese} for additional checkers.
It was in this context that they posed the open problem of whether MST
has a checker which is faster than MST computation, and still uses small memory. (Recall that Theorem \ref{thm:verification-properties} answers the open problem in the affirmative.)

We use a self-stabilizing verifier (of a proof labeling scheme) as a checker. That is, if a fault occurs, then the checker detects it, at least in one node, regardless of the initial configuration.
Such nodes where the fault is detected serve as the detecting nodes used above by the Resynchrnonizer. The following theorem differs from the EAV theorem by stating that the final protocol (resulting from the transformation) also enjoys the good properties of the self-stabilizing verifier. I.e., if the self-stabilizing verifier has a good detection time and good detection distance, then, the detection time and distance of the resulting protocols are good too.

\begin{theorem}\label{thm:self-stab}
Suppose we are
given the following:
\begin{itemize}
\item
A distributed algorithm $\Pi$ to compute an input/output relation~$R$. Whether $\Pi$
is
synchronous or asynchronous,
let $T_\Pi$ and
 $S_\Pi$ denote $\Pi$'s time complexity and  memory size, when executed in {\em synchronous} networks.
\item
An asynchronous (respectively, synchronous)  proof labeling scheme $\Pi'$ for verifying $R$
 with memory size~$S_{\Pi'}$,
whose verifier self-stabilizes
with verification time and detection distance $t_{\Pi'}$~and~$d_{\Pi'}$, and whose
  construction time (of the marker) is $T_{\Pi'}$.
\end{itemize}
Then, the enhanced Resynchronizer  produces an asynchronous  (resp., synchronous) self-stabilizing algorithm whose memory and time complexities are $O( S_\Pi+ S_{\Pi'} + \log n) $ and  $O(T_\Pi+T_{\Pi'}+{ t_{\Pi'}} +n)$ (resp.,  $O(T_\Pi+T_{\Pi'}+{ t_{\Pi'}} +D))$,  and whose
  verification time and detection distance are $t_{\Pi'}$ and $d_{\Pi'}$.
\end{theorem}

\begin{proof}
The proof relies heavily on the Resynchronizer compiler given by the EAV theorem (Theorem~\ref{EAV}). This Resynchronizer  receives as input  the following algorithm $\Pi''$, which is not assumed to be neither self-stabilizing nor asynchronous.  Specifically,  algorithm $\Pi''$ first constructs the relation $R$ using algorithm $\Pi$ and, subsequently, executes the marker algorithm of the proof labeling scheme $\Pi'$.

The resulted Resynchronizer (when executing together with the algorithm $\Pi''$ it transforms) is a detection based self-stabilizing algorithm
 (see the explanation of the detection time and distance in Section \ref{sub:detection-time}). It executes algorithm $\Pi''$ for a set amount of time
 (here, counting the time using the self-stabilizing synchronizer) and then puts all the nodes in an {\em output} state, where it uses the self-stabilizing verifier of the proof labeling scheme $\Pi'$ to check. (Recall, in contrast to the marker algorithm, the verifier algorithm of $\Pi'$ is assumed to be self-stabilizing.) The detection time and the detection distance of the combined algorithm thus follow directly from the detection time and the detection distance of the proof labeling scheme $\Pi'$. This concludes the proof of the theorem.
\end{proof}

Now, as algorithm $\Pi$, we can plugged in Theorem \ref{thm:self-stab} the MST construction algorithm $\ALG$, that uses optimal memory size and runs in $O(n)$ time. Furthermore, two possible proof labeling schemes that can be plugged in Theorem \ref{thm:self-stab} as $\Pi'$ are the schemes of \cite{KormanKutten07,KKP10}. Both these schemes use $O(\log^2 n)$ memory size. Since their detection time is 1, they stabilize trivially. The corresponding distributed markers are simplified versions of the marker of the proof labeling scheme given of the current paper, and hence their construction time  is $O(n)$. Hence, plugging either one of these schemes as $\Pi'$
 yields the following.

 \begin{corollary}
\label{cor:weak-MST}
There exists a self-stabilizing MST algorithm with  $O(\log^2 n)$ memory size and $O(n)$ time. Moreover, its detection time
is $1$ and its detection distance is $f+1$.
\end{corollary}
 Finally, by  plugging  to  the Resynchronizer given in Theorem \ref{thm:self-stab},
the construction algorithm $\ALG$ as~$\Pi$ and our optimal memory proof labeling scheme
 mentioned in Theorem \ref{thm:verification-properties} as $\Pi'$, we obtain
 the following.

\begin{theorem}
\label{cor:verification-properties}
There exists a self-stabilizing MST algorithm that uses optimal $O(\log n)$ memory size and $O(n)$ time.
Moreover,  its detection time
complexity is $O(\log^2 n)$ in synchronous networks and  $O(\Delta\log^3 n)$ in asynchronous ones.  Furthermore, its detection distance is $O(f\log n)$.
\end{theorem}

\subsection{Combining self stabilzing algorithms}

The algorithm in this paper is composed of multiple modules (Figures \ref{fig:Structure1} and \ref{fig:Structure2}).
Some of them are self stabilizing, and some are not. When composing self stabilizing algorithms together, the result may not be self stabilizing, so one should take care \cite{dim-composition}.
We have claimed  the stabilization of composite programs throughout this paper. For the sake of completeness, let us go over all the components here once again, to recall that their composition self stabilizes in spite of the composition.

\begin{figure}
\centering
\includegraphics[scale=.8]{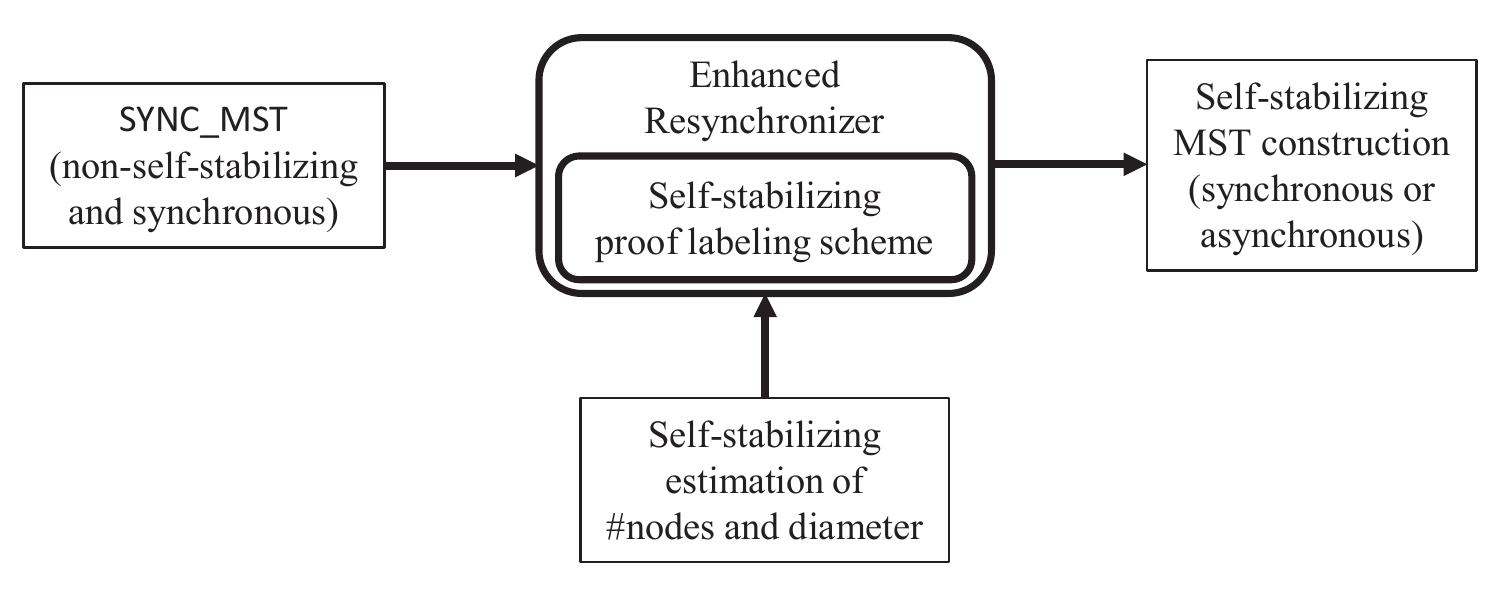}
\caption{Structure of the self-stabilizing (synchronous or asynchronous) MST construction algorithm obtained by the enhanced Resynchronizer.}
\label{fig:Structure1}
\end{figure}

\begin{figure}
\centering
\includegraphics[scale=.8]{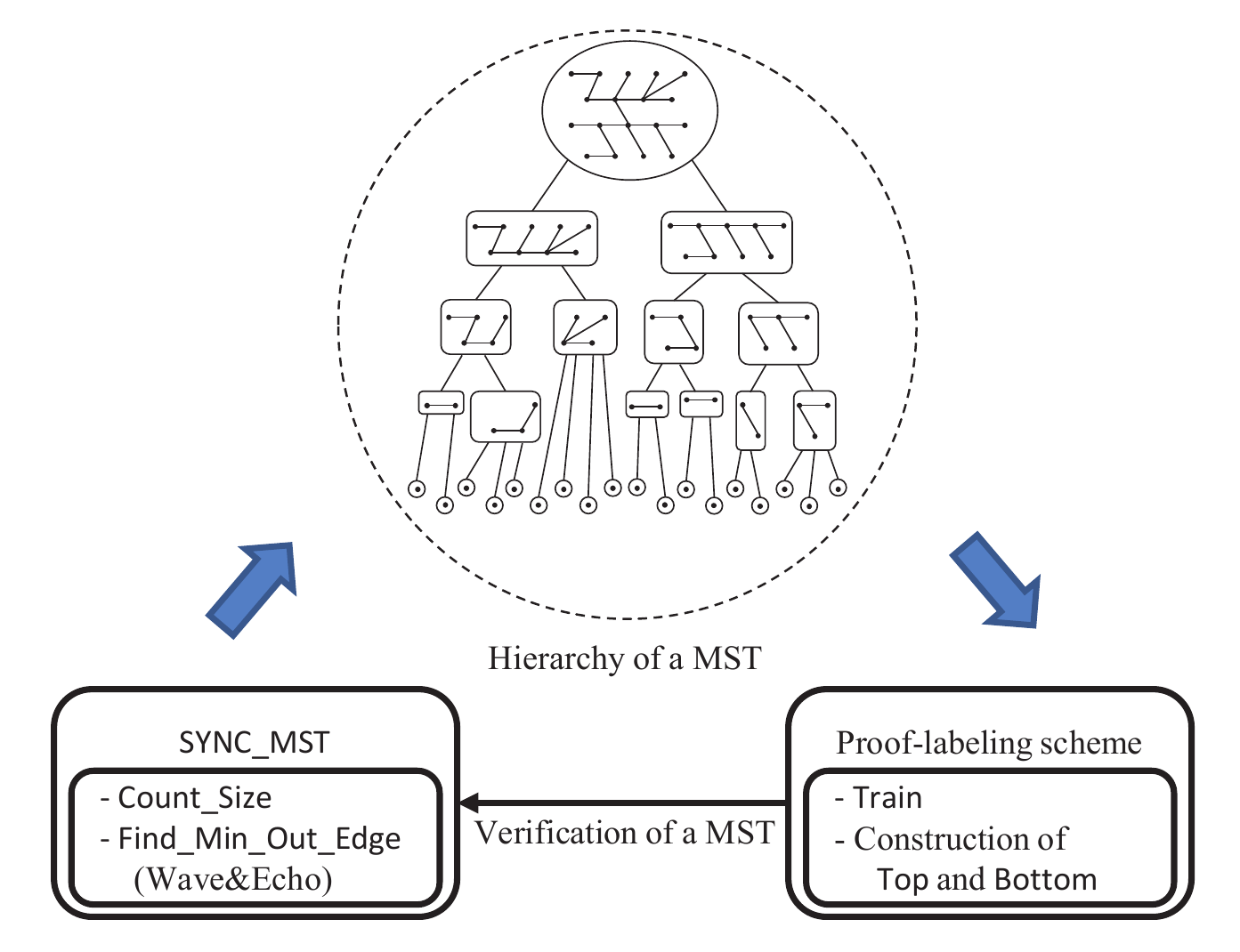}
\caption{Algorithm $\ALG$, mainly consisting of algorithms $\COUNTSIZE$ and $\FINDOUT$, induces the hierarchy of a MST.  From the hierarchy the proof-labeling-scheme, mainly consisting of the trains and the construction of partitions $\Bottom$ and $\Top$, produces a marker and a verifier.}
\label{fig:Structure2}
\end{figure}

 The main composition is that of the transformer algorithm of Awerbuch and Varghese \cite{awerbuch-varghese} together with a checking scheme. The way to perform this composition, as well as its correctness, have been established in \cite{varghese-thesis} (as well as in \cite{awerbuch-varghese}). See Theorem 6.1 in \cite{awerbuch-varghese} and Theorem 9.3.2 in \cite{varghese-thesis}.

A synchronizer uses, as an input, the number of nodes and the value of the diameter computed by other algorithms. Here the correctness follows easily from the ``fair combination'' principle of \cite{dim-composition,protocol-composition}. That is, the algorithms computing these values do not use inputs from the other algorithms in the composition. Moreover, their outputs stabilize to the correct values at some points (from their respective proofs, that do not depend on assumptions in other algorithms). From that time on, their values are correct.

The tree construction itself is not supposed to be self stabilizing for the transformer scheme of  \cite{awerbuch-varghese}. This is also the case with the marker algorithm, since the MST construction algorithm and the marker together constructs a data structure to be verified. (Recall that verifying the MST alone is costly \cite{KKP11}; hence the idea is to construct a ``redundant'' representation of the MST, containing the MST plus the proof labels, such that verifying this redundant data structure is easier).

It is left to argue that the verifier on its own self stabilizes, in spite of the fact that it is composed of several components. Recall that the output of the verifier is a logical AND of several verifiers. That is, if either the verifier for the scheme for the Well-Forming property (Sections \ref{sec:mst-construction} and  \ref{sec:section5}) or the verifier for the scheme for  the Minimality  property (Sections  \ref{sec:proof},  \ref{sub:3.2}, and \ref{sub:utilizing}) outputs ``no'', then the combined (composed together) verifier outputs no. Hence, the different schemes do not interfere with each other. If all of them are self stabilizing, then the composition is self stabilizing. In particular, the scheme for verifying the Well-Forming property runs in one time unit repeatedly. As observed by \cite{awerbuch-varghese}, such a verifier is necessarily self stabilizing. It is then left to show that the verifier for  the Minimality property self stabilizes.

Note that Section \ref{sec:proof} describes a part of the marker, devoted to the scheme for verifying  the Minimality property. Recall that the marker is not required to self stabilize.
Section \ref{sub:3.2} describes the trains process which is composed  of two parts: the convergecast of the information to a part's root, and its broadcast from the root. The second process (the broadcast) inputs (at the root) the results of the first process, but not vice versa. Hence, clearly, the composition self stabilizes as above (that is, after the first process eventually stabilizes, the second process will eventually stabilize too). The pieces of information carried by the train are then used by each node to compare information with its neighbours (in Section \ref{sub:Sampling and synchronizing}) and by the part root
(in Section  \ref{sec:section5}).  Again, the flow of information between modules is one way. That is, from the train process to the computations by each node and by the root. After the trains stabilizes, so does the rest, eventually. (The later computations also input the output of the module computing the number $n$ of nodes in the network; again, the flow of information is only unidirectional, and hence the verifier does stabilize after the $n$ computation stabilizes).

\begin{Comment}
{\bf Using later synchronizers:}
As explained in Section \ref{sec:full-mst-alg}, for simplicity of the presentation we prefer using the synchronizer and the reset protocols built in the scheme of
\cite{awerbuch-varghese}, since the proof of their composition is already covered in \cite{awerbuch-varghese, varghese-thesis}. For those who prefer using the later synchronizers and reset protocols we mentioned, e.g., \cite{dependable-systems}, the composition would remain self stabilizing even if we use those. The correction of this statement has essentially been established in those synchronizers papers. That is, they presented synchronizers such that they can take any algorithm intended to run over a synchronous network, compose with it, and have it run correctly (and in a self-stabilizing manner) in an asynchronous network. The same holds also for self stabilizing reset protocols.

For the sake of completeness, let us recall, nevertheless, why this composition is correct.
For the synchronizer to work, it needs a certain output from the algorithm. This output is TRIVIAL. That is, a SYNCHRONOUS  algorithm at a node at a pulse acts as follows. It receives messages from ALL the neighbours (or at least a statement that no message is going to arrive from a specific neighbour), and then processes a message from each neighbour. Then it is ready for the next pulse.

Thus, the synchronizer needs to know (1) when did the algorithm receive messages from all the neighbours. For this purpose, the synchronizer receives the messages on the algorithm's behalf, and when it receives all of them (or notifications that no messages will be sent), it passes all of them to the algorithm together, which, in turn,  processes all of these messages together. The algorithm needs then to tell the synchronizer that it has finished processing the messages. If this processing generates messages to be sent to neighbours, the algorithm needs to give these new messages to the synchronizer to send them on the algorithm's behalf. (This is done so that if there is a neighbour $u$ to which the node does not send a message in the current pulse, the synchronizer will send a ``dummy'' message, saying that no message will arrive.)

The analysis of the synchronizers (in the papers that presented self stabilizing synchronizers, e.g., \cite{dependable-systems}) were base on the rather obvious observation regarding the correctness of this trivial information for any ``reasonable'' algorithm, starting from the second round. That is, it is {\em not} assumed that the computation or the messages are correct. What is assumed by the synchronizers is  just the fact that the algorithm computed already the messages from the previous round (and is giving the synchronzer the resulting messages).  Obviously, this assumption holds for our algorithm too, so we can rely on the results of the papers where the synchronizers were designed.
\end{Comment}

\end{document}